\documentclass[12pt,thmsa, a4paper]{article}%
\usepackage{amssymb}
\usepackage{amsfonts}
\usepackage{amsmath}
\usepackage{xcolor}
\usepackage{graphicx}
\usepackage{geometry}
\usepackage{probsoln}
\usepackage{enumerate}
\usepackage{natbib}%
\providecommand{\U}[1]{\protect\rule{.1in}{.1in}}
\def\E{\mathbb{E}}
\def\R{\mathbb{R}}

\def\N{\mathbb{N}}
\def\X{\mathcal{X}}
\def\d{\mathrm{d}}
\def\esssup{\operatorname{ess} \sup}
\def\essinf{\operatorname{ess} \inf}
\newtheorem{theorem}{Theorem}

\newtheorem{proposition}[theorem]{Proposition}
\newtheorem{lemma}[theorem]{Lemma}
\newtheorem{definition}{Definition}

\newtheorem{aaxiom}{Axiom}[section]
\newenvironment{proof}[1][Proof.]{\addvspace{\bigskipamount}\noindent \textbf{#1} }{\hfill$\square$\par \addvspace{\bigskipamount}}
\definecolor{indigo}{rgb}{0.0, 0.25, 0.42}
\definecolor{hl}{rgb}{0.0, 0.25, 0.42}

\newenvironment{super-boxer}
{\begin{center}
		\begin{tabular}{||p{0.9\columnwidth}||}
			\hline \hline \\
		}
		{
			\\ \\ \hline \hline
		\end{tabular}
	\end{center}
}

\hideanswers
\showanswers
\makeatletter
\renewcommand{\maketitle}{
	\begin{center}
		{\LARGE \@title \par}
		\vspace{1.5em} 		{\@author \par}
		\vspace{1.5em} 		{\@date \par}
	\end{center}
}
\makeatother
\usepackage{array}
\usepackage{tikz}
\usetikzlibrary{arrows.meta}

\begin{document}

\title{Disappointment Aversion and Expectiles}
\author{
\begin{tabular*}
{\textwidth}[c]{@{\extracolsep{\fill}}ccccc}%
F. Bellini & F. Maccheroni & T. Mao & R. Wang & Q. Wu\\
\emph{Bicocca} & \emph{Bocconi} & \emph{USTC} & \emph{Waterloo} & \emph{USTC}%
\end{tabular*}
}
\date{\textbf{\today}}
\maketitle

\begin{abstract}
\noindent This paper recasts Gul's (1991) theory of disappointment aversion in
a Savage framework, with general outcomes, new explicit axioms of
disappointment aversion, and novel explicit representations. These permit
broader applications of the theory and a better understanding of its
decision-theoretic foundations. Our results exploit an unexpected connection
between Gul's model and the econometric framework of Newey and Powell (1987)
of asymmetric least squares estimation. 
Our main axiomatization result shows that a preference relation over Savage
acts is probabilistically sophisticated, invariant biseparable, and
disappointment hedging if and only if it admits a representation \emph{à la} Gul, and hence all explicit equivalent representations that we present
in the paper. We also derive a neurocomputational foundation of the theory based on recent neuroscience findings and a novel reinforcement learning result. \medskip

 \medskip


\end{abstract}

\begingroup
\renewcommand\thefootnote{}\footnotetext{Some material in this paper was
previously presented in ``Disappointment concordance and duet expectiles''
(\texttt{arXiv:2404.17751}). The present paper subsumes and replaces that
material. We thank Pierpaolo Battigalli, Aur\'elien Baillon, Simone
Cerreia-Vioglio, Soo Hong Chew, David Dillenberger, Gabi Gayer, Peter Klibanoff, Alfonso Maselli, Massimo
Marinacci, Sujoy Mukerji, Peter Wakker, Tianxiong Yu, and especially Marcus Pivato, as well
as the participants at the 2024 RISKLAB workshop in Rabat, the 2024 China
Conference on ASQFRM in Beijing, and the 2025 Science of Decision Making
conference in Chengdu for stimulating discussions and very useful comments.
Fabio Maccheroni gratefully acknowledges the financial support of Grant
101142844 TRAITS-GAMES ERC Europe-2023-ADG. Ruodu Wang is supported by the
Natural Sciences and Engineering Research Council of Canada (CRC-2022-00141,
RGPIN-2024-03728).}
\endgroup

\section{Introduction}

The \emph{theory of disappointment aversion} of Gul (1991) is one of the most
influential contributions to decision-making under risk. This model posits that
individuals evaluate the outcomes of any risky action relative to an
endogenous reference point: the ex-ante value of the action itself. Such a
reference point balances the negative emotional responses to outcomes falling
short of expectations --- that are \emph{disappointing} --- against the positive
emotional responses to outcomes exceeding expectations --- that are
\emph{elating}.

The model captures two core emotions, sadness (from disappointment) and joy
(from elation). It has proven valuable in explaining economic behaviors beyond
the reach of expected utility theory,   {\color{black}has been tested experimentally,} and   aligns closely with recent neuroscientific findings on
reference-dependent decision-making.\textcolor{black}{\footnote{\color{black}For experimental testing see \citet{bleichrodt2007utility}, for a related model of anticipatory utility and disappointment, see
\citet{gollierMuermann2010}, for dynamic asset pricing with generalized
disappointment aversion, see \citet{augustinTedongap2021}.}} Yet, despite its success, the original
framework, including its axiomatic formulation and functional representation,
has some known limitations, outlined below. The present paper addresses these comprehensively,
{while keeping the core theory intact}. \bigskip

Gul's model is characterized by two parameters: a \emph{utility function}
$u:\mathbb{R}\rightarrow\mathbb{R}$, evaluating monetary outcomes, and a
\emph{disappointment aversion coefficient} $\beta\in(-1,\infty)$, measuring
how much a decision maker (DM) discounts elating outcomes relative to
disappointing ones. A positive $\beta>0$ corresponds to disappointment
aversion, $\beta=0$ aligns with standard expected utility preferences, and a
negative $\beta<0$ indicates elation seeking behavior (rather than being
discounted, elating outcomes carry additional value).

Gul (1991) shows that a weakening of the von
Neumann-Morgenstern Independence Axiom (motivated by the Allais paradox)  characterizes evaluation of a random
variable $X$ as the solution $v_{X}$ to the equation:
\begin{equation}
v=\mathbb{E}\left[  k_{v}(X)\right]  \label{eq:Gul}%
\end{equation}
where $k_{v}:\mathbb{R}\rightarrow\mathbb{R}$ is defined as
\[
k_{v}(x)=%
\begin{cases}
\dfrac{u(x)+\beta v}{1+\beta} & \text{if }u(x)\geq v\\
\hspace{16pt}u(x) & \text{if }u(x)<v
\end{cases}
\]
for all $v\in\mathbb{R}$. The DM prefers $X$ over $Y$ if, and only if,
$v_{X}\geq v_{Y}$. \bigskip

Our research is motivated by the following challenges:

\begin{enumerate}
\item The original formulation above is restricted to monetary outcomes, even
though disappointment over non-monetary outcomes, such as the results of
personal endeavors or medical treatments, seems to be equally, if not more, significant.

\item The axiomatic characterization does not include explicit axioms that
capture disappointment aversion or elation seeking,
and, as a result, much of the theory's interpretation depends on its
functional representation.

\item The representation itself is implicit: the evaluation $v_X$ of  $X$ is the solution of equation
\eqref{eq:Gul}.
\end{enumerate}

The insightful work of Cerreia-Vioglio, Dillenberger, and Ortoleva (2020)
addresses the last point by providing an explicit representation for the
monetary certainty equivalents of Gul's model in terms of cautious expected
utility,\footnote{Cerreia-Vioglio, Dillenberger, and Ortoleva (2015).} using
solely the parameters $u$ and $\beta$ of the original implicit representation. A recent manuscript of Gul and Pesendorfer (2026) provides an alternative explicit minimum-over-rank-dependent-utilities representation.\footnote{See the discussion in Section \ref{sec:expectiled-more}.}
\bigskip

The present paper tackles all the challenges outlined above and offers the
following contributions:

\begin{enumerate}
\item An extension of the theory to \emph{general outcomes} in a framework
\`{a} la Savage (1954).

\item A behavioral foundation based on \emph{explicit axioms} of
disappointment aversion/elation seeking, which are easy to interpret and
determine the sign of $\beta$.

\item An explicit representation in terms of \emph{maxmin expected utility}
\`{a} la Gilboa and Schmeidler (1989), and another one in terms of
\emph{asymmetric least squares estimation} of the random payoff $u\left(
X\right)  $ that the DM is facing:%
\begin{equation}
v_{X}=\mathop{\arg\min}_{v\in\mathbb{R}}\mathbb{E}\left[  \ell_{\beta
}(u(X)-v)\right]  \label{eq:spoiler}%
\end{equation}

where the deviation penalty $\ell_{\beta}(s)=s^{2}+\beta s^{2}1_{\left(
-\infty,0\right]  }(s)$ captures the different emotional impact of elating and
disappointing outcomes.
\end{enumerate}

Relative to the last point, similarly to Cerreia-Vioglio et al.~(2020), both
the explicit representations that we obtain use solely the parameters of
equation (\ref{eq:Gul}) and yield exactly the same solution $v_{X}$. \bigskip

Altogether, our results provide a novel framework, a new set of axioms, and some
explicit representations for Gul's model, offering the following benefits:

\begin{enumerate}
\item The extended domain (general outcomes) broadens the applicability of the
model, allowing it to capture disappointment and elation beyond monetary rewards.

\item The new axiomatization opens up avenues for empirical testing and calibration on the
descriptive side, while providing a solid decision-theoretic foundation. 

\item The explicit representations enhance the interpretability of the model,
facilitate comparison with alternative theories, improve its computational
tractability, and relate decision theory to econometric modeling.
\end{enumerate}

Conceptually and mathematically, these contributions are enabled by an
unexpected connection between Gul's (1991) model of disappointment aversion
and the asymmetric least squares estimation framework of Newey and Powell
(1987), as foreshadowed by representation~(\ref{eq:spoiler}) above and further
detailed in the paper outline below.

Last but not least, the same connection also suggests a neurocomputational foundation
for the theory. If a value estimate is updated through reward prediction
errors with different learning gains for positive and negative surprises, then
the long-run value learned by the system is not the expected reward, but an
expectile of the reward distribution. In this interpretation, the coefficient
$\beta$ has a direct process meaning: it measures the relative strength of the
negative prediction-error channel compared with the positive one. Thus
disappointment aversion can be read not only as a behavioral attitude and a
robust decision rule, but also as the limiting outcome of asymmetric
reinforcement learning.\bigskip

\subsection{Outlook of the paper}


\subsubsection{Beyond monetary outcomes}

The original theory of disappointment aversion was formulated for monetary
lotteries, yet the psychological phenomenon it seeks to capture --- the
asymmetric weighting of elation versus disappointment relative to an
endogenous benchmark --- appears at least as forcefully in non-monetary
domains. Everyday life abounds with examples: the novelist whose manuscript is
rejected after years of hope, the PhD applicant denied entry to her dream
program, the marathon runner sidelined by persistent injury despite flawless
preparation, and so on. In each case, the emotional impact is shaped not
merely by the absolute quality of the outcome, but by how far it deviates ---
upward or downward --- from the DM's prior evaluation of the potential of
their choices.
{\color{black}For instance, recent evidence from movie ratings links higher expectations after award
nominations to subsequent disappointment \citep{rossiSchleef2026}.}

To illustrate the importance of extending Gul's model beyond money, consider a
DM choosing among three wellness activities: running,
strength training, and rock climbing. The quality of the experience in each
activity depends critically on an exogenous state variable: how congested
their weekly schedule turns out to be (low, moderate, or high congestion). 
Assume that the DM needs to commit to the chosen activity (possibly due to planning and budget constraints) before observing the state.
The
table below summarizes the subjective, qualitative payoffs the decision maker
anticipates in each state.

\begin{table}[h]
\centering
\resizebox{0.99\textwidth}{!}{	\renewcommand{\arraystretch}{1.3}
\begin{tabular}{>{\centering\arraybackslash}l|p{6cm}|p{6cm}|p{6cm}}
			\textbf{Action / State} & \textbf{Low Congestion} & \textbf{Moderate Congestion} & \textbf{High Congestion} \\ \hline
			\textbf{Running} & \textbf{Flow and endurance} \newline --- steady rhythm, clear mind &
			\textbf{Habitual ease}\newline --- controlled effort, satisfaction &
			\textbf{Minimal but consistent effort}\newline --- sense of discipline preserved \\ \hline
			\textbf{Strength training} & \textbf{Focused strength}\newline --- visible progress and control &
			\textbf{Routine competence}\newline --- balance of effort and reward &
			\textbf{Manageable fallback}\newline --- short workouts still productive \\ \hline
			\textbf{Rock climbing} & \textbf{Exhilaration}\newline --- challenge, mastery and vistas&
			\textbf{Engaged but uneven}\newline --- effortful satisfaction &
			\textbf{Frustration}\newline --- barriers, fatigue, likely skip\\
\end{tabular}	
}\end{table}Running exhibits \emph{low state-dependence} of
outcomes: even under severe time pressure, a short run still delivers a
meaningful sense of discipline and accomplishment. The worst-case experience
remains tolerable and rarely generates sharp disappointment relative to the
activity's own ex-ante expectation.

Strength training shows \emph{moderate state-dependence}: tight schedules
force shorter sessions and slower progress, yet the activity retains a core
feeling of competence and control. Disappointments are real but arguably controlled.

Rock climbing displays \emph{high state-dependence}: when time and energy are
abundant the experience can be intensely rewarding. But the same activity
collapses under high congestion: long commutes, inadequate warm-up, and
fatigue turn exhilaration into frustration or outright abandonment. The
potential for both extreme elation and disappointment is therefore much greater.

A disappointment-averse agent (positive $\beta$) fears states
in which realized experience falls below the act's endogenous reference point
(its own ex-ante value). Such an agent is drawn toward \textbf{running}
precisely because disappointing realizations are both rare and mild: the
activity is emotionally ``robust.''

An elation-seeking agent (negative $\beta$) instead chases large upside potential
even at the cost of more frequent or intense shortfalls. For this agent,
\textbf{rock climbing} becomes especially attractive: the possibility of
transcendent sessions outweighs the risk of repeated frustration.

Finally, an agent with balanced attitudes (zero $\beta$) often finds
\textbf{strength training} most appealing: it offers genuine progress and
reward in good conditions while still delivering acceptable value when
conditions deteriorate.

Extending Gul's original arguments, which relied on lotteries over a bounded
monetary interval $[w,b]$, to this richer setting is conceptually and
mathematically non-trivial. Moreover, a purely von Neumann--Morgenstern lottery framework tends to entangle decision modeling with statistical modeling. By contrast, the Savage (1954) framework separates acts, the objects of choice, from the probability measure on states, the statistical model of risk. This separation is useful not only conceptually but also empirically: in experimental and real-life settings, the language of state-contingent outcomes is often more natural and more common.

In Section~\ref{sec:expectiled-framework} we therefore develop this extended
model in a   Savage framework. Outcomes $x \in\mathcal{X}$ belong to a
general measurable space, risky alternatives are acts $X : \Omega
\to\mathcal{X}$, and the decision maker evaluates each act relative to its own
endogenous certainty equivalent --- exactly as Gul envisioned --- except now
without any restriction to monetary lotteries.

\subsubsection{New formulas and interpretations}

\label{sect:intro}

Section \ref{sec:expectiled-utility} opens with a simple observation
(Proposition \ref{prop:expectiles}): in our general framework, $v_{X}$ is a
solution of equation (\ref{eq:Gul}) if, and only if, it is a solution
of\footnote{We tacitly assume suitable integrability in the Introduction.}
\begin{equation}
\underset{\text{expected elation}}{\underbrace{\mathbb{E}\left[  \left(
u\left(  X\right)  -v\right)  ^{+}\right]  }}\ =\left(  1+\beta\right)
\underset{\text{expected disappointment}}{\underbrace{\mathbb{E}\left[
\left(  v-u\left(  X\right)  \right)  ^{+}\right]  }}. \label{eq:Gul00}%
\end{equation}
For $\beta \ge 0$, this means that expected elation $\mathbb{E}[(u(X)-v_{X})^{+}]$
must overcompensate expected disappointment by a factor $1+\beta$.
Mathematically, setting $U=u\left(  X\right)  $\ and $\alpha=\left(
2+\beta\right)  ^{-1}$, equation (\ref{eq:Gul00}) shows that $v_{X}$ is a
solution of
\begin{equation}
\alpha\mathbb{E}[(U-v)^{+}]=(1-\alpha)\mathbb{E}[(v-U)^{+}]. \label{eq:NP}%
\end{equation}
This establishes the anticipated link between Gul's model (\ref{eq:Gul00})
and econometrics. In fact, Newey and Powell (1987)
introduce the solutions of (\ref{eq:NP}), denoted by $\mathbb{E}_{\beta}[U]$,
as asymmetric least squares estimators. Under the name \emph{expectiles},
these estimators became a popular tool in econometrics due to their robustness
properties and statistical elicitability (e.g.~Gneiting, 2011).
Subsequently they gained prominence as coherent measures of risk in finance and
actuarial science (e.g.~Bellini, Klar, M\"{u}ller, Rosazza Gianin, 2014, and
Ziegel, 2016). More recently, the framework of Cerreia-Vioglio, Corrao, and Lanzani
(2024) suggests a natural role for expectile-type rules in robust
opinion aggregation, especially when agents are more cautious about over-predicting than under-predicting (or conversely).

The original results of Newey and Powell guarantee that equation
(\ref{eq:Gul00}) always has a unique solution
\begin{equation}
v_{X}=\mathbb{E}_{\beta}[u(X)]. \label{eq:expectiled-utility}%
\end{equation}
In turn, our Proposition \ref{prop:expectiles} shows that $\mathbb{E}_{\beta
}[u(X)]$ is the only solution of Gul's equation (\ref{eq:Gul}). Then we call
$\mathbb{E}_{\beta}[u(X)]$ the \emph{expectiled utility} of $X$ (expectile of
the utility), in analogy with the \emph{expected utility} of $X$ (expectation
of the utility) which corresponds to $\beta=0$.

Another direct application of the results of Newey and Powell leads to the
representation%
\begin{equation}
\mathbb{E}_{\beta}[u(X)]=\mathop{\arg\min}_{v\in\mathbb{R}}\mathbb{E}\left[
\ell_{\beta}(u(X)-v)\right]  \label{eq:quadratic}%
\end{equation}
where $\ell_{\beta}(s)=s^{2}+\beta s^{2}1_{\left(  -\infty,0\right]  }(s)$ is
an asymmetric quadratic loss function. This representation recasts the theory
of Gul in terms of (internal) utility estimation on the part of a DM who
anticipates asymmetric emotional responses to positive and negative
deviations. Newey and Powell also show that $\mathbb{E}_{\beta}[\cdot]$ is
positively homogeneous and constant-additive; furthermore, $\mathbb{E}_{\beta
}[\cdot]$ is monotone and, for $\beta\geq0$, superadditive, two important
properties used in the mathematical finance literature.\footnote{See, for
instance, Bellini et al.~(2014).}

To the decision theorist, this says that expectiled utility is a
specification of Gilboa and Schmeidler's (1989) maxmin expected utility in a
Savage framework.\footnote{See Casadesus-Masanell, Klibanoff, and Ozdenoren
(2000), Ghirardato, Maccheroni, Marinacci, and Siniscalchi (2001, 2003), and
Alon and Schmeidler (2014).} Our \textbf{first main result}, Theorem
\ref{th:GulNeweyPowellPlus}, characterizes expectiled utility within the class
of maxmin expected utilities. We show that
\begin{equation}
\mathbb{E}_{\beta}\left[  u\left(  X\right)  \right]  =\min_{Q\in
\mathcal{Q}^{\beta}}\mathbb{E}^{Q}\left[  u\left(  X\right)  \right]
\label{eq:maxmin1}%
\end{equation}
where
\begin{equation}
\mathcal{Q}^{\beta}=\left\{  Q:\frac{\mathrm{d}Q}{\mathrm{d}P}=\frac{1_{D^{c}
}+(1+\beta)1_{D}}{1+\beta P\left(  D\right)  }\text{ for some\textrm{\ }}
D\in\mathcal{F}\right\}  \label{eq:maxmin2}%
\end{equation}
and an optimum is attained at $Q^{\ast}$ given by
\begin{equation}
\frac{\mathrm{d}Q^{\ast}}{\mathrm{d}P}=\frac{1_{D_{X}^{c}}+(1+\beta)1_{D_{X}}
}{1+\beta P\big(D_{X}\big)} \label{eq:maxmin3}%
\end{equation}
where $D_{X}=\left\{  \omega:u\left(  X\left(  \omega\right)  \right)
<\mathbb{E}_{\beta}\left[  u\left(  X\right)  \right]  \right\}  $ is the set
of states $\omega$ in which the ex-post utility $u\left(  X\left(
\omega\right)  \right)  $ of the realized outcome $x=X\left(  \omega\right)  $
is inferior to the ex-ante utility evaluation $\mathbb{E}_{\beta}\left[
u\left(  X\right)  \right]  $ of the act $X$ itself: the \emph{disappointing
states} for act $X$.

The representation (\ref{eq:maxmin1})--(\ref{eq:maxmin3}) admits a clear
interpretation in terms of a (zero-sum) game against Nature.\footnote{See the
classical Milnor (1954) and Luce and Raiffa (1957, p.~279), as well as Gilboa
and Schmeidler (1989).} A disappointment-averse DM envisages an adversarial
Nature who is able to distort the reference probability $P$ within a set
$\mathcal{Q}^{\beta}$ of alternative distributions to minimize her expected
utility. Anticipating this, the evaluation of act $X$ is robustly computed via (\ref{eq:maxmin1}). The set $\mathcal{Q}^{\beta}$, given by (\ref{eq:maxmin2}%
), explicitly captures these adversarial distributions, with Nature being able
to increase the odds of any single event $D$ by a factor of $(1+\beta)$.
Finally, the \textquotedblleft optimal sabotage\textquotedblright\ on the part of
Nature, $Q^{\ast}$ given by (\ref{eq:maxmin3}), precisely embodies the fears
of the DM by making salient the set $D_{X}$ of disappointing states for act
$X$.

\subsubsection{New axiomatizations}

In the light of the above representations, the preferences $\succsim$ of a DM
adhering to Gul's theory have to satisfy some standard axioms on the set of
all simple acts.\footnote{As customary in axiomatic decision theory,
preferences are represented by a binary relation $\succsim$ on the set of all
acts that take a finite number of values in a connected, metric, and separable
space $\mathcal{X}$, called \emph{simple acts}. Moreover, in the tradition of
Savage (1954), $\left(  \Omega,\mathcal{F},P\right)  $ is assumed to be
nonatomic, so that simple acts generate all simple lotteries.} First, these
preferences are \emph{probabilistically sophisticated}:
\[
P\left(  \omega:X\left(  \omega\right)  \succsim x\right)  \geq P\left(
\omega:Y\left(  \omega\right)  \succsim x\right)  \text{ for all }%
x\in\mathcal{X}\text{ implies }X\succsim Y
\]
with strict preference if the inequality is strict for some $x$ (Machina and
Schmeidler, 1992, p.~754). Second, these preferences are \emph{invariant
biseparable} in the sense of Ghirardato, Maccheroni, Marinacci, and
Siniscalchi (2001, henceforth GMMS).\footnote{Ghirardato, Maccheroni,
Marinacci, and Siniscalchi (2003) is the abridged published version.} In
particular, there exist a nonconstant and continuous function $u:\mathcal{X}%
\rightarrow\mathbb{R}$, and a monotone, positively homogeneous, and
constant-additive functional $I$ on the set of simple random variables such
that, for all simple acts $X$ and $Y$,
\begin{equation}
X\succsim Y\Longleftrightarrow I(u\left(  X\right)  )\geq I(u\left(  Y\right)
). \label{eq:repr}%
\end{equation}
GMMS provide axioms characterizing this representation. Under these axioms,
for all outcomes $x$ and $y$, it is possible to elicit (up to indifference) a
\emph{preference midpoint}
\[
\frac{1}{2}x\oplus\frac{1}{2}y\text{ in }\mathcal{X}\text{ such that }u\left(
\frac{1}{2}x\oplus\frac{1}{2}y\right)  =\frac{1}{2}u\left(  x\right)
+\frac{1}{2}u\left(  y\right)  .
\]
Moreover, GMMS show that the additional axiom needed to characterize
\emph{maxmin expected utility} within the class of invariant biseparable ones
is:%
\begin{equation}
X\sim Y\implies X\precsim\dfrac{1}{2}X\oplus\dfrac{1}{2}Y \label{eq:AH}%
\end{equation}
called \emph{ambiguity hedging}. The interpretation of (\ref{eq:AH}) is
literally the one of Schmeidler (1989, p. 582): \textquotedblleft
smoothing\textquotedblright\ or averaging utility distributions makes the DM
better off.

In Section \ref{sec:axiom}, we propose a new axiom, called
\emph{disappointment hedging}, which, together with probabilistic
sophistication and invariant biseparability, characterizes expectiled utility:%
\begin{equation}
X\sim Y\implies\dfrac{1}{2}W\oplus\dfrac{1}{2}X\precsim\dfrac{1}{2}%
W\oplus\dfrac{1}{2}Y \label{eq:DH}%
\end{equation}
for all simple acts $W$ that have the same disappointing states as $X$, that
is, $\{\omega:W(\omega)\prec W\}=\{\omega:X(\omega)\prec X\}$. This novel
axiom admits a clear interpretation in terms of disappointment aversion. If
$W$ has the same disappointing states as $X$, then mixing $X$ with $W$ offers
no protection: both acts disappoint in the same states. By contrast, mixing
$Y$ with $W$ may offer some hedging benefits, thus making the mixture of $Y$
with $W$ preferable. Also observe that considering $W=X$ in (\ref{eq:DH})
shows that disappointment hedging is a stronger axiom than ambiguity hedging.

In the GMMS framework, our Expectiled Utility Theorem (Theorem
\ref{th:expectiled-utility}) shows that a binary relation $\succsim$ between
simple acts is probabilistically sophisticated, invariant biseparable, and
disappointment hedging if, and only if, there exists a nonconstant and
continuous function $u:\mathcal{X}\rightarrow\mathbb{R}$ and a number
$\beta\geq0$ such that
\begin{equation}
X\succsim Y\iff\mathbb{E}_{\beta}\left[  u\left(  X\right)  \right]
\geq\mathbb{E}_{\beta}\left[  u\left(  Y\right)  \right]  . \label{eq:final}%
\end{equation}
In this case, $u$ is cardinally unique and $\beta\geq0$ is unique. When the
preference in (\ref{eq:DH}) is reversed, we have \emph{elation speculating}
preferences, characterized by the representation (\ref{eq:final}) with $-1<\beta\leq0$. This is the \textbf{second main result} of the paper, and
technically the most demanding. It requires the creation of novel techniques
to analyze attitudes toward different dependence structures among random
variables and the development of a \textquotedblleft Riesz representation
theorem\textquotedblright\ for expectiles (Theorem \ref{th-characterizationEx}).

\subsubsection{Behavioral identification of parameters}
Section~\ref{sec:elicitation} turns the Expectiled Utility Theorem  into a direct elicitation procedure. Specifically, the coefficient $\beta$ can be recovered from binary acts alone, before the utility function is known. Fix $x\succ y$ and consider any preference midpoint of theirs. Theorem~\ref{th:elicitation} shows that any event $F\in\mathcal{F}$  such that
\[
\dfrac{1}{2}x\oplus\dfrac{1}{2}y\sim xFy
\]
offers the elicitation
\[
\beta=\frac{2P(F)-1}{1-P(F)},
\]
and such an event always exists.
Moreover, $P(F)>1/2$, $P(F)=1/2$, and $P(F)<1/2$ correspond, respectively, to disappointment aversion, expected utility, and elation seeking. This confirms the behavioral intuition that disappointment-averse behavior  demands elation to overcompensate disappointment.  

Once $\beta$ is retrieved, the techniques of Ghirardato and Marinacci (2001) allow us to efficiently recover the utility function $u$, again with binary-act indifferences only.

\subsubsection{Expectiled rewards: behavioral and neurocomputational foundations}
Section~\ref{sec:exp-v} studies expectiled rewards directly,
that is, the case of utility-valued acts. First, it provides a lean and expressive behavioral characterization of this important case. Second and fundamental, it connects the behavioral model with a simple learning mechanism. Theorem
\ref{th:RPE-learning} shows that, when a value estimate is adjusted by
asymmetric reward prediction errors, with negative surprises weighted by
$1+\beta$ relative to positive surprises, the estimate converges almost surely
to $\mathbb{E}_{\beta}[U]$. This \textbf{third main result} gives a process interpretation of the
same parameter that appears in Gul's representation and in our elicitation
formula: larger $\beta$ corresponds to stronger updating from disappointing
realizations and therefore to lower learned values for risky rewards. The
neurocomputational discussion relates this reduced-form learning rule to
dopamine-based reinforcement learning and to evidence that positive and
negative prediction errors may affect value learning asymmetrically.

\subsubsection{Additional results}

The additional results of Section~\ref{sec:expectiled-more} provide further insights into different aspects of our theory and connect them to the existing literature.
\begin{itemize}
\item Expectiled utilities can be characterized as the only probabilistically sophisticated class of maxmin or maxmax expected utility preferences satisfying betweenness formulated in the Savage framework, connecting to the classic studies on betweenness for lotteries.

\item Expectiled utility is the worst-case expected utility over a Hilbert
projective ball around the reference probability. This connects its maxmin
representation to the density-ratio neighborhoods of robust statistics and to
modern distributionally robust optimization.

\item The axiom of disappointment hedging leads to expectiled utilities also beyond the realm of invariant biseparable preferences.
\end{itemize}

\subsubsection*{Synopsis}

The results that we summarized above show that the
theory of disappointment aversion of Gul (1991) can be \textquotedblleft
remastered\textquotedblright\ in a Savage framework with general outcomes, new
explicit representations, and an axiomatization that directly appeals to the
idea of robustifying decisions to take protection against disappointment (or
exploit elation opportunities); moreover, it has an unexpected neurocomputational foundation based on the differential processing of elation and disappointment in the brain. \bigskip

\section{Framework}

\label{sec:expectiled-framework}

We adopt a Savage framework augmented with a reference probability $P$. Specifically:

\begin{itemize}
\item $(\Omega,\mathcal{F},P)$ is a probability space of \emph{states};

\item $\mathcal{X}$ is a measurable space of (deterministic) \emph{outcomes};

\item \emph{acts} are random outcomes (measurable mappings) $X:\Omega
\rightarrow\mathcal{X}$.
\end{itemize}

An act $X$ represents a risky action the outcome of which is $X\left(
\omega\right)  $ if state $\omega$ occurs. Each act $X$ induces a distribution
$P_{X}=P\circ X^{-1}$ of outcomes, called a \emph{lottery} in the decision
theory jargon. Beyond the advantages discussed in the introduction,
representing risky actions as acts rather than lotteries allows for a natural
description of dependence structures across different actions. When
considering two acts simultaneously, the Savage framework provides a joint
specification of outcomes across states, enabling the modeling of patterns
that would be lost if acts were reduced to separate lotteries (their marginal
distributions). This capability is crucial for our axiomatization since it
allows us to specify when two acts disappoint the DM in exactly the same
states.\footnote{See the \emph{disappointment hedging} axiom below.}
Furthermore, as in Machina and Schmeidler (1992), the probability measure $P$
in the present framework can be interpreted as a subjective one.

Translated into this Savage framework, Gul's model describes a DM who
evaluates act $X:\Omega\rightarrow\mathcal{X}$ through a solution $v_{X}$ of%
\begin{equation}
v=\mathbb{E}\left[  k_{v}(X)\right],  \label{eq:Gulagain}%
\end{equation}
where $u:\mathcal{X}\rightarrow\mathbb{R}$ is a (measurable) \emph{utility
function}, $\beta\in(-1,\infty)$ is a \emph{disappointment aversion coefficient},
and $k_{v}:\mathcal{X}\rightarrow\mathbb{R}$ is the   transformation
of $u$ given by%
\[
k_{v}(x)=%
\begin{cases}
\dfrac{u(x)+\beta v}{1+\beta} & \text{if }u(x)\geq v\\
\hspace{16pt}u(x) & \text{if }u(x)<v
\end{cases}
\]
for all $v\in\mathbb{R}$.

Note that, for $\beta\geq0$,
\[
u(x)\geq v\implies\beta v\leq\beta u\left(  x\right)  \implies\dfrac
{u(x)+\beta v}{1+\beta}\leq\dfrac{u(x)+\beta u(x)}{1+\beta}=u\left(  x\right)
.
\]
Therefore, when $v=v_{X}$, the utility of disappointing outcomes, for which
$u(x)<v_{X}$, is unaffected, while that of elating ones, for which $u(x)\geq
v_{X}$, is discounted. Thus Gul's implicit representation captures the
asymmetry between positive and negative deviations from the endogenously
determined reference point $v_{X}$. This describes the full emotional impact
of disappointment and the deflated one of elation. The coefficient $\beta$
governs the strength of this asymmetry. \bigskip

\noindent\textbf{Nota Bene.} For brevity, we focus mostly on the case of
disappointment aversion, corresponding to $\beta\geq0$. The case $-1<\beta
\leq0$, which reflects elation seeking, admits analogous formal results and
interpretive reversals. We address it only when the duality between
disappointment aversion and elation seeking yields additional conceptual insight.

\section{Expectiled utility representations}

\label{sec:expectiled-utility}

Our first result shows that, also in the framework that we introduced in the
previous section, Gul's equation (\ref{eq:Gulagain}) always has a unique
solution $v_{X}$, an internal equilibrium between elation and
disappointment.\footnote{In reading it, recall that $s^{+}=\max\left\{
s,0\right\}  $ denotes the positive part of a real number $s$.}

\begin{proposition}
[Gul's evaluations are expectiled utilities]
\label{prop:expectiles}If $\mathbb{E}\left[  u\left(  X\right)  \right]  $
exists finite, then the following conditions are equivalent for $v\in
\mathbb{R}$:

\begin{itemize}
\item $v\ $is a solution to Gul's equation (\ref{eq:Gulagain});

\item $v$ is a solution of
\begin{equation}
\mathbb{E}\left[  \left(  u\left(  X\right)  -v\right)  ^{+}\right]
\ =\left(  1+\beta\right)  \mathbb{E}\left[  \left(  v-u\left(  X\right)
\right)  ^{+}\right]  . \label{eq:DA}%
\end{equation}

\end{itemize}

\noindent In particular, $v$ exists and is unique, denoted $\mathbb{E}_{\beta
}\left[  u\left(  X\right)  \right]  $.
\end{proposition}

This formulation provides an alternative take on Gul's model. Rather than
working with a piecewise transform $k_{v}$ of $u$, Proposition \ref{prop:expectiles}
characterizes the DM's evaluation $\mathbb{E}_{\beta}\left[  u\left(
X\right)  \right]  $ of act $X$ as the sure utility level $v$ that balances
expected elation, $\mathbb{E}\left[  \left(  u\left(  X\right)  -v\right)
^{+}\right]  $, and expected disappointment, $\mathbb{E}\left[  \left(
v-u\left(  X\right)  \right)  ^{+}\right]  $. The factor $(1+\beta)$
explicitly measures the degree of emotional compensation required: expected
elation must outweigh expected disappointment by such a factor to achieve
internal equilibrium.

As anticipated in the introduction, $\mathbb{E}_{\beta}\left[  u\left(
X\right)  \right]  $ is the $\alpha$\emph{-expectile} of $u(X)$ in the sense
of Newey and Powell (1987) for $\alpha=\left(  2+\beta\right)  ^{-1}$.
Therefore, in analogy with the \emph{expected utility} terminology for
$\mathbb{E}\left[  u\left(  X\right)  \right]  $, expectation of the utility
of $X$, we call $\mathbb{E}_{\beta}\left[  u\left(  X\right)  \right]  $
\emph{expectiled utility} of $X$.
The econometric connection that we have just
established presents us with another representation that bears a conceptual
nuance of subjective estimation.

\begin{theorem}
[Expectiled utility as asymmetric least squares]
\label{th:expectiles}If $\mathbb{E}\left[  u\left(  X\right)  ^{2}\right]  $
is finite, then
\[
\mathbb{E}_{\beta}[u(X)]=\mathop{\arg\min}_{v\in\mathbb{R}}\mathbb{E}\left[
\ell_{\beta}(u(X)-v)\right]  ,
\]
where $\ell_{\beta}:\mathbb{R}\rightarrow\mathbb{R}$ is given by
\[
\ell_{\beta}(s)=%
\begin{cases}
s^{2} & \text{if }s\geq0,\\
(1+\beta)s^{2} & \text{if }s<0.
\end{cases}
\]

\end{theorem}

Here the agent looks for the sure utility level that best approximates the
random utility levels yielded by $X$. In doing so, he acknowledges that ex-post
utility gains and shortfalls will have different psychological
\textquotedblleft distance\textquotedblright\ from the ex-ante evaluation of
the act. This formulation highlights the optimization nature of value
formation. The DM behaves like an asymmetric least squares estimator:
minimizing the average squared deviation of payoffs from a reference value,
but doing so with a bias that mirrors emotional asymmetry. This perspective
reframes the DM's valuation of risk not as a fixed point solution, but as the
result of a deliberate, internal optimization based on emotional awareness.

Proposition \ref{prop:expectiles} and Theorem \ref{th:expectiles} describe the
DM as an active mitigator of emotional distress aligning with neuroscientific
evidence on nonlinear loss responses (e.g. Tom, Fox, Trepel, and Poldrack,
2007).\bigskip


Our \textbf{first main result} shows that expectiled utility is a special
case of maxmin expected utility. In reading it, recall that, given an event $D
\in\mathcal{F}$ and a bias coefficient $\beta\geq0$, the \emph{Jeffrey
update} of $P$ given $D$ with bias $\beta$ is given by%
\[
Q_{D}^{\beta}(A)=\frac{P\left(  A\cap D^{c}\right)  +(1+\beta)P\left(  A\cap
D\right)  }{P\left(  D^{c}\right)  +(1+\beta)P\left(  D\right)  }\qquad\forall
A\in\mathcal{F}\text{.}%
\]
This soft version of Bayesian updating proposed by Jeffrey (1965) describes
the result of anticipating --- making salient --- event $D$. It has the
following features:

\begin{itemize}
\item for $\beta=0$, $Q_{D}^{\beta}$ coincides with the original probability
measure $P$, no anticipation;

\item for any $\beta>0$, the odds of $D$ under $Q_{D}^{\beta}$ are inflated by
a factor $1+\beta$ with respect to those under $P$, that is,
\[
\frac{Q_{D}^{\beta}\left(  D\right)  }{Q_{D}^{\beta}\left(  D^{c}\right)
}=\left(  1+\beta\right)  \frac{P\left(  D\right)  }{P\left(  D^{c}\right)  }
\]
while the conditional probabilities remain unaffected
\[
Q_{D}^{\beta}\left(  A\mid D\right)  =P\left(  A\mid D\right)  \qquad
\text{and} \qquad Q_{D}^{\beta}\left(  A\mid D^{c}\right)  =P\left(  A\mid
D^{c}\right)
\]
for all $A\in\mathcal{F}$;

\item as $\beta\rightarrow\infty$, $Q_{D}^{\beta}$ converges to the Bayesian
update of $P$ on the event $D$, and the DM behaves as if $D$  will  occur for sure.
\end{itemize}

We are now ready for the statement:

\begin{theorem}
[Expectiled utility as maxmin expected utility]
\label{th:GulNeweyPowellPlus}If $\mathbb{E}\left[  u\left(  X\right)  \right]
$ exists finite and $\beta\geq0$, then
\[
\mathbb{E}_{\beta}\left[  u\left(  X\right)  \right]  =\min_{Q\in
\mathcal{Q}^{\beta}}\mathbb{E}^{Q}\left[  u\left(  X\right)  \right]  ,
\]
where
\[
\mathcal{Q}^{\beta}=\left\{  Q_{D}^{\beta}:D\in\mathcal{F}\right\}
\]
and an optimum is attained at $Q^{\ast}$ given by
\begin{equation}
Q^{\ast}=Q_{D_{X}}^{\beta} \label{q:stella}%
\end{equation}
where $D_{X}=\left\{  \omega:u\left(  X\left(  \omega\right)  \right)
<\mathbb{E}_{\beta}\left[  u\left(  X\right)  \right]  \right\}  $.
\end{theorem}

As discussed in the introduction (see Section \ref{sect:intro}), Theorem
\ref{th:GulNeweyPowellPlus} reveals disappointment aversion as a \emph{robust
approach against fictitious adversarial scenarios}. {It is straightforward to check that $Q_{D}^{\beta}$ is given by the Radon--Nikodym derivative  
\[
\frac{\mathrm{d}Q_{D}^{\beta}}{\mathrm{d}P}=\frac{1_{D^{c} }+(1+\beta)1_{D}%
}{1+\beta P\left(  D\right)  }
\]
for all $D\in\mathcal{F}$.} 
This maxmin representation features:

\begin{itemize}
\item A game-theoretic interpretation. The DM views the decision problem as a
game against a malevolent Nature that is able to distort probabilities,
overweighting disappointing states, to minimize his utility.

\item Robustness against disappointment. The coefficient $\beta$ quantifies
defensive preparedness. Higher $\beta$ implies greater anticipated sabotage of
favorable outcomes, but higher protection comes at a higher opportunity cost.
Formally, this corresponds to the fact that $\mathbb{E}_{0}\left[  u\left(
X\right)  \right]  =\mathbb{E}\left[  u\left(  X\right)  \right]  $, while
$\mathbb{E}_{\beta}\left[  u\left(  X\right)  \right]  $ decreases as $\beta$
increases, and the limit $\mathbb{E}_{\infty}\left[  u\left(  X\right)
\right]  $ is the essential infimum of $u\left(  X\right)  $ under $P$ as
$\beta\rightarrow\infty$.

\item A meaningful minimizer. The optimal sabotage $Q^{\ast}$ magnifies the
likelihood of states $\omega$ in which $u(X(\omega))<\mathbb{E}_{\beta}\left[
u\left(  X\right)  \right]  $, making disappointment salient through
probability inflation.
\end{itemize}

Denoting by $\succsim$ the preference relation between acts $X$ that is
represented by $\mathbb{E}_{\beta}\left[  u\left(  X\right)  \right]  $,
Theorem \ref{th:GulNeweyPowellPlus} provides an entirely new lens through
which Gul's theory can be understood. For each act $X$, the event%
\[
\{\omega:u(X(\omega))<\mathbb{E}_{\beta}\left[  u\left(  X\right)  \right]
\}=\{\omega:X(\omega)\prec X\}
\]
is the one in which the DM is disappointed. The emotional salience of this
event induces the DM to overweight the probability of this set and to evaluate
the expected utility of $X$ with respect to the distorted probability
$Q^{\ast}\ $given by (\ref{q:stella}) that inflates the odds of this event by
a factor $1+\beta$. 

Finally, Theorem \ref{th:GulNeweyPowellPlus} provides the general
perspective behind the fundamental intuition of Ghirardato and Marinacci
(2001) regarding the behavior of a disappointment-averse DM who confronts bets
(binary acts). Consider two outcomes $x\succ y$ and an act $xEy$ that delivers
$x$ if $E$ occurs and $y$ otherwise, so that the elating states are the ones
in $E$ and the disappointing states the ones its complement $E^{c}$. Ghirardato and
Marinacci show that
\begin{equation}
\mathbb{E}_{\beta}\left[  u\left(  xEy\right)  \right]  =\frac{P\left(E\right)  }{1+\beta P\left(  E^{c}\right)  }u\left(  x\right)  +\left(
1+\beta\right)  \frac{P\left(  E^{c}\right)  }{1+\beta P\left(  E^{c}\right)
}u\left(  y\right)  \label{eq:GM}%
\end{equation}
highlighting how the probability of elating states in $E$ reduced and that of disappointing states in $E^{c}$ is augmented. The formula%
\[
\mathbb{E}_{\beta}\left[  u\left(  X\right)  \right]  =\mathbb{E}^{Q^{\ast}%
}\left[  u\left(  X\right)  \right]
\]
implied by Theorem \ref{th:GulNeweyPowellPlus} is the generalization of (\ref{eq:GM}) to nonbinary acts.

Formula~\eqref{eq:GM} also presents a simple experimental procedure to elicit the parameters $u$ and $\beta$ of the model, as shown in the next Section~\ref{sec:elicitation}.

\section{Expectiled utility theory}

\label{sec:expectiled-axioms}

\subsection{Axiomatization}
\label{sec:axiom}
In this section, we consider a preference relation $\succsim$ over the set
$\mathbb{X}$ of all acts that take a finite number of values in a connected,
metric, and separable space $\mathcal{X}$, called \emph{simple acts}.
Moreover, in the tradition of Savage (1954), $\left(  \Omega,\mathcal{F}%
,P\right)  $ is assumed to be adequate, that is, either nonatomic or such that
$\mathcal{F}$ is generated by a finite partition (of at least two cells) over which $P$ is uniform.
The set of all simple random variables, which are $\mathcal{F}$-measurable
real-valued functions that take finitely many values, is denoted by
$B_{0}(\Omega, \mathcal{F})$. In this framework we provide an axiomatic
characterization of the preferences on $\mathbb{X}$ that admit an expectiled
utility representation.

The results of the previous sections guarantee that expectiled utility
preferences are probabilistically sophisticated and invariant biseparable. The
first property, probabilistic sophistication (Machina and Schmeidler, 1992),
is a preferential form of first-order stochastic dominance. It says that acts
that deliver better outcomes with higher probability are preferred. Formally:
\[
P\left(  \omega:X\left(  \omega\right)  \succsim x\right)  \geq P\left(
\omega:Y\left(  \omega\right)  \succsim x\right)  \text{ for all }%
x\in\mathcal{X}\text{ implies }X\succsim Y
\]
with strict preference if the inequality is strict for some $x$.

The second property, \emph{invariant biseparability}, requires the possibility
of separating a cardinal utility from attitudes toward uncertainty
(Ghirardato, Maccheroni, and Marinacci, 2005). Formally, this means that there
exists a continuous nonconstant function $u:\mathcal{X}\rightarrow\mathbb{R}$
and a monotone, positively homogeneous, and constant-additive functional
$I:B_{0}(\Omega,\mathcal{F})\rightarrow\mathbb{R}$ such that, for all acts $X$
and $Y$,
\begin{equation}
X\succsim Y\Longleftrightarrow I(u\left(  X\right)  )\geq I(u\left(  Y\right)
) \label{eq:nove}%
\end{equation}
with $I(1_{F})\in\left(  0,1\right)  $ for some $F\in\mathcal{F}$. GMMS
provide an axiomatization of invariant biseparable preferences in the
framework of this section. The utility function $u$ that they obtain is
cardinally unique and $I$ is unique. Recently, Castagnoli, Cattelan,
Maccheroni, Tebaldi, and Wang (2022) and Chandrasekhar, Frick, Iijima, and Le
Yaouanq (2022) provided concrete representations of the functional $I$
appearing in (\ref{eq:nove}). When preferences are invariant biseparable, for
all outcomes $x$ and $y$, it is possible to elicit from betting behavior (that
is from the restriction of $\succsim$ to bets) a \emph{preference midpoint}
\[
\frac{1}{2}x\oplus\frac{1}{2}y\text{ in }\mathcal{X}\text{ such that }u\left(
\frac{1}{2}x\oplus\frac{1}{2}y\right)  =\frac{1}{2}u\left(  x\right)
+\frac{1}{2}u\left(  y\right).
\]
The resulting preference midpoint does not depend on the positive affine normalization of \(u\) in (\ref{eq:nove}) and is unique up to indifference.
\bigskip

\noindent\textbf{Remark.} This was first shown by GMMS (Lemma 3); subsequently
the theoretical and the experimental literature presented alternative ways
of obtaining $(1/2)x\oplus\left(  1/2\right)  y$ from the preference
$\succsim$, and put this technology into action. We refer the reader to
K\"{o}bberling and Wakker (2003), Abdellaoui, Bleichrodt, and Paraschiv (2007),
Baillon, Driesen, and Wakker (2012), Dean and Ortoleva (2017), Ghirardato and
Pennesi (2020), Chateauneuf, Maccheroni, and Zank (2026). \bigskip

The additional axiom required to characterize \emph{maxmin expected utility
preferences} within invariant biseparable ones is:%
\begin{equation}
X\sim Y\implies X\precsim\dfrac{1}{2}X\oplus\dfrac{1}{2}Y\label{eq:AHHA}%
\end{equation}
called \emph{ambiguity hedging} (GMMS, Proposition 10).\footnote{Maxmin
expected utility preferences are represented by $v_{X}=\min_{Q\in
\mathcal{Q}}\mathbb{E}^{Q}\left[  u\left(  X\right)  \right]  $ where $\mathcal{Q}$ is a set of probabilities over $\mathcal{F}$. Their
duals, \emph{maxmax expected utility preferences}, represented by $v_{X}%
=\max_{Q\in\mathcal{Q}}\mathbb{E}^{Q}\left[  u\left(  X\right)
\right]  $, are characterized by the  \emph{ambiguity
speculating} axiom, obtained by reversing the preference in (\ref{eq:AHHA}). See again GMMS.}
Our novel key axiom,
\emph{disappointment hedging}, is indeed a stronger version of it:
\begin{equation}
X\sim Y\implies\frac{1}{2}W\oplus\frac{1}{2}X\precsim\frac{1}{2}W\oplus
\frac{1}{2}Y \label{eqDhhD}%
\end{equation}
for all simple acts $W$ that have the same disappointing states as $X$, that
is, $\{\omega:W(\omega)\prec W\}=\{\omega:X(\omega)\prec X\}$. We already
discussed its interpretation in the introduction: if $W$ has the same
disappointing states as $X$, then mixing $X$ with $W$ offers no protection
since both acts disappoint in the same event. By contrast, mixing $Y$ with
$W$ may offer some hedging benefits, thus making the mixture of $Y$ with $W$
preferable. By reversing the preference in (\ref{eqDhhD}) we obtain a dual
axiom, \emph{elation speculating}, with the inverse meaning.  

\begin{theorem}
[Expectiled utility]\label{th:expectiled-utility}Let $(\Omega,\mathcal{F},P)$
be an adequate probability space and $\mathcal{X}$ be a connected and
separable metric space. The following conditions are equivalent for a binary
relation $\succsim$ on $\mathbb{X}$:

\begin{enumerate}
\item[(i)] $\succsim$ is probabilistically sophisticated, invariant
biseparable, and disappointment hedging (resp.~elation speculating);
\end{enumerate}

\begin{enumerate}
\item[(ii)] there exists a continuous and nonconstant function $u:\mathcal{X}%
\rightarrow\mathbb{R}$ and a number $\beta\geq0$ (resp.~$-1<\beta\leq0$) such
that
\[
X\succsim Y\iff\mathbb{E}_{\beta}\left[  u\left(  X\right)  \right]
\geq\mathbb{E}_{\beta}\left[  u\left(  Y\right)  \right]  .
\]

\end{enumerate}

In this case, $u$ is cardinally unique and $\beta$ is unique.
\end{theorem}

This is the  \textbf{second main theorem} of this paper providing a novel
axiomatic foundation of Gul's theory: in a general Savage framework, with
transparent axioms that reveal the asymmetry between disappointment and
elation, and supported by the explicit representations of the previous section.

\subsection{Elicitation}

\label{sec:elicitation}

Remarkably, the parameters $u$ and $\beta$ appearing in Theorem
\ref{th:expectiled-utility} can be directly elicited from preferences. 

\begin{theorem}
[Elicitation of expectiled utility's parameters]\label{th:elicitation}Let $P$ be nonatomic, $\succsim$ satisfy
the equivalent conditions of Theorem~\ref{th:expectiled-utility}, and $x,y \in \mathcal{X}$ be such that $x \succ y$.

\begin{enumerate}[(i)]
\item Any $F\in\mathcal{F}$ with $P(F)\in (0,1)$ such
that
\begin{equation}
\label{eq:pro}
\dfrac{1}{2}x\oplus\dfrac{1}{2}y\sim xFy
\end{equation}
yields
\begin{equation}
\beta=\frac{2P(F)-1}{1-P(F)},\label{eq:beta-elicitation}%
\end{equation}
and such an event $F$ always exists.

\item  For each $z$ such that $x \succsim z \succsim y$, any  $E\in\mathcal{F}$ such that
\begin{equation*}
z\sim xEy
\end{equation*}
yields
\begin{equation}
u(z)=\frac{P(E)}{1+\beta(1-P(E))}u(x)+(1+\beta)\frac{1-P(E)}{1+\beta
(1-P(E))}u(y),\label{eq:utility-elicitation}%
\end{equation}
and such an event $E$ always exists.
\end{enumerate}

\end{theorem}

Part (i) of Theorem \ref{th:elicitation} shows that $\beta$ can be retrieved by eliciting a single indifference between a
sure outcome and a binary act. Also, the probability of the ``matching  event'' $F$ immediately declares the attitudes of the DM: $P(F)>1/2$
corresponds to disappointment aversion, $P(F)=1/2$ to expected utility, and
$P(F)<1/2$ to elation seeking.
Also observe that event $F$ is characterized by
\[
P\left(  F\right)  =\frac{1+\beta}{2+\beta}%
\]
irrespective of $x$ and $y$. This provides another testable implication of Theorem \ref{th:expectiled-utility}. 
Moreover, after $\beta$ has been elicited, under the innocuous normalization $u(x)=1$ and $u(y)=0$,  \eqref{eq:utility-elicitation} in  part (ii) becomes
\[
u(z)=\frac{P(E)}{1+\beta(1-P(E))}.
\]

Finally, our approach differs from that of Cerreia-Vioglio et al.~(2020): we only consider
indifferences between binary acts, $\beta$ is retrieved first, and $u$ is derived
from $\beta$. By contrast, Cerreia-Vioglio et al.~(2020) first retrieve $u$, through the \textquotedblleft
expected utility core\textquotedblright\ of the whole preference over
lotteries, then obtain $\beta$ from $u$ by the lottery counterpart of equation
(\ref{eq:pro}).
{\color{black}The techniques of 
	\citet{wakkerDeneffe1996}, \citet{abdellaoui2000}, and \citet{abdellaouiBleichrodt2007} also
elicit $u$ first.}
\bigskip

\noindent\textbf{Remark.} An immediate corollary of Theorem \ref{th:elicitation} is that: given any two expectiled utility preferences $\succsim_{\mathrm A}$ and $\succsim_{\mathrm B}$ that agree on $x\succ y$, the condition $\beta_{\mathrm A}\geq\beta_{\mathrm B}$ is equivalent to
\begin{equation}
\label{comparative}
xEy \succsim_{\mathrm A} z_{\mathrm A}
\implies
xEy \succsim_{\mathrm B} z_{\mathrm B},
\end{equation}
for every event $E$ and any preference midpoints $z_{\mathrm A}$ and $z_{\mathrm B}$ of $x$ and $y$.

In words, whenever the potential elation from betting on $E$ offsets disappointment for decision maker $\mathrm A$, the same is true for decision maker $\mathrm B$. Thus, comparative disappointment aversion can be assessed independently of the taste components $u_{\mathrm A}$ and $u_{\mathrm B}$.

 \bigskip

\noindent\textbf{Remark.}  
We can use Theorem \ref{th:elicitation} to obtain a simple lab elicitation procedure for $\beta$ involving   three indifference judgments on 
fair coin tosses,
and one on a binary bet. 
Fix  any two outcomes $x\succ y$ and an event $H$ with $P(H)=1/2$. 
Let $z,z_{-},z_{+}\in\mathcal{X}$ be  the 
certainty equivalents  given by
\begin{equation}
 z\sim xHy,\qquad z_{+}\sim xHz,\qquad z_{-}\sim zHy. \label{eq:elicitation-consequences}
\end{equation}
These exist by  \eqref{eq:GM} because $u$ is continuous and $u\left( \mathcal{X} \right)$ is an interval. Moreover, by the same formula,
\[x \succ z_{+}\succ z\succ z_{-} \succ y. \] 
 Direct calculation shows that $z$ is a preference midpoint of $z_+$ and $z_-$. Now by Theorem \ref{th:elicitation}, there is an event $F$ such that  the binary act $z_{+}F z_{-}$ is indifferent to $z$,
 and $\beta$ is elicited by \eqref{eq:beta-elicitation}.

\section{Expectiled reward: behavioral and learning foundations}

\label{sec:exp-v}

In this section, we take a direct look at the essential object behind
expectiled utility (axiomatized above), that is, \emph{expectiled reward}. In
this case, the acts that the DM is facing are bounded random variables; thus
$\mathbb{X}$ is $L^{\infty}\left(  P\right)  $, constant acts are real numbers,
and subjective mixtures (with $\oplus$) are replaced by usual averages (with
$+$). The elements of $\mathbb{X}$ are interpreted as random rewards (utility-valued acts in the decision-theoretic jargon, neural representations in the neuroscience jargon).

\subsection{Expectiled reward}

The following properties are considered for a binary relation $\succsim$ on
$\mathbb{X}$, where typical elements are denoted by $U$, $V$, and $Z$.

\begin{itemize}
\item \emph{monotonicity:} given any $x,y\in\mathbb{R}$, $x\geq y$ implies
$x\succsim y$ and conversely;

\item \emph{continuity:} the upper and lower level sets of $\succsim$ are
closed in the $L^{\infty}\left(  P\right)  $-norm;

\item \emph{disappointment aversion:} given any $U,V\in L^{\infty}\left(  P\right)  $,
\[
U\sim V\implies Z+U\precsim Z+V
\]
for all $Z\in L^{\infty}\left(  P\right)$ that have the same disappointing states as $U$.

\end{itemize}

The latter is the only new axiom since the others are well known and well
studied in the literature. Its interpretation is very clear: the DM shies away
from the accumulation of disappointment. If $\{\omega:Z(\omega)\prec
Z\}=\{\omega:U(\omega)\prec U\}$, and the DM holding $Z$ acquires $U$ then he
runs the risk of being simultaneously disappointed by both his (utility)
payoffs, whereas $V$ may allow some elating compensation when $Z$ disappoints.
An even more precise name for this axiom would be \emph{disappointment
accumulation aversion}; we opted for the shorter version
for brevity. The obvious dual axiom is called \emph{elation (accumulation)
seeking}.

\begin{theorem}
[Expectiled reward]\label{th:expectiled-value}\label{th-ax-expectileutils} Let
$(\Omega,\mathcal{F},P)$ be an adequate probability space. The following
conditions are equivalent for a binary relation $\succsim$ on $L^{\infty}\left(  P\right)$:

\begin{enumerate}
\item[(i)] $\succsim$ is a monotone, continuous, probabilistically
sophisticated, and disappointment averse (resp.~elation seeking) {\color{black}weak order};

\item[(ii)] there exists a number $\beta\geq0$ (resp.~$-1<\beta\leq0$) such
that
\[
U\succsim V\iff\mathbb{E}_{\beta}\left[  U \right]  \geq\mathbb{E}_{\beta
}\left[  V \right]  ;
\]

\item[(iii)] $\succsim$ is a monotone, continuous, probabilistically
sophisticated, and disappointment hedging (resp.~elation
speculating) {\color{black}weak order}.\footnote{On the real line we take $\oplus= +$, also recall that a binary relation is a weak order if it is complete and transitive.}
\end{enumerate}

In this case, $\beta$ is unique.
\end{theorem}

The axiomatization of expected utility \`{a} la Savage (1954) descends from
that of \emph{expected reward} of de Finetti (1931). Moreover, the latter
builds on a set of even more transparent axioms.
Theorem~\ref{th:expectiled-value} shows that the same is true for expectiled
utility and expectiled reward: the behavioral force behind the Savage
representation, aversion to accumulated disappointments in the same states,
precisely pins down expectiled reward, under some other standard decision-theoretic assumptions.

\subsection{Neurocomputational foundation}

Here we show how the behavior we characterized above can arise from temporal-difference learning: when a value estimate is updated through asymmetric reward prediction errors, the long-run learned value is precisely the expectiled reward. This kind of asymmetric updating is exactly what
reward-based learning in the brain is now known to exhibit.

Consider a
neural system repeatedly facing  the same random outcome and trying to
assign it a value estimate. For instance, a rat might press a given button
and receive, in each trial, a random amount of fruit juice drawn from the
distribution of $X$ in an i.i.d.\ manner, as in \citet{DabneyEtAl2020}.

In contemporary neuroscience, the neural value estimate $V_t=V_t(X)$ of the
random outcome $X$ after trial $t$ is understood as being represented by
activity in valuation circuits, including the ventral striatum,
orbitofrontal cortex, and medial prefrontal cortex, where value signals of
this kind are most robustly recorded. Such an estimate is updated after
each trial in proportion to the reward prediction error (RPE), namely the
difference $\delta_t=U_{t+1}-V_t$ between the reward $U_{t+1}=U_{t+1}(X)$
received at trial $t+1$ and the current estimate $V_t$. The canonical
neuroscientific finding is that phasic activity of midbrain dopamine
neurons, especially in the ventral tegmental area (VTA) and substantia
nigra pars compacta (SNc), encodes such a prediction-error signal:
dopamine-neuron firing increases when rewards are better than expected, is
approximately unchanged when rewards are as expected, and decreases or
pauses when expected rewards fail to arrive. This is the central finding of
\citet{SchultzDayanMontague1997}, and later work has quantified dopamine
firing as an RPE-like signal in progressively more detail, notably in the
distributional reinforcement learning account of \citet{DabneyEtAl2020}.

The fact that negative and positive prediction errors
$\delta_t=U_{t+1}-V_t$ of the same magnitude can produce different
value-estimate updates is also supported both at a physiological and at a
behavioral level. In particular, \citet{FrankSeebergerOReilly2004} show
that Parkinson's patients with depleted striatal dopamine learn
preferentially from negative outcomes off medication, and preferentially
from positive outcomes on dopaminergic medication. A simple biologically
plausible reduced-form model of neural value dynamics is
\[
V_{t+1}=V_t+\alpha_t
\begin{cases}
w_+(U_{t+1}-V_t), & \text{if } U_{t+1}\geq V_t, \\
w_-(U_{t+1}-V_t), & \text{if } U_{t+1}<V_t,
\end{cases}
\]
where $w_+$ and $w_-$ are the effective gains on positive and negative
RPEs, respectively, and the sequence of numerical coefficients  $\alpha_t$ describes the learning effects of these gains over trials. This is the type of dynamic suggested by \citet{DabneyEtAl2020}. The biological interpretation is that the coefficients multiplying
positive and negative RPEs summarize the net efficacy with which
dopaminergic fluctuations induce plasticity in downstream circuits.\footnote{A
useful reduced-form picture is to associate the positive-RPE channel with
dopamine bursts and D1/direct-pathway plasticity, and the negative-RPE
channel with dopamine dips, D2/indirect-pathway plasticity, and aversive
or omission-related pathways \citep{GerfenSurmeier2011}.} 

A natural specification for the learning plasticity
coefficients is
\[
\alpha_{t}=\alpha_{0}\frac{T}{T+t},
\]
with $\alpha_{0}>0$ and $T\in \mathbb{N}$, so that $\alpha
_{T}=\alpha_{0}/2$ can be interpreted as a \textquotedblleft learning
half-life\textquotedblright.  Moreover
the normalization
\[
w_+=1, \qquad w_-=1+\beta,
\]
is completely innocuous, since it is obtained by the general two-gain model via
\[
\widetilde{\alpha}_t=\alpha_t w_+, \qquad 1+\beta=\frac{w_-}{w_+}.
\]

Our final theorem shows that this learning dynamic leads to expectiled reward evaluation: it  provides a plausible neurocomputational
foundation of Gul's theory of disappointment aversion and shows how simple
stylized agents can compute expectiled rewards.

\begin{theorem}
[Asymmetric RPE learning converges to expectiled reward]%
\label{th:RPE-learning} Let $\beta>-1$. Let $(U_{t})_{t\in\mathbb{N}}$ be an
i.i.d.~sequence of copies of  a bounded random variable $U$. Consider the
reinforcement learning process
\[
V_{t+1}=V_{t}+\alpha_{t}%
\begin{cases}
U_{t+1}-V_{t}, & \text{if }U_{t+1}\geq V_{t},\\
(1+\beta)(U_{t+1}-V_{t}), & \text{if }U_{t+1}<V_{t},
\end{cases}
\]
with deterministic initial condition $V_{0}\in\mathbb{R}$. If the  learning
plasticity coefficients $(\alpha_{t})_{t\ge 0}$ satisfy
\[
\alpha_{t}>0,\qquad\sum_{t=0}^{\infty}\alpha_{t}=\infty,\qquad\sum
_{t=0}^{\infty}\alpha_{t}^{2}<\infty,
\]
then
\[
V_{t}\rightarrow\mathbb{E}_{\beta}[U]%
\]
almost surely under $P$.
\end{theorem}

This \textbf{third} (and final) \textbf{main result} finding connects dopamine-based learning and disappointment
aversion through the econometric theory of expectiles. Neural agents whose
reward estimates are updated by asymmetric RPE learning converge to
expectiled reward evaluations.

{\color{black}Relatedly, \citet{beggs2022} studies reference-dependent learning in a finite-state, finite-action dynamic programming framework and establishes convergence to a recursive utility evaluation of the Epstein--Zin \citeyearpar{epsteinZin1989} type, with Gul's disappointment-aversion certainty equivalents as a special case.\footnote{\color{black}More generally, Beggs's framework accommodates certainty equivalents associated with betweenness preferences in the sense of Chew (1989) and Dekel (1986).} By contrast, our Theorem~\ref{th:RPE-learning} establishes the almost-sure convergence of asymmetric reward-prediction-error learning directly to expectiled utility for general bounded i.i.d.\ rewards.}

\section{More expectiled utility theory}
\label{sec:expectiled-more}

This section develops additional axiomatic characterizations of expectiled utility preferences. 
These results offer a more comprehensive picture of the role of expectiled utilities among some classic decision models.


\subsection{Disappointment hedging, ambiguity hedging, and betweenness}

We clarify the relation between disappointment hedging and
ambiguity hedging by providing an alternative axiomatization of expectiled
utility based on a Savagean version of betweenness.


Betweenness is one of the classical weakenings of the von Neumann--Morgenstern
Independence Axiom. On a lottery space, it requires that, whenever two
lotteries $p$ and $q$ are indifferent, every probabilistic mixture of them is
indifferent to both:
\begin{equation}
p\sim q\quad\Longrightarrow\quad\lambda p+(1-\lambda)q\sim p,\qquad\lambda
\in\lbrack0,1].\label{eq:betweenness}%
\end{equation}
Thus betweenness preserves neutrality toward randomization among equally good
alternatives, while allowing preferences to be nonlinear away from
indifference sets. This idea is central in the weighted-utility and
implicit-utility tradition; see, among others, Chew (1983), Dekel (1986), Chew
(1989), and Chew and Epstein (1989). Gul's disappointment-aversion model is a
leading member of this class, and Cerreia-Vioglio et al.~(2020) provide
explicit representations for disappointment aversion and related betweenness preferences.

To reproduce betweenness in a Savage framework, one can randomize between acts
$X$ and $Y$ by means of an event $\Lambda \in\mathcal{F}$ that is independent
of the acts being mixed. Indeed, if $\Lambda $ is independent of both $X$ and
$Y$, then the act $X\Lambda  Y$, which agrees with $X$ on $\Lambda $ and with $Y$
on $\Lambda ^{c}$, has distribution $P(\Lambda )P_{X}+(1-P(\Lambda ))P_{Y}$,
where $P_{X}$ and $P_{Y}$ are the distributions of $X$ and $Y$. This motivates
the following definition of betweenness.

We say that $\succsim$ on $\mathbb{X}$ satisfies \emph{betweenness} if, for
all acts $X$ and $Y$ and every event $\Lambda \in\mathcal{F}$ that is
independent of both $X$ and $Y$,
\begin{equation}
X\sim Y\implies X\Lambda  Y\sim X.\label{eq:independent-betweenness}%
\end{equation}
This is a weakening of Savage's Sure Thing Principle that clearly mirrors the
weakening (\ref{eq:betweenness}) of von Neumann--Morgenstern's Independence
Axiom, actually it formalizes the way in which (\ref{eq:betweenness}) is
typically phrased in expository presentations.

\begin{theorem}
[Expectiled utility and betweenness]
\label{th:expectiled-utility-bis}Let $(\Omega,\mathcal{F},P)$ be a nonatomic
probability space and $\mathcal{X}$ be a connected and separable metric space.
The following conditions are equivalent for a binary relation $\succsim$ on
$\mathbb{X}$:

\begin{enumerate}
\item[(i)] $\succsim$ is a probabilistically sophisticated maxmin
(resp.~maxmax) expected utility preference that satisfies betweenness;

\item[(ii)] there exists a continuous and nonconstant function $u:\mathcal{X}
\rightarrow\mathbb{R}$ and a number $\beta\geq0$ (resp.~$-1<\beta\leq0$) such
that
\[
X\succsim Y\iff\mathbb{E}_{\beta}\left[  u\left(  X\right)  \right]
\geq\mathbb{E}_{\beta}\left[  u\left(  Y\right)  \right]  .
\]

\end{enumerate}

In this case, $u$ is cardinally unique and $\beta$ is unique.
\end{theorem}

Although both lead to the expectiled utility model, the key condition in
Theorem~\ref{th:expectiled-utility} and that in
Theorem~\ref{th:expectiled-utility-bis} are conceptually quite different:
disappointment hedging has a transparent content in terms of avoidance of certain states, whereas betweenness reconnects our analysis to the nonexpected
utility theories developed at the end of the Eighties, but its behavioral
content in terms of disappointment/elation attitudes is unclear. Be that as it may,
for probabilistically sophisticated and invariant biseparable preferences,
disappointment hedging is equivalent to ambiguity hedging and betweenness,
yielding expectiled utility. Schematically:
\bigskip

\begin{figure}[h!]
	\centering
	\resizebox{.97\textwidth}{!}{\begin{tikzpicture}[
			box/.style={
				draw,
				rounded corners=3pt,
				line width=.55pt,
				align=center,
				inner xsep=8pt,
				inner ysep=5pt,
				minimum height=1.45cm,
				minimum width=4.35cm,
				font=\small
			},
			equiv/.style={
				font=\normalsize,
				inner sep=1pt
			}
			]
			\node[box] (A) at (0,0)
			{\textbf{Invariant biseparability}\\[2pt]
				$+$\\[2pt]
				\textbf{disappointment hedging}};
			\path (A.east) ++(.35,0) node[equiv, anchor=west] (AB) {$\iff$};
			\path (AB.east) ++(.35,0) node[box, anchor=west] (B)
			{\textbf{Maxmin expected utility}\\[2pt]
				$+$\\[2pt]
				\textbf{betweenness}};
			\path (B.east) ++(.35,0) node[equiv, anchor=west] (BC) {$\iff$};
			\path (BC.east) ++(.35,0) node[box, minimum width=3.75cm, anchor=west] (C)
			{\textbf{Expectiled utility}\\[2pt]
				$X\mapsto \E_\beta[u(X)]$\\[1pt]
				$\beta\ge 0$};
	\end{tikzpicture}}
\end{figure}

\noindent \textbf{Remark.} The maxmin/maxmax assumption in
Theorem~\ref{th:expectiled-utility-bis} (which is stronger than the invariant
biseparability used in Theorem~\ref{th:expectiled-utility}) is essential. It
cannot be replaced by invariant biseparability alone. Indeed, fix $p>1$ and
$\alpha\in(0,1)$, and let $q_{\alpha,p}(U)$ be the $L^{p}$-quantile of a
bounded random variable $U$, that is, it is the unique number $m$ solving
\[
\alpha\mathbb{E}\left[  (U-m)_{+}^{p-1}\right]  =(1-\alpha)\mathbb{E}\left[
(m-U)_{+}^{p-1}\right]  .
\]
For $p=2$ these functionals are precisely expectiles (see Proposition \ref{prop:expectiles}). Bellini et al.~(2014) study these generalized quantiles and show, in
particular, that the convex or concave members of this class are the
expectiles. For $p\neq2$, the preferences 
\[
X\succsim Y\quad\Longleftrightarrow\quad q_{\alpha,p}(u(X))\geq q_{\alpha
,p}(u(Y))
\]
are probabilistically sophisticated, invariant biseparable, and satisfy
betweenness, without being maxmin or maxmax.\bigskip

\noindent \textbf{Remark.}
Theorem \ref{th:expectiled-utility-bis} also clarifies the relation between
our Theorem \ref{th:GulNeweyPowellPlus} and a recent manuscript of
Gul and Pesendorfer (2026). Their analysis is formulated in a standard
monetary-lottery setting. Within their class of Lorenz Expected Utility
preferences, they characterize Gul's model by first-order risk aversion and
betweenness, and represent Gul's evaluation as the minimum over a
one-parameter family of rank-dependent expected-utility functionals with
piecewise-linear probability distortions.

The betweenness characterization in our Theorem
\ref{th:expectiled-utility-bis} links ambiguity hedging to first-order risk
aversion, or aversion to risk in utility units, in the spirit of Dekel's
(1989) analysis of diversification without the independence axiom. Moreover,
after translating between monetary lotteries and utility-valued acts, their
minimum-over-rank-dependent-utilities representation can be rewritten as the
maxmin expected-utility representation in Theorem
\ref{th:GulNeweyPowellPlus}. In sum, their representation can be viewed as
the distributional, monetary, Lorenz-curve counterpart of our Savagean,
act-based maxmin representation for general outcomes.

\subsection{Robust statistics and distributionally robust optimization}
\label{sec:hilbert-dro}

The maxmin representation in Theorem~\ref{th:GulNeweyPowellPlus} also places
expectiled utility in the robust-statistics tradition of density-ratio
neighborhoods. Wasserman (1992) studied these neighborhoods through their
invariance properties, including their stability under Bayesian updating and
measurable transformations. Wasserman and Kadane (1992a) developed Monte
Carlo methods for computing bounds on expectations over these sets of
probability measures, while Wasserman and Kadane (1992b) studied the related
symmetric upper and lower probabilities. The construction below shows that the lower expectation over each such density-ratio neighborhood is an expectile; consequently, expectiled utilities are the corresponding distributionally robust optimization solutions.

The \emph{Hilbert projective distance} between two probability measures $Q$
and $P$ is
\[
\mathcal{H}\left(Q,P\right)=
\begin{cases}
\operatorname*{ess\,sup}\left(\log\dfrac{\mathrm{d}Q}{\mathrm{d}P}\right)
-\operatorname*{ess\,inf}\left(\log\dfrac{\mathrm{d}Q}{\mathrm{d}P}\right),
& \text{if }Q\approx P,\\
\infty, & \text{otherwise},
\end{cases}
\]
where $Q\approx P$ denotes mutual absolute continuity.\footnote{The essential extrema under $P$ are understood in the extended-real
sense. In particular, $\mathcal{H}(Q,P)=\infty$ if either
$\mathrm{d}Q/\mathrm{d}P$ or $\mathrm{d}P/\mathrm{d}Q$ is not essentially
bounded. All almost-sure (a.s.) statements are understood with respect to $P$.} Its closed ball of
radius $\kappa\geq0$ is
\[
B_{\mathcal H}(P,\kappa)
=\left\{Q:\mathcal H(Q,P)\leq\kappa\right\}.
\]
Equivalently, if $\varphi=\mathrm{d}Q/\mathrm{d}P$, then
\begin{equation}
Q\in B_{\mathcal H}(P,\kappa)
\quad\Longleftrightarrow\quad
\frac{\operatorname*{ess\,sup}\varphi}
{\operatorname*{ess\,inf}\varphi}\leq e^{\kappa}.  
\label{eq:hilbert}
\end{equation}

Therefore, the Hilbert projective neighborhood $B_{\mathcal H}(P,\kappa)$ of $P$ with radius $\kappa$ is precisely the Wasserman and Kadane \emph{density-ratio neighborhood} of $P$ with radius $e^\kappa$.
Moreover, the Hilbert projective distance satisfies the data processing inequality; see Cohen and Fausti
(2024).

In modern distributionally robust optimization (DRO), an expected utility is
replaced by its worst-case value over a neighborhood in some statistical distance. Common specifications
use $f$-divergence neighborhoods (for instance, Kullback--Leibler neighborhoods), or Wasserstein neighborhoods; see Ben-Tal et
al.~(2013) and Mohajerin Esfahani and Kuhn (2018).
{\color{black}\citet{vanParys2021} give a statistical foundation for relative-entropy
DRO by controlling out-of-sample disappointment. Here instead, an Hilbert neighborhood
and its radius are determined by disappointment-averse preferences, as the next result shows.}

\begin{theorem}
[Expectiled utility and Hilbertian DRO]
\label{th:hilbert-DRO}
If $\mathbb{E}[u(X)]$ exists finite and $\beta\geq0$, then
\begin{equation}
\mathbb{E}_{\beta}[u(X)]
=\min\left\{\mathbb{E}^{Q}[u(X)]:
\mathcal H(Q,P)\leq\log(1+\beta)\right\}.
\label{DRO}
\end{equation}
Moreover, if $P$ is nonatomic, then
$\mathcal Q^\beta$ consists exactly of the center $P$ and the set of all the extreme points of $B_{\mathcal H}(P, \log(1+\beta))$; and in any case $\mathcal Q^\beta$ is contained in such a ball. 
\end{theorem}

In view of Theorem~\ref{th:GulNeweyPowellPlus},
\[
Q^\ast=Q_{D_X}^{\beta}
\]
is a least-favorable prior for $u(X)$ in the density-ratio neighborhood of $P$ of radius $1+\beta$. This presents expectiled utility simultaneously
as a normatively motivated (in the sense of Wasserman, 1992) robustification of $\mathbb E[u(X)]$, and as the solution of the
distributionally robust optimization problem \eqref{DRO}. Whereas the Kullback--Leibler
divergence controls for the average log likelihood ratio, the Hilbert projective distance controls for the oscillation (the essential range) of the log
likelihood ratio. Unlike the Wasserstein distance, it does not require a metric
on the state space.

Finally, the connection \eqref{DRO} between expectiles and the lower expectations  of Wasserman and Kadane
(1992a)  is new, and the Monte Carlo techniques that they obtain to compute these lower expectations provide yet another way to compute expectiled utilities.

\subsection{Weakening invariant biseparability}

In Section \ref{sec:expectiled-axioms}, our characterization of expectiled
utility builds on the assumption of invariant biseparability, which induces a
mixture space operation $\oplus$ on $\mathcal{X}$ (in the sense of Fishburn,
1982, p.~14) and yields a representation of preferences
\begin{equation}
X\succsim Y\iff I(u(X))\geq I(u(Y))\label{eq:sepa}%
\end{equation}
on $\mathbb{X}$, where $I:B_{0}(\Omega,\mathcal{F})\rightarrow\mathbb{R}$ is
monotone, positively homogeneous, and constant-additive (\emph{a fortiori},
supnorm continuous), and $u:\mathcal{X}\rightarrow\mathbb{R}$ is continuous,
nonconstant, and affine under $\oplus$ (\emph{a fortiori}, cardinally unique).
To highlight the scope of our disappointment hedging axiom, we next assume
that $\mathcal{X}$ can be endowed with a mixture space operation $\oplus$, and
we consider the class of preference relations on $\mathbb{X}$ that admit a
representation (\ref{eq:sepa}), where $I$ is only required to be monotone,
supnorm continuous, and normalized,\footnote{Normalization means that $I(v)=v$
for all $v\in\mathbb{R}$.} that is, a Chisini mean, and $u:\mathcal{X}%
\rightarrow\mathbb{R}$ to be continuous, nonconstant, and affine under
$\oplus$. Borrowing the ideas of Cerreia-Vioglio et al.~(2011), we call these \emph{rational preferences (under uncertainty)} and note that they encompass most of the preference relations studied in decision theory under uncertainty (ibid.).

\begin{theorem}
[Expectiled utility under rational preferences] \label{th-generalex} Let
$(\Omega,\mathcal{F},P)$ be an adequate probability space and $\mathcal{X}$ be
a connected and separable metric space endowed with a mixture operation
$\oplus$. The following conditions are equivalent for a rational preference
$\succsim$ on $\mathbb{X}$ represented by \eqref{eq:sepa}:

\begin{enumerate}
\item[(i)] $\succsim$ is probabilistically sophisticated and satisfies
disappointment hedging (resp.~elation speculating) under
$\oplus$;

\item[(ii)] there exists $\beta\geq0$ (resp.~$-1<\beta\leq
0$) such that $I=\mathbb{E}_{\beta}$ on $B_{0}(\Omega,\mathcal{F}%
,u(\mathcal{X}))$.\footnote{As usual, $B_{0}(\Omega,\mathcal{F},u(\mathcal{X}%
))$ denotes the set of all simple random variables taking values in
$u(\mathcal{X}):=\{u(x):x\in\mathcal{X}\}$.}%
\end{enumerate}

In this case, $u$ is cardinally unique and $\beta$ is unique.
\end{theorem}

This result shows that the substantive axioms that deliver Gul’s theory are indeed disappointment hedging and probabilistic sophistication. The latter assumption is relaxed in the companion paper Bellini, Mao, Wang, and Wu (2024). Again, note that for invariant biseparable preferences the mixture operation $\oplus$ is endogenously derived, while for rational preferences it is typically exogenously assumed (or obtained from an additional strength-of-preference relation).
\begin{appendix}

\section{Properties of expectiles}
In this section, we provide several properties of expectiles that will be used in the proofs of our results.
Denote by $L^1(P)$ the space of all random variables with finite expectation under $P$. We write $\essinf$ and $\esssup$ for the
essential infimum and essential supremum with respect to $P$, respectively.
We say that a sequence $\{X_n\}_{n\in\mathbb N}$ converges bounded a.s.
to $X$ if $X_n\to X$ a.s. and
$
\{\esssup|X_n|\}_{n\in\N}
$
is bounded.
\begin{lemma}\label{lm-propertyEx}
Let $X,Y\in L^1(P)$, and $\beta>-1$. We have the following properties of the expectile $\E_{\beta}$.
\begin{enumerate}[(i)]
\item {Law invariance}:~If $P_X=P_Y$, then $\E_{\beta}[X]=\E_{\beta}[Y]$;
\item {Strong monotonicity}:~If $P(X \ge t) \ge P(Y \ge t)$ for all $t \in \R$, then $\E_{\beta}[X] \ge \E_{\beta}[Y]$, with strict inequality if the inequality is strict for some $t$;
\item {Constant additivity}:~For any $m\in\R$, $\E_{\beta}[X+m]=\E_{\beta}[X]+m$;
\item {Positive homogeneity}: For any $\lambda\ge 0$, $\E_{\beta}[\lambda X]=\lambda\E_{\beta}[X]$;
\item {Superadditivity (resp.~subadditivity)}:~$\beta\ge 0$ (resp.~$\beta\in(-1,0]$) if, and only if, $\E_{\beta}[X+Y]\ge \E_{\beta}[X]+\E_{\beta}[Y]$ (resp.~$\E_{\beta}[X+Y]\le \E_{\beta}[X]+\E_{\beta}[Y]$);
\item {Bounded a.s.~continuity}:~If a sequence $\{X_n\}_{n\in\N}$ converges to $X$ bounded a.s., then $\E_{\beta}[X_n]\to \E_{\beta}[X]$;
\item {Additivity in concordant sums}:~If $\{\omega: X(\omega)<\E_{\beta}[X]\}=\{\omega: Y(\omega)<\E_{\beta}[Y]\}$, then $\E_{\beta}[X+Y]=\E_{\beta}[X]+\E_{\beta}[Y]$;
\item {Dual representation}:~We have the following representation:
\begin{align*}
\E_{\beta}[X]=\begin{cases}
\min_{\varphi\in\mathcal M_{\beta}} \E[\varphi X]\quad &{\rm if}~\beta\ge 0,\\
\max_{\varphi\in\mathcal M_{\beta}} \E[\varphi X]\quad &{\rm if}~\beta\in(-1,0],
\end{cases}
\end{align*}
where
$$
\mathcal{M}_\beta=\left\{\varphi \in L^{\infty}: \varphi>0 \text { a.s., } \E[\varphi]=1, \frac{\operatorname{ess} \sup \varphi}{\operatorname{ess} \inf \varphi} \leq \gamma(\beta)\right\},
$$
with $\gamma(\beta)=\max \left\{1+\beta,1/(1+\beta)\right\}$. Moreover, an optimal scenario $\bar{\varphi}$  is given by
$$
\bar{\varphi}:=\frac{ 1_{\left\{X>\E_{\beta}[X]\right\}}+(1+\beta) 1_{\left\{X \leq \E_{\beta}[X]\right\}}}{\E\left[ 1_{\left\{X>\E_{\beta}[X]\right\}}+(1+\beta) 1_{\left\{X \leq \E_{\beta}[X]\right\}}\right]}.
$$
\end{enumerate}
\end{lemma}
\begin{proof}
The properties (i) and (iii)-(v) follow  from the studies of expectiles as coherent risk measures (see e.g.~Corollary 4.6 of Ziegel, 2016).
The properties (vi), (vii) and (viii) follow from
Theorem 10 of Bellini et al.~(2014),
Theorem 3 of Bellini et al.~(2021) and Proposition 8 of Bellini et al.~(2014), respectively. To see the property (ii),
if $P(X \ge t) \ge P(Y \ge t)$ for all $t \in \R$, then, for $v=\E_\beta[X]$ and $v'> v$, we have
\begin{align*}
\E[(Y-v')^+] \le
\E[(Y-v)^+] &\le \E[(X-v)^+]
\\&=(1+\beta)\E[(v-X)^+]
\\& \le(1+\beta) \E[(v-Y)^+]
< (1+\beta) \E[(v'-Y)^+],
\end{align*}
where the last  strict inequality  follows from $P(Y<v') \ge P(Y\le v)\ge P(X\le v) >0$.
This shows  that $\E_\beta[Y] >v$ cannot hold. Therefore, $\E_\beta[X] \ge \E_\beta[Y] $.
To see the strict inequality under strict dominance, note that if $\E_\beta[Y]=v$, then  at least one of the two inequalities in
\begin{align*}
(1+\beta) \E[(v-Y)^+]=
\E[(Y-v)^+] &\le \E[(X-v)^+]
\\&=(1+\beta)\E[(v-X)^+]
\le(1+\beta) \E[(v-Y)^+]
\end{align*}
is strict, a contradiction.
\end{proof}
\section{Characterizing expectiles within functionals}
Recall the properties of expectiles presented in Lemma~\ref{lm-propertyEx}.
In this section, we establish a characterization theorem showing that expectiles are the only class of functionals that satisfy strong monotonicity, constant additivity, positive homogeneity, and additivity in concordant sums.
In this appendix, we always assume that $(\Omega,\mathcal F,P)$ is an adequate probability space.

\begin{theorem}\label{th-characterizationEx}
A functional $I:B_0(\Omega,\mathcal F)\to\R$  is strongly monotone, constant-additive, positively homogeneous, and additive in concordant sums if and only if there exists $\beta>-1$ such that $I=\E_{\beta}$.
\end{theorem}
Theorem~\ref{th-characterizationEx} does not constrain the sign of the parameter $\beta$, as additivity in concordant sums has an additive structure. To capture the effect of concordance, we introduce the two properties of monotonicity in concordant sums. A functional $I: B_0(\Omega,\mathcal F) \to \R$ is \emph{decreasing} (resp.~\emph{increasing}) \emph{in concordant sums} if, for all $X, Y, W \in B_0(\Omega,\mathcal F)$ such that $I(X) = I(Y)$ and $\{\omega : X(\omega) < I(X)\} = \{\omega : W(\omega) < I(W)\}$, it holds that $I(W + X) \le I(W + Y)$ (resp.~$I(W + X) \ge I(W + Y)$).
The next result, fundamental to the proofs of Theorems \ref{th:expectiled-utility} and \ref{th-ax-expectileutils}, provides an alternative characterization of expectiles in which decrease or increase in concordant sums serves as a key axiom, alongside two standard conditions: strong monotonicity and $L^\infty$-norm continuity. Specifically, $L^\infty$-norm continuity refers to the property that  $I(X_n) \to I(X)$ if $X_n\to X$ in $L^\infty$, and it can be deduced from strong monotonicity and constant additivity (see e.g.~Lemma 4.3 in F\"{o}llmer and Schied, 2016).
\begin{theorem}\label{th-characterizationEx1}
Let $I:B_0(\Omega,\mathcal F)\to\R$ satisfy $I(v)=v$ for all $v\in\R$. It is strongly monotone, $L^\infty$-norm continuous, and  decreasing (resp.~increasing) in concordant sums if and only if there exists $\beta\ge 0$ (resp.~$\beta\in(-1,0]$) such that $I=\E_{\beta}$.
\end{theorem}
Theorem~\ref{th-characterizationEx1} can be extended to the case where $I$ is defined on the set of acts supported on a nonempty interval $A \subseteq \R$ with $0$ in its interior. In this setting, the properties of decrease and increase in concordant sums need to be slightly modified, since an interval is generally not a linear space. Specifically, denote by $B_0(\Omega,\mathcal F, A)$ the set of all simple random variables taking values in $A$. A
functional $I: B_0(\Omega,\mathcal F, A) \to \R$  is \emph{decreasing} (resp.~\emph{increasing}) \emph{in concordant mixtures} if, for all $X, Y, W \in B_0(\Omega,\mathcal F,A)$ such that $I(X) = I(Y)$ and $\{\omega : X(\omega) < I(X)\} = \{\omega : W(\omega) < I(W)\}$, it holds that
\begin{align*}
I\left(\frac{W+X}{2}\right) \le I\left(\frac{W + Y}{2}\right)\quad \left({\rm resp.}~I\left(\frac{W+X}{2}\right) \ge I\left(\frac{W + Y}{2}\right)\right).
\end{align*}
\begin{theorem}\label{th-characterization_general}
Let $A\subseteq \R$ be an interval with $0$ in its interior. Suppose that $I:B_0(\Omega,\mathcal F,A)\to\R$ satisfies $I(v)=v$ for all $v\in A$. It is strongly monotone, $L^\infty$-norm continuous, and decreasing (resp.~increasing) in concordant mixtures if and only if there exists $\beta\ge 0$ (resp.~$\beta\in(-1,0]$) such that $I=\E_{\beta}$.
\end{theorem}
Below, we present the complete proofs of Theorems \ref{th-characterizationEx}, \ref{th-characterizationEx1} and \ref{th-characterization_general}. For an act $X$ and a functional $I$, we define
\begin{align}\label{eq-D_I}
D_I(X)=\{\omega:X(\omega)<I(X)\}.
\end{align}
In the proofs of Theorems \ref{th-characterizationEx1} and \ref{th-characterization_general}, we focus on the decreasing case, as the increasing case can be treated analogously. 

\begin{proof}[Proof of Theorem \ref{th-characterizationEx}.]
The sufficiency is established in Lemma~\ref{lm-propertyEx}. We now turn to the necessity and address it in two parts, by considering first the finite case and then the infinite case of $\Omega$.
\\
\textbf{Proof for the case of finite $\Omega$.} Assume that $\Omega = \{\omega_1, \dots, \omega_n\}$ and $P(\omega_i) = 1/n$ for all $i \in [n]$. Since $I$ satisfies constant additivity and positive homogeneity, it is natural to focus on the study of the following set:
\begin{align*}
\mathcal{X}^0 = \{X \in B_0(\Omega,\mathcal F) : I(X) = 0\}.
\end{align*}
To make use of additivity in concordant sums, we aim to construct acts $X$ and $Y$ such that $\{\omega : X(\omega) \ge I(X)\} = \{\omega : Y(\omega) \ge I(Y)\}$. This motivates us to introduce the following subset:
\begin{align}\label{eq-X_S}
\mathcal{X}_S = \{X \in \mathcal{X}^0 : \{X \ge 0\} = S\}, \quad S \in \mathcal{F}.
\end{align}
It is clear that $\mathcal X^0=\bigcup_{S\in\mathcal F}\mathcal X_S$. Note that positive homogeneity implies that $I(0)=0$, and thus, combining with strong monotonicity yields
$\mathcal X_{\Omega}=\{0\}$, $\mathcal X_{\varnothing}=\varnothing$, and $\mathcal X_S\neq \varnothing$ for all $S\in\mathcal F\setminus\{\Omega,\varnothing\}$. We proceed with the rest of the proof in three steps.
\begin{enumerate}[(a)]
\item For any $S\in\mathcal F$, there exist finite measures $P_S$ and $Q_S$ on $(\Omega,\mathcal F)$
such that $P_S(S)Q_S(S^c)>0$ if $S\in\mathcal F\setminus \{\Omega,\varnothing\}$ and $\E^{P_S}[X_+]=\E^{Q_S}[X_-]$ for all $X\in\mathcal X_S$. 
This step is the most challenging.
\item Note that strong monotonicity implies law invariance. We use
law invariance to show that $P_S(\omega)=\lambda_S$ for all $\omega\in S$ and $Q_S(\omega)=\eta_S$ for all $\omega\in S^c$, where $\lambda_S,\eta_S$ are two constants. Moreover,
$
{\lambda_{S_1}}/{\eta_{S_1}}={\lambda_{S_2}}/{\eta_{S_2}}
$
if $|S_1|=|S_2|$, where $|S|$ is the cardinality of $S$.
\item Note that strong monotonicity and constant additivity together imply the continuity of $I$ with respect to uniform convergence (see, e.g.~Lemma 4.3 of F\"{o}llmer and Schied, 2016), that is, $\esssup|X_n-X|\to 0$ implies $I(X_n)\to I(X)$.
We apply this continuity property to deduce that $\lambda_S / \eta_S$ must be constant for any $S \in \mathcal{F}$.
This observation uniquely determines the form of $I$, implying that $I = \E_\beta$ for some $\beta>-1$.
\end{enumerate}
\underline{(a)}
The cases $S=\Omega$ and $S=\varnothing$ are trivial.
For $S\in\mathcal F\setminus\{\Omega,\varnothing\}$,
denote by $a_S=I(1_S)$, and strong monotonicity of $I$ yields $a_S\in(0,1)$. Based on this, we define a nonempty set as
$$
\mathcal X_{S}^+=\{X\in B_0(\Omega,\mathcal F): 0<\max X<\min X/a_S\},
$$
and a functional on $\mathcal X_{S}^+$ as
\begin{align*}
	\phi_S(X):=I(X1_S),~~X\in\mathcal X_{S}^+.
\end{align*}
Since $\max (X+Y)\le \max X+\max Y$ and $\min (X+Y)\ge \min X+\min Y$,
it is straightforward to verify that $\mathcal X_{S}^+$ is a convex cone. We also have that $\phi_S(\lambda)=\lambda a_S$ for all $\lambda\ge 0$ by the positive homogeneity of $I$.
For any $X\in\mathcal X_{S}^+$, we have
\begin{align*}
	0<I(X 1_{S})\le I((\max X) 1_S)=\phi_S(\max X)=(\max X) a_S<\min X,
\end{align*}
where the first and the second inequalities follow from the strong monotonicity of $I$.
This gives
\begin{align}\label{eq-disX_S^+}
D_I(X1_{S})=S^c,\quad X\in\mathcal X_S^+.
\end{align}
Therefore,
\begin{align*}
\phi_S(X+Y)=	I(X1_{S}+Y1_{S})=I(X1_{S})+I(Y1_{S})=\phi_S(X)+\phi_S(Y),~~\forall X,Y\in\X_S^+,
\end{align*}
where we have used additivity in concordant sums in the second equality.
Hence, $\phi_S$ is additive on $\mathcal X_S^+$.
Note that for any $X\in B_0(\Omega,\mathcal F)$, $X+m\in\mathcal X_S^+$ for large enough $m>0$. This allows us to
extend $\phi_S$ to $B_0(\Omega,\mathcal F)$ by defining
\begin{align*}
	\widehat\phi_S(X)=\phi_S(X+m_X+1)-(m_X+1) a_S,~~{\rm with}~m_X=\inf\{m:X+m\in\mathcal X_{S}^+\},
\end{align*}
for $X\in B_0(\Omega,\mathcal F)$.
We clarify that $X + m_X$ may not belong to $\mathcal{X}_S^+$ when $\mathcal{X}_S^+$ is an open set. Therefore, we use $X + m_X + 1$ to ensure it lies within the domain of $\phi_S$.
We assert that
\begin{align}\label{eq-+phihat}
	\widehat\phi_S(X)=\phi_S(X+m+1)-(m+1) a_S~~{\rm for~any}~m> m_X.
\end{align}
To see this, for $m> m_X$,
\begin{align*}
	\phi_S(X+m+1)-(m+1) a_S&=	\phi_S(X+m_X+1+m-m_X)-(m+1) a_S\\
	&=	\phi_S(X+m_X+1)+\phi_S(m-m_X)-(m+1) a_S\\
	&=	\phi_S(X+m_X+1)+(m-m_X)a_S-(m+1) a_S\\
	&=	\phi_S(X+m_X+1)-(m_X+1) a_S
	=\widehat\phi_S(X),
\end{align*}
where the second equality follows from the additivity of $\phi_S$. For any $X,Y\in B_0(\Omega,\mathcal F)$, it holds that $$X+m_X+1,~Y+m_Y+1,~X+Y+m_X+m_Y+2\in\mathcal X_{S}^+.
$$ Hence,
\begin{align*}
	\widehat\phi_S(X+Y)
	&=\phi_S(X+Y+m_X+m_Y+2)-(m_X+m_Y+2)a_S\\
	&=\phi_S(X+m_X+1)+\phi_S(Y+m_Y+1)-(m_X+m_Y+2)a_S
	\\&=\widehat\phi_S(X)+\widehat\phi_S(Y),
\end{align*}
where the second equality follows from the additivity of $\phi_S$. This implies that $\widehat\phi_S: B_0(\Omega,\mathcal F)\to\R$ is additive. Note that $\phi_S:\X_S^+\to\R$ is monotone as $I$ is strongly monotone. By the representation of $\widehat\phi_S$ in \eqref{eq-+phihat}, we have that $\widehat\phi_S$ is also monotone.
Hence, there exists a measure $P_S$ on $(\Omega,\mathcal F)$ such that $\widehat\phi_S(X)=\E^{P_S}[X]$ for $X\in B_0(\Omega,\mathcal F)$, and we have
\begin{align*}
	\phi_S(X)=I(X1_S)=\E^{P_S}[X],~~\forall X\in\X_S^+.
\end{align*}
It is straightforward to verify that $P_S(S^c)=0$ as $X1_{S}=0$ on $S^c$ for any $X\in\X_S^+$. Then, we claim that $P_S(S)>0$. Let $X=(2+a_S)1_S+3a_S1_{S^c}$ and $Y=(1+2a_S)1_S+3a_S1_{S^c}$.  It holds that $X,Y\in\mathcal X_S^+$ and
\begin{align*}
\phi_S(X)=I(X1_S)&=\E^{P_S}[X]=P_S(S)(2+a_S);\\
\phi_S(Y)=I(Y1_S)&=\E^{P_S}[Y]=P_S(S)(1+2a_S).
\end{align*}
Since $a_S\in(0,1)$, one can verify that $X 1_S\ge Y 1_S$, and $X 1_S> Y 1_S$ on $S$. Strong monotonicity yields $I(X1_S)>I(Y1_S)$, and thus, $P_S(S)(2+a_S)>P_S(S)(1+2a_S)$, which in turn implies that $P_S(S)>0$.
Next, we turn to the negative counterpart, which follows similarly from the preceding arguments. Hence, we provide a brief proof for this case. Recall that $a_S=I(1_S)\in(0,1)$.
Define
$$\mathcal X_{S}^-=\{X\in B_0(\Omega,\mathcal F): \max X/(1-a_S)<\min X<0\}$$
and
\begin{align*}
	\psi_S(X):=I(X1_{S^c}),~~X\in\mathcal X_{S}^-.
\end{align*}
We aim to show that $\psi_S(X)=\E^{Q_S}[X]$ for $X\in\mathcal X_S^-$ with some measure $Q_S$ on $(\Omega,\mathcal F)$. 
It follows from constant additivity of $I$ that $I(1_S-1)=I(1_S)-1=a_S-1$. Hence, we have
\begin{align*}
\psi_S(m)&=I((-m)(1_S-1))=(-m)I(1_S-1)=(1-a_S)m,~~m<0,
\end{align*}
where the third equality follows from the positive homogeneity of $I$.
For any $X\in\mathcal X_{S}^-$, we have
\begin{align*}
	0&> I(X1_{S^c})\ge I((\min X)(1-1_S))=\psi_S(\min X)=(1-a_S)\min X>\max X.
\end{align*}
where the first and the second inequalities follow from the strong monotonicity of $I$.
This implies that
\begin{align}\label{eq-disX_S^-}
D_I(X 1_{S^c})=S^c,\quad X\in\mathcal X_S^-.
\end{align}
By arguments similar to those above, we can
extend $\psi_S$ to $B_0(\Omega,\mathcal F)$ by defining
\begin{align*}
	\widehat\psi_S(X)=\psi_S(X-m_X-1)+(m_X+1) (1-a_S),~~{\rm with}~m_X=\inf\{m:X-m\in\mathcal X_{S}^-\},
\end{align*}
for $X\in B_0(\Omega,\mathcal F)$.
Similarly, one can check that $\widehat\psi_S$ is monotone and additive. Hence, there exists a measure $Q_S$ on $(\Omega,\mathcal F)$ such that $\widehat\psi_S=\E^{Q_S}$, and this implies
\begin{align*}
	\psi_S(X)=I(X1_{S^c})=\E^{Q_S}[X],~~\forall X\in\X_S^-.
\end{align*}
Similarly,
we can also conclude that $Q_S(S)=0$ and $Q_{S}(S^c)>0$ for $S\in\mathcal F \setminus\{\Omega,\varnothing\}$.
We now assume that $X\in\mathcal X_S$. Choose large enough $\eta>0$ such that $X-\eta\in\mathcal X_S^-$ and $X+\eta\in\mathcal X_S^+$. By the definition of $\mathcal X_S$ in \eqref{eq-X_S}, and combining with \eqref{eq-disX_S^+} and \eqref{eq-disX_S^-},  we have that $S^c=D_I(Z)$ for all $Z$ to be the following acts:
\begin{align*}
X,~-\eta1_{S^c},~\eta1_{S},~(X-\eta)1_{S^c},~(X+\eta)1_{S},~-\eta1_{S^c}+\eta1_{S}.
\end{align*}
Hence,
the following equality chain holds:
\begin{align*}
-\E^{Q_S}[\eta]+\E^{P_S}[\eta]
&=I(X)+\psi_S(-\eta)+\phi_S(\eta)\\
&=I(X)+I(-\eta1_{S^c})+I(\eta1_S)
=I((X-\eta)1_{S^c}+(X+\eta)1_S)\\
&=I((X-\eta)1_{S^c})+I((X+\eta)1_S)
=\psi_S(X-\eta)+\phi_S(X+\eta)\\
&=\E^{Q_S}[X]+\E^{P_S}[X]-\E^{Q_S}[\eta]+\E^{P_S}[\eta]\\
&=-\E^{Q_S}[X_-]+\E^{P_S}[X_+]-\E^{Q_S}[\eta]+\E^{P_S}[\eta],
\end{align*}
where the first equality follows from $I(X)=0$ as $X\in\mathcal X_S$, and we have used additivity in concordant sums of $I$ in the third and fourth equalities, and the last equality holds because $P_S(S^c)=Q_S(S)=0$ and $\{X\ge 0\}=S$.
This completes the step (a).\\
\underline{(b)} From the step (a), we have concluded that
\begin{align}\label{eq-PSQS}
\E^{P_S}[X_+]=\E^{Q_S}[X_-],~~X\in\X_S,~S\in\mathcal F,
\end{align}
where $P_S(S)>0$ and $Q_S(S^c)>0$ if $S\in\mathcal F\setminus \{\Omega,\varnothing\}$.
Fix $S\in\mathcal F$ with $|S|\ge 2$, let $\theta_1,\theta_2\in S$ with $\theta_1\neq \theta_2$, and we aim to verify $P_S(\theta_1)=P_S(\theta_2)$.
Let $X\in\mathcal X_S$ satisfying $X(\theta_1)\neq X(\theta_2)$.
Define $Y=X1_{\Omega\setminus\{\theta_1,\theta_2\}}+X(\theta_1)1_{\theta_2}+X(\theta_2)1_{\theta_1}$. Note that strong monotonicity implies law invariance, which further yields $I(Y)=I(X)=0$ as $P_Y=P_X$. By the definition of $\mathcal X_S$ and noting that $\theta_1,\theta_2\in S$, we have $X(\theta_1),X(\theta_2)\ge 0$, and hence, $\{Y\ge 0\}=\{X\ge 0\}=S$. Therefore, we have concluded that $Y\in\mathcal X_S$.
Substituting both $X$ and $Y$ into \eqref{eq-PSQS}, we have
\begin{align*}
&\sum_{\omega\in S\setminus\{\theta_1,\theta_2\}}X(\omega)P_S(\omega)
+P_S(\theta_1)X(\theta_1)+P_S(\theta_2)X(\theta_2)=\sum_{\omega\in S^c}|X(\omega)|Q_S(\omega);\\
&\sum_{\omega\in S\setminus\{\theta_1,\theta_2\}}X(\omega)P_S(\omega)
+P_S(\theta_1)X(\theta_2)+P_S(\theta_2)X(\theta_1)=\sum_{\omega\in S^c}|X(\omega)|Q_S(\omega).
\end{align*}
This yields
\begin{align*}
P_S(\theta_1)(X(\theta_1)-X(\theta_2))=P_S(\theta_2)(X(\theta_1)-X(\theta_2)),
\end{align*}
and thus, $P_S(\theta_1)=P_S(\theta_2)$.
Therefore, there exists $\lambda_S>0$ such that $P_S(\omega)=\lambda_S$ for all $\omega\in S$.
Following similar arguments, we can also verify that there exists $\eta_S>0$ such that $Q_S(\omega)=\eta_S$ for all $\omega\in S^c$. Suppose now that $|S_1|=|S_2|$. Let $X=a1_{S_1}-b 1_{S_1^c}$ satisfying $I(X)=0$ with some $a,b>0$. It is immediate to get $X\in \mathcal X_{S_1}$.
Define $Y=a1_{S_2}-b1_{S_2^c}$. By law invariance of $I$,  we have $I(Y)=I(X)=0$, and thus, $Y\in\mathcal X_{S_2}$.
Substituting $X$ and $Y$ into \eqref{eq-PSQS}, we have
\begin{align*}
&a|S_1|\lambda_{S_1}=aP_{S_1}(S_1)=bQ_{S_1}(S_1^c)=b|S_1^c|\eta_{S_1};\\
&a|S_2|\lambda_{S_2}=aP_{S_2}(S_2)=bQ_{S_2}
(S_2^c)=b|S_2^c|\eta_{S_2}.
\end{align*}
Combining with $|S_1|=|S_2|$ yields $\lambda_{S_1}/\eta_{S_1}=\lambda_{S_2}/\eta_{S_2}$, and this completes the proof of the step (b).\\
\underline{(c)}
From step (b), we have shown that $P_S(\omega)=\lambda_S$ for all $\omega\in S$ and
$Q_{S}(\omega)=\eta_S$ for all $\omega\in S^c$, and  $S\mapsto\lambda_S/\eta_S$ is a constant mapping on the sets $\{S\in\mathcal F:|S|=k\}$ for each $k\in[n]$. In this step, we aim to verify that $S\mapsto\lambda_S/\eta_S$ is a constant mapping on $\mathcal F$. To see this, it suffices to check that $\lambda_{S_1}/\eta_{S_1}=\lambda_{S_2}/\eta_{S_2}$
for $S_1=\{\omega_1,\dots,\omega_s\}$ and $S_2=\{\omega_1,\dots,\omega_t\}$ with $1\le s<t\le n-1$. We first verify that for sufficiently small $\epsilon\ge 0$, there is a unique positive root of the mapping:
\begin{align}\label{eq-20261023-3}
t\mapsto I(1_{S_1}-\epsilon1_{S_2\setminus S_1}-t1_{S_2^c}).
\end{align}
At $\epsilon=0$ and $t=0$ the value of the mapping is positive.
As $t\to\infty$, positive homogeneity and continuity give
\[
 \frac{I(1_{S_1}-t1_{S_2^c})}{t}
 \longrightarrow I(-1_{S_2^c})<0.
\]
The value of the mapping is continuous and strictly decreasing, so a unique
positive root exists. For sufficiently small $\epsilon>0$, the value
at $t=0$ remains positive and the same argument gives a unique positive root. Denote by $f(\epsilon)$ the root of the mapping \eqref{eq-20261023-3}, and 
define
\begin{align*}
X_\epsilon=1_{S_1}-\epsilon1_{S_2\setminus S_1}-f(\epsilon)1_{S_2^c},~~\epsilon\ge 0,
\end{align*}
and thus, $I(X_{\epsilon})=0$. Indeed, if $\epsilon=0$, then $X_\epsilon\in \mathcal X_{S_2}$, and applying \eqref{eq-PSQS} yields
\begin{align}\label{eq-f(0)}
f(0)=P_{S_2}(S_1)/Q_{S_2}(S_2^c).
\end{align}
For  $0<\epsilon<P_{S_1}(S_1)/Q_{S_1}(S_2\setminus S_1)$, we have $X_{\epsilon}\in \mathcal X_{S_1}$, and using \eqref{eq-PSQS} yields
\begin{align}\label{eq-f(epsilon)}
f(\epsilon)=\frac{P_{S_1}(S_1)-\epsilon Q_{S_1}(S_2\setminus S_1)}{Q_{S_1}(S_2^c)}>0.
\end{align}
Because $I$ satisfies strong monotonicity, the function $\epsilon\mapsto f(\epsilon)$ is strictly decreasing. We claim that $\lim_{\epsilon\downarrow0}f(\epsilon)=f(0)$.
Indeed,
\begin{align}\label{eq-limf(epsilon)}
I\left(1_{S_1}-f(0)1_{S_2^c}\right)=I(X_0)=0=\lim_{\epsilon\downarrow 0}I(X_\epsilon)=I\left(1_{S_1}-\lim_{\epsilon\downarrow0}f(\epsilon)1_{S_2^c}\right),
\end{align}
where the last equality holds because strong monotonicity and constant additivity together imply  $L^\infty$-norm continuity of $I$, and $X_{\epsilon}\to 1_{S_1}-\lim_{\epsilon\downarrow0}f(\epsilon)1_{S_2^c}$ uniformly.
Combining \eqref{eq-limf(epsilon)} with strong monotonicity of $I$ yields $\lim_{\epsilon\downarrow0}f(\epsilon)=f(0)$.
Therefore,
\begin{align*}
\frac{s\lambda_{S_2}}{(n-t)\eta_{S_2}}=\frac{P_{S_2}(S_1)}{Q_{S_2}(S_2^c)}=f(0)=\lim_{\epsilon\downarrow 0}f(\epsilon)=\frac{P_{S_1}(S_1)}{Q_{S_1}(S_2^c)}=\frac{s\lambda_{S_1}}{(n-t)\eta_{S_1}},
\end{align*}
where the second and the fourth equalities follow from \eqref{eq-f(0)} and \eqref{eq-f(epsilon)}, respectively.
This concludes that $S\mapsto \lambda_S/\eta_S$ is a constant on $\mathcal F$, which is denoted as $1/(1+\beta)$. Since it is positive, we have $\beta>-1$.
For any $X\in\mathcal X^0$, using the representation \eqref{eq-PSQS} implies that
\begin{align*}
\E[X^+]=(1+\beta)\E[X^-].
\end{align*}
For any $X\in B_0(\Omega,\mathcal F)$, note that $X-I(X)\in \mathcal X^0$ by translation invariance of $I$, and thus,
\begin{align*}
\E[(X-I(X))^+]=(1+\beta)\E[(I(X)-X)^+],
\end{align*}
which means that $I=\E_{\beta}$ with $\beta>-1$. This completes the proof of the case of finite $\Omega$.\\
\textbf{Proof for the case of infinite $\Omega$.}
We now assume that $(\Omega, \mathcal F, P)$ is a nonatomic probability space. The proof builds upon the result established for the finite case and proceeds by applying standard convergence arguments.
We begin by introducing some notation from Maccheroni et al.~(2025), which will be used in the proof.
 Denote by $q_X$ the quantile function of $X$ under $P$, i.e., $q_X(\alpha)=\inf\{x \in \R: P(X\le x)\ge \alpha\}$ for $\alpha \in (0,1]$. Also write $q_X(0)= \inf\{x \in \R: P(X\le x)>0\}$.
 Let $V$ be a random variable with uniform distribution on $(0,1)$, i.e., $P(V\le x)=x$ for $x\in[0,1]$, such that $X=q_X(V)$ almost surely (for the existence of $V$, see Lemma A.32 of F\"{o}llmer and Schied, 2016).
Define
\begin{align*}
\Psi_{k}=\left\{  \left(  \frac{0}{2^{k}},\frac{1}{2^{k}}\right]  ,\left(
\frac{1}{2^{k}},\frac{2}{2^{k}}\right]  ,\dots,\left(  \frac{2^{k}-1}{2^{k}%
},\frac{2^{k}}{2^{k}}\right]  \right\},~~k\in\N
\end{align*}
as the partition of $\left( 0,1\right]  $ into segments of equal length $2^{-k}$. Further define
\begin{align*}
\Pi_{k}=V^{-1}\left(\Psi_{k}\right),~~k\in\N
\end{align*}
as a partition of $\Omega$ in $\mathcal F$ such that $P(E)=1/2^{k}$ for all $E\in
\Pi_{k}$. By setting $\mathcal F_{k}=\sigma(  \Pi
_{k})  =V^{-1}\left(  \sigma\left(  \Psi_{k}\right)
\right)  $ for all $k\in\mathbb{N}$, we have a filtration $\left\{  \mathcal F
_{k}\right\}_{k\in\mathbb{N}}$ in $\mathcal F$.
Denote by $P|_{\mathcal F_k}$ the restriction of $P$ on $\mathcal F_k$.
Let $B_0(\Omega,\mathcal F_k)$ denote the set of all $\mathcal F_k$-measurable simple random variables, and let $I|_{\mathcal F_k}$ represent the restriction of the functional $I$ to the domain $B_0(\Omega,\mathcal F_k)$.
It is straightforward to verify that $I|_{\mathcal F_k}$ satisfies  strong monotonicity, constant additivity, positive homogeneity and additivity in concordant sums on $B_0(\Omega,\mathcal F_k)$. By the result for finite case, we have $I|_{\mathcal F_k}=\E_{\beta_k}$ with $\beta_k> -1$ for all $k\in\N$. Since $\left\{  \mathcal F
_{k}\right\}_{k\in\mathbb{N}}$ is a filtration, we have
$$
\E_{\beta_k}[X]=I|_{\mathcal F_{k}}(X)=I|_{\mathcal F_{1}}(X)=\E_{\beta_1}[X]
$$
for any $X\in B_0(\Omega,\mathcal F_1)$. We claim that $\beta_k=\beta_1$ for all $k\in\N$.
Denote $\Pi_1=\{E_1,E_2\}$ with $P(E_1)=P(E_2)=1/2$. Let $X=1_{E_1}$, and denote $t=\E_{\beta_1}[X]=\E_{\beta_k}[X]$. It holds that $t\in(0,1)$ by strong monotonicity, and
\begin{align*}
1-t=(1+\beta_1)t\quad {\rm and}\quad 1-t=(1+\beta_k)t
\end{align*}
implying $\beta_k=\beta_1$.
Therefore, the following representation holds:
\begin{align}\label{eq-IV}
I(X)=\E_{\beta}[X],~~X\in\bigcup_{k\in\N} B_0(\Omega,\mathcal F_k),
\end{align}
where we denote $\beta=\beta_k$ for $k\in\N$.
Below, we aim to verify that the representation  holds for all acts in $B_0(\Omega,\mathcal F)$.
We first recall that strong monotonicity implies law invariance.
Let $\underline{V}_k = 2^{-k} \lceil 2^{k}V -1\rceil  $ and $\overline{V}_k = 2^{-k} \lceil 2^{k}V\rceil $
for $k\in \N$, where $\lceil x \rceil$ represents the least integer not less than $x$; that is, $(\underline{V}_k,\overline{V}_k]$ is the interval in $\Psi_k$ that contains $V$.
For $X\in B_0(\Omega,\mathcal F)$, denote by
$\underline X_k=q_X(\underline{V}_k)$ and $\overline X_k= q_X(\overline{V}_k)$.  It is straightforward to check that $\underline X_k, \overline X_k\in \bigcup _{k\in\N}B_0(\Omega,\mathcal F_k)$ and
$\underline X_k \le X \le \overline X_k$. By strong monotonicity, we have
$$
I(\underline X_k ) \le I(X) \le I( \overline X_k),
$$
and combining with \eqref{eq-IV} implies
\begin{align}\label{eq:sandwich}
\E_{\beta}[\underline X_k ] \le I(X) \le \E_{\beta}[\overline X_k].
\end{align}
Since $\underline X_k$ and $\overline X_k$ both converge to $q_X(V)$ bounded a.s., we have
\begin{align*}
\E_{\beta}[\underline X_k ], \E_{\beta}[\overline X_k ]\to \E_{\beta}[q_X(V)]=\E_{\beta}[X],
\end{align*}
where the convergence follows from the continuity property in Lemma \ref{lm-propertyEx}, and the equality is due to law invariance.
Combining with \eqref{eq:sandwich} yields the representation of $\E_{\beta}$ on $B_0(\Omega,\mathcal F)$. This completes the proof.
\end{proof}
\begin{proof}[Proof of Theorem \ref{th-characterizationEx1}.]
\textbf{Sufficiency.} Suppose that $I=\E_{\beta}$ with $\beta\ge 0$. By Lemma \ref{lm-propertyEx},
strong monotonicity and $L^\infty$-norm continuity hold. To see decrease in concordant sums, let $X, Y, W \in B_0(\Omega,\mathcal F)$ be such that $\E_{\beta}[X] = \E_{\beta}[Y]$ and
$D_{\E_{\beta}}(X)=D_{\E_{\beta}}(W)$, where $D_{\E_{\beta}}(X)$ is defined by \eqref{eq-D_I} with the form:
\begin{align*}
D_{\E_\beta}(X)=\{\omega:X(\omega)<\E_{\beta}[X]\}.
\end{align*}
It holds that
\begin{align*}
\E_{\beta}[W+X]=\E_{\beta}[W]+\E_{\beta}[X]=\E_{\beta}[W]+\E_{\beta}[Y]\le \E_{\beta}[W+Y],
\end{align*}
where the first equality and the inequality follow from additivity in concordant sums and superadditivity in Lemma \ref{lm-propertyEx}, respectively.\\
\textbf{Necessity.} Suppose that $I$ satisfies strong monotonicity, $L^\infty$-norm continuity, and decrease in concordant sums with $I(m)=m$ for all $m\in\R$. We aim to show that these properties together imply constant additivity, superadditivity, additivity in concordant sums, and positive homogeneity.
Consequently, by Theorem~\ref{th-characterizationEx}, we conclude that $I = \E_{\beta}$ for some $\beta > -1$. Moreover, superadditivity implies that $\beta \ge 0$, as established in Lemma~\ref{lm-propertyEx}. Below, we proceed to verify each of the required properties in turn.
\underline{Constant additivity.}
The case that $X\in B_0(\Omega,\mathcal F)$ is a constant is trivial as $I(m)=m$ for all $m\in\R$.
For a nonconstant $X\in B_0(\Omega,\mathcal F)$, we have
\begin{align*}
I(X+m)\ge I(I(X)+m)=I(X)+m,
\end{align*}
where the inequality follows from decrease in concordant sums and the fact that $I(I(X))=I(X)$ and $D_I(I(X))=D_I(m)=\varnothing$.
On the other hand,
for $\epsilon> 0$, define $Y_{\epsilon}=m1_{D_I(X)}+(m+\epsilon)1_{D_I(X)^c}$. Strong monotonicity implies $I(Y_{\epsilon})\in (I(m),I(m+\epsilon))=(m,m+\epsilon)$, and thus, $D_I(Y_{\epsilon})=D_I(X)$. Decrease in concordant sums yields
\begin{align*}
	I(X+Y_\epsilon)\le I(I(X)+Y_\epsilon)~~\forall \epsilon>0.
\end{align*}
Note that $\esssup |Y_\epsilon-m|\to0$.
Letting $\epsilon\to 0$ in the above equation, and using $L^\infty$-norm continuity of $I$ implies
\begin{align*}
I(X+m)\le I(I(X)+m)=I(X)+m.
\end{align*}Hence, we have concluded that $I(X+m)=I(X)+m$ for all $m\in\R$, and constant additivity holds.
\underline{Superadditivity.}
For $X, Y \in B_0(\Omega,\mathcal F)$, if either $X$ or $Y$ is constant, then constant additivity immediately yields $I(X + Y) = I(X) + I(Y)$.
Now, consider the case where both $X$ and $Y$ are nonconstant.
Define $X_{\epsilon}=\epsilon 1_{D_I(Y)^c}$ for $\epsilon>0$. Strong monotonicity implies $D_I(X_{\epsilon})=D_I(Y)$. Note that
$$
I(X-I(X)+I(X_{\epsilon}))=I(X)-I(X)+I(X_\epsilon)=I(X_\epsilon),
$$
where we have used constant additivity in the first equality. Therefore,
\begin{align*}
I(X_{\epsilon}+Y)\le I(X-I(X)+I(X_{\epsilon})+Y)=I(X+Y)+I(X_\epsilon)-I(X),
\end{align*}
where the inequality follows from decrease in concordant sums and the equality is due to constant additivity. Letting $\epsilon\downarrow 0$ in the above equation, and using $L^\infty$-norm continuity of $I$ yields
\begin{align*}
I(Y)\le I(X+Y)-I(X).
\end{align*}
This gives superadditivity.
\underline{Additivity in concordant sums.} Suppose that $X,Y\in B_0(\Omega,\mathcal F)$ satisfy $D_I(X)=D_I(Y)$. By decrease in concordant sums and constant additivity, we have $I(X+Y)\le I(I(X)+Y)=I(X)+I(Y)$. Combining with superadditivity yields additivity in concordant sums.
\underline{Positive homogeneity.}
For $X \in B_0(\Omega,\mathcal F)$, $D_I(X) = D_I(X)$ holds trivially.
By additivity in concordant sums, we have $I(2X) = 2I(X)$, which implies $D_I(X) = D_I(2X)$.
Applying additivity in concordant sums again, we obtain $I(3X) = 3I(X)$.
By iterating this argument, it follows that $I(\lambda X) = \lambda I(X)$ for all rational $\lambda > 0$.
Since $I$ satisfies $L^\infty$-norm continuity, this identity can be extended to all real $\lambda \ge 0$ via a standard convergence argument.
\end{proof}
\begin{proof}[Proof of Theorem \ref{th-characterization_general}.]
The proof of sufficiency is similar to that of Theorem \ref{th-characterizationEx1}, and follows directly from Lemma~\ref{lm-propertyEx}. We now consider the necessity. Suppose that $I:B_0(\Omega,\mathcal F, A)\to \R$ satisfies strong monotonicity, $L^\infty$-norm continuity, and decrease in concordant mixtures with $I(m)=m$ for all $m\in A$.
The proof proceeds by verifying the following claims:
\begin{enumerate}[(a)]
\item $I\left(\frac{X+m}{2}\right)=\frac{1}{2}(I(X)+m)$ for all $X\in B_0(\Omega,\mathcal F, A)$ and $m\in A$.
\item $I(\lambda X)=\lambda I(X)$ for all $\lambda\in[0,1]$ and $X\in B_0(\Omega,\mathcal F, A)$. Because $0$ is in the interior of $A$, the conditions $\lambda\in[0,1]$ and $X\in B_0(\Omega,\mathcal F,A)$ imply $\lambda X\in B_0(\Omega,\mathcal F,A)$.
\item The functional $I$ can be extended to $B_0(\Omega,\mathcal F)$ in such a way that the extension satisfies strong monotonicity, $L^\infty$-norm continuity, and decrease in concordant sums, thereby allowing the desired result to follow from Theorem~\ref{th-characterizationEx1}.
\end{enumerate}
\underline{(a)}
For $X\in B_0(\Omega,\mathcal F,A)$ and $m\in A$, it is trivial for the case that $X$ is a constant. Suppose now that $X$ is nonconstant. It holds that
\begin{align}\label{eq-cxTI1}
\frac{I(X)+m}{2}=I\left(\frac{I(X)+m}{2}\right)\le I\left(\frac{X+m}{2}\right),
\end{align}
where the inequality follows from decrease in concordant mixtures by noting that
$I(I(X))=I(X)$ and $D_I(I(X))=D_I(m)=\varnothing$, where we recall that $D_I$ is defined by \eqref{eq-D_I} with the form:
\begin{align*}
D_I(Z)=\{\omega: Z(\omega)<I(Z)\},\quad Z\in B_0(\Omega,\mathcal F,A).
\end{align*}
On the other hand,
strong monotonicity implies $I(X)\in(\essinf X,\esssup X)$, and thus, $P(D_I(X))>0$.
Define $Y_{\epsilon}=m-\epsilon 1_{D_I(X)}$ if $m$ is the right endpoint of $A$, and $Y_\epsilon=m+\epsilon 1_{D_I^c(X)}$ otherwise.
Applying strong monotonicity again, it is straightforward to verify that  $D_I(Y_{\epsilon})=D_I(X)$ for all $\epsilon>0$ such that $Y_\epsilon\in B_0(\Omega,\mathcal F,A)$. Hence, we have
\begin{align}\label{eq-cxTI2}
I\left(\frac{X+m}{2}\right) \overset{\epsilon\downarrow 0}{\longleftarrow}I\left(\frac{X+Y_\epsilon}{2}\right)\le I\left(\frac{I(X)+Y_\epsilon}{2}\right)\overset{\epsilon\downarrow 0}{\longrightarrow}I\left(\frac{I(X)+m}{2}\right)=\frac{I(X)+m}{2},
\end{align}
where the convergences are due to the $L^\infty$-norm continuity of $I$ and the inequality follows from decrease in concordant mixtures. Hence, statement (a) holds by combining \eqref{eq-cxTI1} and \eqref{eq-cxTI2}.\\
\underline{(b)}
We first consider a result as follows:
\begin{align}\label{eq-cxPH}
I\left(\frac{k}{2^n}X\right)=\frac{k}{2^n}I(X)~~{\rm for~all}~~X\in B_0(\Omega,\mathcal F,A),~n\in\N,~k\in[2^n].
\end{align}
We use induction to prove this conclusion. We only consider the nontrivial case that $X$ is not a constant.
The cases of $(n,k)\in\{(0,1),(1,2)\}$ hold directly. The case of $(n,k)=(1,1)$ follows immediately from statement (a) with $m=0$.
Assume now that \eqref{eq-cxPH} holds for $(n-1,k)$ with $k\in[2^{n-1}]$ and $(n,k)$ with $k\in[2^{n}]$. Let us consider the situations that $(n+1,k)$ with $k\in [2^{n+1}]$.
If $k\in [2^{n+1}]$ is even, then we have
\begin{align*}
I\left(\frac{k}{2^{n+1}}X\right)=I\left(\frac{k/2}{2^{n}}X\right)
=\frac{k/2}{2^n}I(X)=\frac{k}{2^{n+1}}I(X),
\end{align*}
where we have used the inductive assumption in the second equality. 
Notice that \eqref{eq-cxPH} in the case of $k=1$    follows from a repeated application of part (a) with $m=0$.

For odd $k\in [2^{n+1}]$ and $k\ne 1$, let
\[
 \lambda_1=\frac{k-1}{2^{n+1}},\qquad
 \lambda_2=\frac{k+1}{2^{n+1}},\qquad
 \theta=\frac{\lambda_1+\lambda_2}{2}=\frac{k}{2^{n+1}}.
\]
Here $0<\lambda_1<1$ and $0<\lambda_2\le1$. We aim to prove $I(\theta X)=\theta I(X)$ for nonconstant $X\in B_0(\Omega,\mathcal F,A)$.
Suppose $X\in B_0(\Omega,\mathcal F,A)$, and let
$D=D_I(X)$ and $c=I(\lambda_1X)$. Strong monotonicity and the induction
hypothesis give
\[
 c<\esssup(\lambda_1X),\quad
 0<P(D)<1,\quad
 D_I(\lambda_1X)=D_I(\lambda_2X)=D.
\]
For
$0<\epsilon<\esssup(\lambda_1X)-c$, set
\[
 Y_\epsilon=c+\epsilon1_{D^c}.
\]
Both of its values lie in $A$, since $A$ is an interval containing $c$
and $\esssup(\lambda_1X)$. Strong monotonicity implies
$c<I(Y_\epsilon)<c+\epsilon$ and $D_I(Y_\epsilon)=D$.
All values of $\lambda_1X$ lie in the interior of $A$ because $X$ has finite
range, $0\in\operatorname{int}A$, and $0<\lambda_1<1$.
Choose $\delta>0$ such that $\lambda_1X+t$ takes values in $A$ for
$0\le t\le\delta$. The function
\[
 f(t)=I(\lambda_1X+t),\qquad 0\le t\le\delta,
\]
is continuous and strictly increasing, with $f(0)=c$. Since
$I(Y_\epsilon)\downarrow  c$ as $\epsilon\to 0$, for every sufficiently small $\epsilon$
there is a unique $\eta_\epsilon\in(0,\delta)$ with
\[
f(\eta_{\epsilon})= I(\lambda_1X+\eta_\epsilon)=I(Y_\epsilon).
\]
Moreover, it is straightforward to verify  $\eta_\epsilon\to0$ as $\epsilon\to 0$ by the continuity and strict monotonicity of $f$.
Decrease in concordant mixtures now gives
\begin{align}\label{eq-20261023-1}
 I\left(\frac{Y_\epsilon+\lambda_2X}{2}\right)
 \le I\left(\frac{\lambda_1X+\eta_\epsilon+\lambda_2X}{2}\right).
\end{align}
Therefore,
\begin{align}\label{eq-20261023-2}
 \theta I(X)
 =\frac{I(\lambda_1 X)+I(\lambda_2 X)}{2}
 =I\left(\frac{c+\lambda_2X}{2}\right)
 \le I(\theta X),
\end{align}
where the first equality holds because we have shown in the even case that $I(\lambda_i X)=\lambda_i I(X)$ for $i=1,2$,
the second equality follows from statement (a), and the inequality holds by letting $\epsilon\downarrow0$ in \eqref{eq-20261023-1} and
using the continuity of $I$. On the other hand, we have
\begin{align*}
I(\theta X)\le I\left(\frac{I(\lambda_1 X)+\lambda_2 X}{2}\right)=\frac{I(\lambda_1 X)+I(\lambda_2 X)}{2}=\theta I(X),
\end{align*}
where the inequality is due to decrease in concordant mixtures, and the equalities follows from the same reasons to those in \eqref{eq-20261023-2}. Hence, $I(\theta X)=\theta I(X)$, that is \eqref{eq-cxPH} holds for $(n+1,k)$ with odd $k\in[2^{n+1}]$.
This completes the proof of \eqref{eq-cxPH}. Since $\{k/2^n:~n\in\N,~k\in[2^n]\}$ is a dense subset of $[0,1]$, this combined with the $L^\infty$-norm continuity of $I$ yields $I(\lambda X)=\lambda I(X)$ for all $X\in B_0(\Omega,\mathcal F,A)$  and $\lambda\in[0,1]$.
\\
\underline{(c)}
Define
$$
\widetilde{I}(X)=\frac{1}{\lambda_X}I\left({\lambda_X X}\right)~{\rm for}~X\in B_0(\Omega,\mathcal F),
$$
where $$\lambda_X:=\frac{1}{2}\sup\{\lambda\in[0,1]: \lambda X\in B_0(\Omega,\mathcal F,A)\}.
$$
We clarify that $\lambda_X$ is defined as half the value of the supremum problem above, rather than the full value, in order to avoid situations where $\lambda_X X \notin B_0(\Omega, \mathcal{F}, A)$ when $A$ is an open set.
We aim to prove that $\widetilde{I}$ is an extension of $I$ on $B_0(\Omega,\mathcal F)$ satisfying strong monotonicity, $L^\infty$-norm continuity and decrease in concordant sums with $\widetilde{I}(m)=m$ for all $m\in \R$.
If $X\in B_0(\Omega,\mathcal F,A)$, it is clear that $\lambda_X=1/2$, and statement (b) implies $\widetilde{I}(X)=I(X)$, which means that $\widetilde{I}$ is an extension of $I$. Let $X\in B_0(\Omega,\mathcal F)$, and we claim that
\begin{align}\label{eq-extension1}
\widetilde{I}(X)=\frac{1}{\lambda} I(\lambda X)~{\rm whenever}~\lambda\in(0,1]~{\rm and}~\lambda X\in B_0(\Omega,\mathcal F,A).
\end{align}
If $\lambda\ge \lambda_X$, then denote by $\theta=\lambda_X/\lambda\in[0,1]$, and we have
\begin{align*}
\frac{1}{\lambda} I(\lambda X)=\frac{1}{\lambda} \left(\frac{1}{\theta} I\left(\theta(\lambda X)\right)\right)=\frac{1}{\lambda_X}I\left({\lambda_X X}\right)=\widetilde{I}(X).
\end{align*}
If $\lambda \le \lambda_X$, a similar argument holds with $\lambda$ and $\lambda_X$ interchanged. Hence, \eqref{eq-extension1} holds. For $m\in\R$, let $\lambda\in(0,1]$ be such that $\lambda m\in A$, and we have $\widetilde{I}(m)=I(\lambda m)/\lambda=m$. For $X,Y\in B_0(\Omega,\mathcal F)$ with $P(X\ge t)\ge P(Y\ge t)$ for all $t\in\R$, let $\lambda\in(0,1]$ be such that $\lambda X, \lambda Y\in B_0(\Omega,\mathcal F,A)$, and we have
\begin{align*}
\widetilde{I}(X)=\frac{I(\lambda X)}{\lambda}\ge \frac{I(\lambda Y)}{\lambda}=\widetilde{I}(Y),
\end{align*}
where the inequality follows from strong monotonicity of $I$ and $P(\lambda X\ge t)\ge P(\lambda Y\ge t)$ for all $t\in\R$. 
If, in addition, $P(X\ge t_0)>P(Y\ge t_0)$ for some $t_0\in\mathbb R$, then the above inequality is strict.
Hence, $\widetilde{I}$ satisfies strong monotonicity. For $\{X_n\}_{n\in\N}\subseteq B_0(\Omega,\mathcal F)$ and $X\in B_0(\Omega,\mathcal F)$ with $X_n\to X$ in $L^\infty$, there exists $\lambda\in(0,1]$ such that $\lambda X\in B_0(\Omega,\mathcal F,A)$ and $\lambda X_n\in B_0(\Omega,\mathcal F,A)$ for all $n\in \N$, and hence,
\begin{align*}
\widetilde{I}(X_n)=\frac{I(\lambda X_n)}{\lambda}\to \frac{I(\lambda X)}{\lambda}=\widetilde{I}(X),
\end{align*}
where the convergence follows from the $L^\infty$-norm continuity of $I$. This yields the $L^\infty$-norm continuity of $\widetilde{I}$.
Finally, let $W,X,Y\in B_0(\Omega,\mathcal F)$ with $\widetilde{I}(X)=\widetilde{I}(Y)$ and $D_{\widetilde{I}}(W)=D_{\widetilde{I}}(X)$. There exists $\lambda\in(0,1]$ such that $\lambda W,\lambda X,\lambda Y, \lambda(W+X), \lambda(W+Y)\in B_0(\Omega,\mathcal F,A)$, and we have
\begin{align}\label{eq-extensionDC1}
\frac{I(\lambda X)}{\lambda}=\widetilde{I}(X)=\widetilde{I}(Y)= \frac{I(\lambda Y)}{\lambda}
\end{align}
and
\begin{align}\label{eq-extensionDC2}
D_I(\lambda W)&=\left\{\omega: W(\omega)< \frac{I(\lambda W)}{\lambda}\right\}=\left\{\omega: W(\omega)< \widetilde{I}(W)\right\}=D_{\widetilde{I}}(W)\notag\\
&=D_{\widetilde{I}}(X)
=\left\{\omega: X(\omega)< \widetilde{I}(X)\right\}=\left\{\omega: X(\omega)< \frac{I(\lambda X)}{\lambda}\right\}=D_I(\lambda X).
\end{align}
Therefore,
\begin{align*}
\widetilde{I}(W+X)=\frac{I(\lambda W/2+\lambda X/2)}{\lambda/2}\le \frac{I(\lambda W/2+\lambda Y/2)}{\lambda/2} = \widetilde{I}(W+Y),
\end{align*}
where the inequality follows from the property of decrease in concordant mixtures, along with the relations in \eqref{eq-extensionDC1} and \eqref{eq-extensionDC2}. This completes the proof.
\end{proof}
\section{Proofs of results in Section \ref{sec:expectiled-utility}}
\begin{proof}[Proof of Proposition \ref{prop:expectiles}.]
The solution to \eqref{eq:DA} is an expectile and is unique; see, e.g.~Bellini et al.~(2014). Next, we aim to verify that a solution to Gul's equation \eqref{eq:Gulagain} is a solution of \eqref{eq:DA}, and conversely. Suppose that $v$ is a solution of  \eqref{eq:Gulagain}, and thus,
\begin{align*}
v=\E\left[\frac{u(X)+\beta v}{1+\beta}1_{\{u(X)\ge v\}}+u(X)1_{\{u(X)<v\}}\right].
\end{align*}
It is straightforward to verify that the two terms inside the bracket in the above equation can be reformulated as follows:
\begin{align*}
\frac{u(X)+\beta v}{1+\beta}1_{\{u(X)\ge v\}}=\frac{1}{1+\beta} \left((u(X)-v)^++(1+\beta)v1_{\{u(X)\ge v\}}\right)
\end{align*}
and
\begin{align*}
u(X)1_{\{u(X)<v\}}=-(v-u(X))^++v 1_{\{u(X)<v\}}.
\end{align*}
Therefore, $v$ is a solution of \eqref{eq:Gulagain} if and only if 
\[
v=\E\left[\frac{1}{1+\beta} \left((u(X)-v)^++(1+\beta)v1_{\{u(X)\ge v\}}\right)-(v-u(X))^++v 1_{\{u(X)<v\}}\right],
\]
equivalently, 
\[
v=\frac{1}{1+\beta}\E[(u(X)-v)^+]-\E[(v-u(X))^+]+v,
\]
that is, $v$ is a solution of \eqref{eq:DA}. This completes the proof.
\end{proof}
\begin{proof}[Proof of Theorem \ref{th:expectiles}.]
The result follows directly from Newey and Powell (1987) and Bellini et al.~(2014), where expectiles are characterized as solutions to asymmetric least squares minimization problems and as coherent risk measures under suitable conditions.
\end{proof}
\begin{proof}[Proof of Theorem \ref{th:GulNeweyPowellPlus}.]
For notational simplicity, we denote $U=u(X)$.
Define $\mathcal P$ as the set of all probability measures.
Note that $\beta\ge 0$. By the dual representation of $\E_{\beta}$ in Lemma \ref{lm-propertyEx}, we have
\begin{align*}
\E_{\beta}[U]=
\min_{Q\in\mathcal P^{\beta}} \E^Q[U],
\end{align*}
where
$$
\mathcal P^{\beta}=\left\{Q \in\mathcal P:  \frac{\operatorname{ess} \sup \d Q/\d P}{\operatorname{ess} \inf \d Q/\d P} \leq 1+\beta\right\},
$$
and the optimal probability measure $Q^*$  is given by
$$
\mathrm{d}Q^{\ast}=\frac{1_{D_{X  }^{c}}+(1+\beta
)1_{D_{X   }}}{1+\beta P\big(D_{X  }%
\big)}\mathrm{d}P ,
$$
and it is clear that $Q^{\ast} \in \mathcal Q^\beta$.
The conclusion follows from this and the observation  $\mathcal Q^{\beta} \subseteq \mathcal P^{\beta}$, which can be checked directly.  
\end{proof}
\section{Proofs  of results and related axioms in Section \ref{sec:expectiled-axioms}}
In this section, we begin by presenting the detailed axioms underlying invariant biseparable preferences  and some related results as introduced in GMMS.
\begin{aaxiom}[Weak Order]
\label{ax:wo} (a) For all $X,Y\in \mathbb{X}$, $%
X\succsim Y$ or $Y\succsim X$. (b) For all $X,Y,Z\in \mathbb{X} $, if $X\succsim Y$ and $Y\succsim Z$, then $X\succsim Z$.
\end{aaxiom}
\begin{aaxiom}[Dominance]
\label{ax:D} For every $X,Y\in \mathbb{X}$, if $X(\omega)\succsim Y(\omega)$ for all $\omega\in\Omega$, then $X\succsim Y$.
\end{aaxiom}
\begin{aaxiom}[Essentiality]
\label{ax:E}
There exists an event $E\in\mathcal F$ such that $x \succ xEy \succ y$ for some  consequences $x, y \in \mathcal{X}$. Such an event is called essential.
\end{aaxiom}
Given $E$, we denote by $\sigma(E)$ the algebra generated by $E$.
We use the following terminology: An event $A \in \mathcal F$ is null (resp.~universal) if $y \sim x A y$ (resp.~$x \sim x A y$ ) for every $x \succ y$. It follows from Axiom \ref{ax:wo} that an event can be only one of null, essential, or universal.
\begin{aaxiom}[$E$-Monotonicity]
\label{ax:EM}
For every non-null $A \in \sigma(E)$ and every $x, y \succsim z \in \X$,
$$
x \succ y \Longrightarrow x A z \succ y A z.
$$
For every non-universal $A \in \sigma(E)$ and every $x, y \precsim z \in X$,
$$
x \succ y \Longrightarrow z A x \succ z A y.
$$
\end{aaxiom}
\begin{aaxiom}[$E$-Continuity]
\label{ax:C}\label{ax:c}
Let $\{X_n\}_{n \in \N} \subseteq \mathbb{X}$ be a sequence of $\sigma(E)$-measurable acts that pointwise converges  to $X$. For every $Y \in \mathbb{X}$, if $X_n \succsim Y$ (resp.~$Y \succsim X_n$) for all $n\in\N $, then $X \succsim Y$ (resp.~$Y \succsim X$).
\end{aaxiom}
It is straightforward to show (see e.g.~Lemma 12 of GMMS) that any binary relation satisfying axioms \ref{ax:wo}--\ref{ax:E} and \ref{ax:c} has certainty equivalents. That is, for every $X \in \mathbb{X}$, there exists $x \in \X$ such that $x \sim X$. Granted this, we henceforth denote by $c_X$ an arbitrarily chosen certainty equivalent of $X \in \mathbb{X}$.
The next axiom imposes a behavioral restriction. We write $x \succsim\left\{z', z^{\prime \prime}\right\}$ (resp.~$\left\{z^{\prime}, z^{\prime \prime}\right\} \succsim$ $y)$ if $x \succsim z^{\prime}$ and $x \succsim z^{\prime \prime}\left(\right.$ resp.~$z^{\prime} \succsim y$ and $\left.z^{\prime \prime} \succsim y\right)$.
\begin{aaxiom}[$E$-Substitution]
\label{ax:ES}
For all $x, y, z^{\prime}, z^{\prime \prime} \in X$ and $A, B \in \sigma(E)$. Suppose that $x \succsim$ $\left\{z^{\prime}, z^{\prime \prime}\right\} \succsim y$. Then
$$
c_{x A z^{\prime}} B c_{z^{\prime \prime} A y} \sim c_{x B z^{\prime \prime}} A c_{z^{\prime} B y}
$$
\end{aaxiom}
The final axiom of GMMS uses the notion of mixture thus derived to impose a very weak and natural
property of separability of preferences. Before stating the axiom, we first introduce some necessary preliminaries.
\begin{definition}[Preference Average]\label{def:midpoint}
Given $x, y \in \X$ such that $x \succsim y$ (resp.~$y \succsim x$), the preference average of $x,y$ given the event $E$, denoted by $(1 / 2) x \oplus(1 / 2) y$, is
a consequence $z \in\X$ such that $x \succsim z \succsim y$ (resp.~$y \succsim z \succsim x$) and
$$
x E y \sim c_{x E z} E c_{z E y}~~~\left(\text {resp. } y E x \sim c_{y E z} E c_{z E x}\right).
$$
\end{definition}
We note that the term preference average used in the above definition corresponds to what is referred to as the preference midpoint in the main text of this paper.
Based on Axioms \ref{ax:wo}--\ref{ax:ES}, Lemma 1 of GMMS establishes that the DM’s preferences admit a canonical representation over the set of all $\sigma(E)$-measurable acts, characterized by a continuous and nonconstant canonical utility function $u$.
\begin{lemma}[Lemma 1 of GMMS]\label{lm-CR}
The binary relation $\succsim$ satisfies Axioms \ref{ax:wo}--\ref{ax:ES} if and only if there is a continuous nonconstant utility index $u: \X \rightarrow \mathbb{R}$ and a capacity $\rho_E: \sigma(E) \rightarrow[0,1]$, with $\rho_E(E) \in(0,1)$, such that the functional $V: \mathbb X \rightarrow \mathbb{R}$ defined by $V(X)=u\left(c_X\right)$ for any $X \in \mathbb X $ represents $\succsim$, it is $u$-monotone,\footnote{This property means: if $u(X)\geq u(Y)$ pointwise then $V(X)\geq V(Y)$.} and it satisfies, for all $x \succsim y$ and all $A \in \sigma(E)$,
$$
V(x A y)=u(x) \rho_E(A)+u(y)\left(1-\rho_E(A)\right) .
$$
Moreover, such $u$ and $V$ are unique up to a positive affine transformation and $\rho_E$ is unique.
\end{lemma}
Furthermore, Lemma 3 of GMMS shows that this utility function is additive with respect to the preference average operator. We present the corresponding result below.
\begin{lemma}[Lemma 3 of GMMS]\label{lm-midpoint}
Suppose that the binary relation $\succsim$ satisfies Axioms \ref{ax:wo}--\ref{ax:ES}.
For any $x,y\in\mathcal X$, there exists $z=(1/2)x\oplus(1/2)y$.
If $u$ is the cardinal utility that represents $\succsim$ by Lemma \ref{lm-CR}, then
\begin{align*}
u\left(\frac12 x \oplus\frac12 y\right)=\frac12u(x)+\frac12u(y).
\end{align*}
\end{lemma}
Let $w$ be a preference average of $x$ and $y$. Then, $z$ is a preference average of $x$ and $w$ if and only if $u(z)=(3 / 4) u(x)+(1 / 4) u(y)$; that is, $z$ is a $(3 / 4):(1 / 4)$ utility mixture of $x$ and $y$. This allows us to identify $(3 / 4) x \oplus(1 / 4) y$ behaviorally. Proceeding along these lines and using the continuity axiom (Axiom \ref{ax:c}), it is possible to identify behaviorally the $\alpha: 1-\alpha$ utility mixtures of $x$ and $y$, for any $\alpha \in[0,1]$. Specifically, Lemma 13 of GMMS shows that
\begin{align*}
u(\alpha x\oplus (1-\alpha)y)=\alpha u(x)+(1-\alpha)u(y),\quad\forall x,y\in\mathcal X,~\alpha\in[0,1].
\end{align*}
Subjective mixtures of acts may then be defined pointwise, as usual. That is, given $X, Y \in \mathbb{X}$ and $\alpha \in[0,1], \alpha X \oplus(1-\alpha) Y$ is the act $Z \in \mathbb{X}$ defined by $Z(\omega)=\alpha X(\omega) \oplus(1-\alpha) Y(\omega)$ for any $\omega \in \Omega$.
\begin{aaxiom}[Weak Certainty Independence]
\label{ax:WCI}
For all $X, Y \in \mathbb{X},~x \in \X$ and $\alpha\in[0,1]$,
$$
X \sim Y \Longrightarrow \alpha X \oplus(1-\alpha) x \sim \alpha Y \oplus(1-\alpha) x.
$$
\end{aaxiom}
So far, we have presented all the relevant axioms from GMMS. A binary relation that satisfies Axioms \ref{ax:wo}--\ref{ax:WCI} is referred to as \emph{invariant biseparable}.
Next, we present the representation result for invariant biseparable preference given by GMMS.
\begin{lemma}[Theorem 5 of GMMS]\label{lm-IBGMMS}
Let $\succsim$ be a binary relation on $\mathbb{X}$. Then $\succsim$ is invariant biseparable
if and only if there exist a continuous nonconstant function $u: \X \rightarrow \mathbb{R}$ and a monotone, constant-additive and positively homogeneous functional $I: B_0(\Omega,\mathcal F) \rightarrow \mathbb{R}$ such that for all $X, Y \in \mathbb{X}$,
$$
X \succsim Y \Longleftrightarrow I(u (X)) \geq I(u(Y))
$$
and such that $I\left(1_E\right) \notin\{0,1\}$ for some $E \in \mathcal F$. Moreover, $u$ is unique up to a positive affine transformation and $I$ is unique.
\end{lemma}
We are now ready to present the complete proof of our main characterization result --- Theorem \ref{th:expectiled-utility}.

\begin{proof}[Proof of Theorem \ref{th:expectiled-utility}.]
In this proof, we focus on the case of disappointment hedging, as the elation-speculating case can be treated analogously.

(ii) $\Rightarrow$ (i).
By Lemma \ref{lm-propertyEx}, $\E_{\beta}$ satisfies monotonicity, constant additivity and positive homogeneity. Combining with Lemma \ref{lm-IBGMMS} implies that $\succsim$ is invariant biseparable. To see probabilistic sophistication, let $X,Y\in \mathbb{X}$ be such that
\begin{align}\label{eq-suf1}
P\left(  \omega:X\left(  \omega\right)  \succsim x\right)  \geq P\left(
\omega:Y\left(  \omega\right)  \succsim x\right)   \text{ for all }%
x\in\mathcal{X}.
\end{align}
This implies
\begin{align*}
P\left(  u(X)\ge u(x)\right)  \geq P\left(
u(Y)\ge u(x)\right)   \text{ for all } x\in\mathcal{X},
\end{align*}
which is equivalent to
\begin{align}\label{eq-suf2}
P\left(  u(X)\ge t\right)  \geq P\left(
u(Y)\ge t\right)   \text{ for all } t\in\R.
\end{align}
By strong monotonicity of $\E_{\beta}$ in Lemma \ref{lm-propertyEx}, we have $\E_{\beta}[u(X)]\ge \E_{\beta}[u(Y)]$, and thus, $X\succsim Y$. If the inequality in \eqref{eq-suf1} is strict for some $x\in\X$, then the inequality in \eqref{eq-suf2} is strict for $t=u(x)$, and using strong monotonicity again yields $X\succ Y$. It remains to verify disappointment hedging.
Let $X,Y,W\in \mathbb{X}$ be such that $X\sim Y$ and $D_X=D_W$. It holds that
\begin{align*}
\E_{\beta}[u(X)]=\E_{\beta}[u(Y)]\quad\text{and}\quad \{\omega: u(X(\omega))<\E_{\beta}[u(X)]\}=\{\omega: u(W(\omega))<\E_{\beta}[u(W)]\}.
\end{align*}
Therefore,
\begin{align*}
\E_{\beta}\left[u\left(\frac12 W\oplus \frac12 X\right)\right]
&=\E_{\beta}\left[\frac12 u(W)+\frac12 u(X)\right]\\
&=\frac{1}{2}\E_{\beta}[u(W)+u(X)]\\
&=\frac{1}{2}(\E_{\beta}[u(W)]+\E_{\beta}[u(X)])\\
&=\frac{1}{2}(\E_{\beta}[u(W)]+\E_{\beta}[u(Y)])\\
&\le \frac{1}{2}\E_{\beta}[u(W)+u(Y)]\\
&=\E_{\beta}\left[\frac12 u(W)+\frac12 u(Y)\right]
=\E_{\beta}\left[u\left(\frac12 W\oplus \frac12 Y\right)\right],
\end{align*}
where the first and the last equalities follow from Lemma \ref{lm-midpoint}, the second and the fifth equalities are due to positive homogeneity of $\E_{\beta}$, and the third equality and the inequality come from additivity in concordant sums and superadditivity of $\E_{\beta}$ in Lemma \ref{lm-propertyEx}, respectively. Thus, we have verified disappointment hedging.

(i) $\Rightarrow$ (ii). Suppose now that $\succsim$ is probabilistically sophisticated, invariant biseparable, and satisfies disappointment hedging. Then, there exist a continuous nonconstant function $u:\X\to\R$ and a monotone,  positively homogeneous,
and constant-additive $I: B_0(\Omega,\mathcal F)\to \R$ such that,
for all acts $X,Y\in \mathbb{X}$,
\begin{align*}
X\succsim Y\Longleftrightarrow I(u\left(  X\right)  )\geq I(u\left(  Y\right)).
\end{align*}
Below, we aim to establish that $I = \E_{\beta}$ for some $\beta \ge 0$.
Note that monotonicity
and constant additivity together imply $L^\infty$-norm  continuity (see e.g.~Lemma 4.3 of F\"{o}llmer and Schied, 2016).
According to Theorem~\ref{th-characterizationEx1}, it suffices to verify that $I$ satisfies strong monotonicity and disappointment aversion.
Denote $B_0(\Omega,\mathcal F,u(\X))$ by the set of all simple random variables taking values on $u(\X)$.
One can check that probabilistic sophistication implies strong monotonicity of $I$ on $B_0(\Omega,\mathcal F,u(\X))$, that is, for $U,V\in B_0(\Omega,\mathcal F,u(\X))$
if $P(U \ge t) \ge P(V \ge t)$ for all $t \in \R$, then $I(U) \ge I(V)$, with strict inequality if the inequality is strict for some $t$. 
Moreover, for $X,Y,W\in \mathbb{X}$ with $X\sim Y$ and $D_W=D_X$,
\begin{align*}
\frac12 I(u(W)+u(X))
&=I\left(\frac{1}2(u(W)+u(X))\right)\\&=
I\left(u\left(\frac12 W\oplus\frac12 X\right)\right)\\
&\le I\left(u\left(\frac12 W\oplus\frac12 Y\right)\right)
\\& =I\left(\frac{1}2(u(W)+u(Y))\right)
=\frac12 I(u(W)+u(Y)),
\end{align*}
where the first and the last equalities follow from positive homogeneity of $I$, the second and the third equalities come from Lemma \ref{lm-midpoint}, and the inequality is due to disappointment hedging. This yields disappointment aversion of $I$ restricted to $B_0(\Omega,\mathcal F,u(\X))$. Since $u$ is continuous and nonconstant, and $\X$ is connected, we have that $u(\X)\subseteq\R$ is a nonempty interval. Hence, for any finite collection of acts $X_1,\dots,X_k\in B_0(\Omega,\mathcal F)$, there exist $a>0$ and $b\in\R$ such that $a X_i+b\in B_0(\Omega,\mathcal F,u(\mathcal X))$ for all $i\in[k]$. The strong monotonicity of $I$ can be immediately extended from $B_0(\Omega,\mathcal F,u(\mathcal X))$ to $B_0(\Omega,\mathcal F)$. 
Further, for $X,Y,W\in B_0(\Omega,\mathcal F)$, suppose $I(X)=I(Y)$, $\{\omega: X(\omega)<I(X)\}=\{\omega: W(\omega)<I(W)\}$, and $aX+b,aY+b,aW+b\in B_0(\Omega,\mathcal F,u(\mathcal X))$ for some $a>0$ and $b\in\R$. Positive homogeneity and constant additivity of $I$ imply 
\begin{align*}
I(aX+b)=I(aY+b)\quad {\rm and}\quad \{\omega:a X(\omega)+b<I(a X+b)\}=\{\omega:a W(\omega)+b<I(a W+b)\}.
\end{align*}
The disappointment aversion of $I$ restricted to $B_0(\Omega,\mathcal F,u(\X))$ yields
\begin{align*}
I((aX+b)+(aW+b))\le I((aY+b)+(aW+b)),
\end{align*}
and applying positive homogeneity and constant additivity again, we conclude $I(X+W)\le I(Y+W)$. This extends disappointment aversion of $I$ from 
$B_0(\Omega,\mathcal F,u(\X))$ to $B_0(\Omega,\mathcal F)$. 

To see the uniqueness of $u$ and $\beta\ge 0$, Lemma \ref{lm-IBGMMS} suggests that we only need to find an event $E\in \mathcal F$ such that $\E_{\beta}[1_{E}]\not\in\{0,1\}$. Indeed, if $P(E)\in(0,1)$, then
\begin{align*}
\E_{\beta}[1_{E}]=\frac{P(E)}{1+\beta(1-P(E))}\in(0,1).
\end{align*}
Hence, we complete the proof.
\end{proof}

\begin{proof}[Proof of Theorem \ref{th:elicitation}.]
Because $u$ is cardinally unique, normalize $u(x)=1$ and $u(y)=0$. 
Write $p=P(F)$. Applying \eqref{eq:GM}, formula \eqref{eq:pro} is equivalent to  
\begin{equation}
    \label{eq:pro-solve}  
\frac{1}{2} =\frac{p}{1+\beta(1- p) } .
\end{equation}
Solving \eqref{eq:pro-solve} gives $\beta=(2p-1)/(1-p)$, proving \eqref{eq:beta-elicitation} in part (i). 
Such an event $F$ exists because $P$ is nonatomic, and  the right-hand side  of \eqref{eq:pro-solve} is a continuous function of $p\in (0,1)$ with range $(0,1)$.

The existence of $E$ in part (ii) follows from the same argument in \eqref{eq:pro-solve} but now with $1/2$ replaced by $u(z)$ and the range  $[0,1]$ for $p=P(E)$. 
Finally, \eqref{eq:utility-elicitation} follows
directly from formula~\eqref{eq:GM} applied to the indifference $z\sim xEy$.  
\end{proof}

\section{Proofs  of results in  Section \ref{sec:exp-v}}

\begin{proof}[Proof of Theorem \ref{th-ax-expectileutils}.]
First, it is an immediate and well-known consequence that monotonicity and probabilistic sophistication, together with continuity of a weak order, imply  the existence of a unique certainty equivalent. To see this, for $U\in L^{\infty}(P)$, monotonicity implies that the sets $A:=\{t\in\R:t\succsim  U\}$ and $B:=\{t\in\R:U\succsim t\}$ are both intervals. Continuity further yields $A\cap B\neq\varnothing$. Using probabilistic sophistication, we know that $A\cap B$ is a singleton, and its unique element is the certainty equivalent of $\succsim$.
Second, we denote such a certainty equivalent as $I$. It is straightforward to verify that, probabilistic sophistication,
the monotonicity, continuity, and disappointment aversion 
(resp.~disappointment hedging) of $\succsim$ on $L^{\infty}(P)$ are equivalent to 
the strong monotonicity, $L^\infty$-norm continuity, and decrease in
concordant sums (resp.~mixtures) of the induced functional $I$ on $L^{\infty}(P)$. 
Moreover, under probabilistic sophistication, $L^{\infty}(P)$ can be viewed as an 
extension of $B_0(\Omega,\mathcal F)$, up to the usual identification of random 
variables that agree $P$-almost surely.
Therefore, the implications (i)$\Rightarrow$(ii) and (iii)$\Rightarrow$(ii) follow 
directly from Theorem~\ref{th-characterizationEx1} and 
Theorem~\ref{th-characterization_general}, respectively. Conversely, the two 
reverse implications follow from Lemma~\ref{lm-propertyEx}, since the properties 
established in that lemma are defined on the larger space $L^1(P)$ and hence 
apply, in particular, to $L^{\infty}(P)$.
\end{proof}

 \begin{lemma}
\label{lem:piecewise}If $f\left(  v\right)  =(u-v)^{+}-(1+\beta)(v-u)^{+}$,
with $u\in\mathbb{R}$ and $\beta>-1$, then
\[
(v-v^{\prime})\left(  f\left(  v\right)  -f\left(  v^{\prime}\right)  \right)
\leq-\min\left(  1,1+\beta\right)  (v-v^{\prime})^{2}%
\]
for all $v,v^{\prime}\in\mathbb{R}$.
\end{lemma}

\begin{proof}
Since $f(v)=(u-v)^{+}-(1+\beta)(v-u)^{+}$, we may write $f$ piecewise as
\[
f(v)=%
\begin{cases}
u-v, & v\leq u,\\[2pt]%
-(1+\beta)(v-u), & v\geq u.
\end{cases}
\]
The two branches agree at $v=u$ (both vanish), so $f$ is continuous and
piecewise linear, with slope $-1$ on $(-\infty,u)$ and slope $-(1+\beta)$ on
$(u,\infty)$.

Let
\[
m:=\min(1,1+\beta)\in(0,1],
\]
which is well defined and positive since $\beta>-1$.

Both sides of
\[
(v-v^{\prime})\big(f(v)-f(v^{\prime})\big)\ \leq\ -m(v-v^{\prime})^{2}%
\]
are invariant under exchanging $v$ and $v^{\prime}$, so it suffices to prove
the inequality for $v\geq v^{\prime}$.

\medskip\noindent\textbf{Case 1: }$v^{\prime}\leq v\leq u$\textbf{.} Here
$f(v)-f(v^{\prime})=(u-v)-(u-v^{\prime})=-(v-v^{\prime})$, so
\[
(v-v^{\prime})\big(f(v)-f(v^{\prime})\big)=-(v-v^{\prime})^{2}\leq
-m(v-v^{\prime})^{2},
\]
since $m\leq1$.

\medskip\noindent\textbf{Case 2: }$u\leq v^{\prime}\leq v$\textbf{.} Here
$f(v)-f(v^{\prime})=-(1+\beta)(v-v^{\prime})$, so
\[
(v-v^{\prime})\big(f(v)-f(v^{\prime})\big)=-(1+\beta)(v-v^{\prime})^{2}%
\leq-m(v-v^{\prime})^{2},
\]
since $m\leq1+\beta$.

\medskip\noindent\textbf{Case 3: }$v^{\prime}\leq u\leq v$\textbf{.} Let
$a:=v-u\geq0$ and $b:=u-v^{\prime}\geq0$, so that $v-v^{\prime}=a+b$. Then
\[
f(v)=-(1+\beta)a,\qquad f(v^{\prime})=b,
\]
hence
\[
(v-v^{\prime})\big(f(v)-f(v^{\prime})\big)=(a+b)\big(-(1+\beta
)a-b\big)=-(1+\beta)a^{2}-(2+\beta)ab-b^{2}.
\]
We must show this is $\leq-m(a+b)^{2}=-ma^{2}-2mab-mb^{2}$, i.e.
\begin{equation}
\big(m-(1+\beta)\big)a^{2}+\big(2m-2-\beta\big)ab+(m-1)b^{2}\leq0.
\label{eq:uine}%
\end{equation}

\begin{itemize}
\item If $\beta\geq0$, then $m=1$, and (\ref{eq:uine}) becomes
\[
-\beta a^{2}-\beta ab=-\beta\,a(a+b)\leq0,
\]
which holds since $\beta\geq0$ and $a,\,a+b\geq0$.

\item If $-1<\beta<0$, then $m=1+\beta$, and (\ref{eq:uine}) becomes
\[
\beta ab+\beta b^{2}=\beta\,b(a+b)\leq0,
\]
which holds since $\beta<0$ and $b,\,a+b\geq0$.
\end{itemize}

In either case (\ref{eq:uine}) holds, proving Case 3.

Cases 1--3 exhaust all possibilities for $v\geq v^{\prime}$, thus proving the statement.
\end{proof}

\begin{lemma}
\label{lem:freezing}Let $V$ and $U$ be bounded random variables, and let
$\mathcal{H}$ be a sub-$\sigma$-algebra of $\mathcal{F}$. Suppose that $V$ is
$\mathcal{H}$-measurable and that $U$ is independent of $\mathcal{H}$. Let
$G(v,u)$ be jointly continuous, and define
\[
g(v):=\mathbb{E}\bigl[G(v,U)\bigr].
\]
Then
\[
\mathbb{E}\bigl[G(V,U)\mid\mathcal{H}\bigr]=g(V)\qquad\text{a.s.}%
\]

\end{lemma}

\begin{proof}
Since $V$ and $U$ are bounded and $G$ is jointly continuous, the random
variable $G(V,U)$ is bounded and hence integrable.

First suppose that $V$ is simple. Then there exist $v_{1},\dots,v_{n}$ and
sets $H_{1},\dots,H_{n}\in\mathcal{H}$ such that
\[
V=\sum_{i=1}^{n}v_{i}1_{H_{i}}.
\]
Hence
\[
G(V,U)=\sum_{i=1}^{n}1_{H_{i}}G(v_{i},U).
\]
Taking conditional expectations with respect to $\mathcal{H}$, we obtain
\begin{align*}
\mathbb{E}\bigl[G(V,U)\mid\mathcal{H}\bigr]  &  =\mathbb{E}\left[  \sum
_{i=1}^{n}1_{H_{i}}G(v_{i},U)\mid\mathcal{H}\right] \\
&  =\sum_{i=1}^{n}1_{H_{i}}\mathbb{E}\bigl[G(v_{i},U)\mid\mathcal{H}\bigr].
\end{align*}
Since $U$ is independent of $\mathcal{H}$, for each fixed $v_{i}$ we have
\[
\mathbb{E}\bigl[G(v_{i},U)\mid\mathcal{H}\bigr]=\mathbb{E}\bigl[G(v_{i}%
,U)\bigr]=g(v_{i}).
\]
Therefore
\[
\mathbb{E}\bigl[G(V,U)\mid\mathcal{H}\bigr]=\sum_{i=1}^{n}1_{H_{i}}%
g(v_{i})=g(V).
\]

Now let $V$ be any bounded $\mathcal{H}$-measurable random variable. Choose
simple $\mathcal{H}$-measurable random variables $V^{(m)}$ such that
\[
V^{(m)}\to V
\]
pointwise and such that the sequence $\{V^{(m)}\}_{m\geq1}$ is uniformly bounded. Since
$U$ is bounded and $G$ is continuous, the variables $G(V^{(m)},U)$ are
uniformly bounded, and
\[
G(V^{(m)},U)\to G(V,U)
\]
pointwise. Thus, by the Dominated Convergence Theorem,
\[
G(V^{(m)},U)\to G(V,U) \qquad\text{in } L^{1}.
\]

By the Dominated Convergence Theorem, again, $g$ is continuous. Hence
\[
g(V^{(m)})\to g(V)
\]
pointwise. Since the sequence $\{g(V^{(m)})\}_{m\geq1}$ is uniformly bounded, another
application of the Dominated Convergence Theorem gives
\[
g(V^{(m)})\to g(V) \qquad\text{in } L^{1}.
\]

For each $m$, the identity already proved for simple random variables gives
\[
\mathbb{E}\bigl[G(V^{(m)},U)\mid\mathcal{H}\bigr]=g(V^{(m)}).
\]
Passing to the limit in $L^{1}$ and using the $L^{1}$-continuity of
conditional expectation, we obtain
\[
\mathbb{E}\bigl[G(V,U)\mid\mathcal{H}\bigr]=g(V).
\]
This proves the claim.
\end{proof}

\begin{proof}
[Proof of Theorem \ref{th:RPE-learning}.]
Since $U$ is bounded, there
exist 
two real numbers $\underline{u}$ and $\overline{u}$ such that
\[
\underline{u}\leq U_{t}\left(  \omega\right)  \leq\overline{u} \qquad \mbox{a.s.}%
\]
for  all $t\in\mathbb{N}$.
Without loss of generality we can assume the above holds pointwise, because  we will prove almost-sure convergence. 
Define
\[
G_{\beta}(v,u)=(u-v)^{+}-(1+\beta)(v-u)^{+}.
\]
Then the learning rule is
\[
V_{t+1}=V_{t}+\alpha_{t}G_{\beta}(V_{t},U_{t+1}).
\]
Let
\[
h_{\beta}(v):=\mathbb{E}[G_{\beta}(v,U)].
\]
That is,
\[
h_{\beta}(v)=\mathbb{E}[(U-v)^{+}]-(1+\beta)\mathbb{E}[(v-U)^{+}].
\]

The expectiled reward $\mathbb{E}_{\beta}[U]$ is, by definition, the unique solution
$v^{\star}$ of
\[
\mathbb{E}[(U-v)^{+}]=(1+\beta)\mathbb{E}[(v-U)^{+}].
\]
Equivalently, $v^{\star}$ is the unique solution  of
 $
h_{\beta}(v)=0.
 $  
We now prove that the stochastic recursion converges to this unique zero.
Since
 $
\sum_{t=0}^{\infty}\alpha_{t}^{2}<\infty,
 $ 
we have
 $ 
\alpha_{t}\rightarrow0
 $
as $t\rightarrow\infty$. Let
\[
\lambda:=\max\{1,1+\beta\},\qquad\mu:=\min\{1,1+\beta\}.
\]
Since $\beta>-1$, both $\lambda$ and $\mu$ are strictly positive. Also there
exists a finite time $t_{0}$ such that, for all $t\geq t_{0}$,
\[
0<\alpha_{t}\lambda\leq1.
\]
In particular, $1\leq\lambda$ and $1+\beta\leq\lambda$ imply $0<\alpha_{t}%
\leq1$ and $0<\left(  1+\beta\right)  \alpha_{t}\leq1$ for all $t\geq t_{0}$.

Since $V_{0}$ is deterministic and
\[
V_{t+1}=V_{t}+\alpha_{t}1_{\left\{ U_{t+1}\geq V_{t}\right\}  }(U_{t+1}%
-V_{t})+\alpha_{t}\left(  1+\beta\right)  1_{\left\{  U_{t+1}<V_{t}\right\}
}(U_{t+1}-V_{t})
\]
all of the $V_{t}$ are bounded. In particular, there exists a compact interval
$J$ of the real line such that all of the $U_{t}$ and $V_{t_{0}}$ take values
in $J$, briefly written as $U_{t},V_{t_{0}}\in J$. 
For $t\geq t_{0}$, if $U_{t+1}\geq V_{t}$, then
\[
V_{t+1}=V_{t}+\alpha_{t}(U_{t+1}-V_{t})=(1-\alpha_{t})V_{t}+\alpha_{t}%
U_{t+1}.
\]
Since $0<\alpha_{t}\leq1$, this is a convex combination of $V_{t}$ and
$U_{t+1}$. If $U_{t+1}<V_{t}$, then
\[
V_{t+1}=V_{t}+\alpha_{t}(1+\beta)(U_{t+1}-V_{t})=(1-\alpha_{t}(1+\beta
))V_{t}+\alpha_{t}(1+\beta)U_{t+1}.
\]
Since $0<\alpha_{t}(1+\beta)\leq1$, this is also a convex combination of
$V_{t}$ and $U_{t+1}$. Therefore, by induction, $V_{t}\in J\ $for all $t\geq
t_{0}$ and hence the   process $(V_{t})_{t\in \N}$ is
uniformly bounded.
Let
\[
Z_{t}:=V_{t}-v^{\star}.
\]
For fixed $u$, the function $v\mapsto G_{\beta}(v,u)$ is the one appearing in
Lemma \ref{lem:piecewise}. Therefore, for all $v,v^{\prime}\in\mathbb{R}$,
\[
(v-v^{\prime})\bigl(G_{\beta}(v,u)-G_{\beta}(v^{\prime},u)\bigr)\leq
-\mu(v-v^{\prime})^{2}.
\]
Hence%
\[
(v-v^{\prime})\bigl(G_{\beta}(v,U)-G_{\beta}(v^{\prime},U)\bigr)\leq
-\mu(v-v^{\prime})^{2},
\]
and taking expectations, we get
\[
(v-v^{\prime})\bigl(h_{\beta}(v)-h_{\beta}(v^{\prime})\bigr)\leq
-\mu(v-v^{\prime})^{2}.
\]
Since $h_{\beta}(v^{\star})=0$, setting $v^{\prime}=v^{\star}$ gives
\[
(v-v^{\star})h_{\beta}(v)\leq-\mu(v-v^{\star})^{2}.
\]
In particular, $(V_{t}-v^{\star})h_{\beta}(V_{t})\leq-\mu(V_{t}-v^{\star}%
)^{2}$ and%
\[
Z_{t}h_{\beta}(V_{t})\leq-\mu Z_{t}^{2}%
\]
for all $t\in\mathbb{N}$.

Let
\[
\mathcal{F}_{t}:=\sigma(V_{0},U_{1},\dots,U_{t}).
\]
By construction (and induction), $V_{t}$ is $\mathcal{F}_{t}$-measurable,
while $U_{t+1}$ is independent of $\mathcal{F}_{t}$. Then, by Lemma
\ref{lem:freezing},
\[
\mathbb{E}[G_{\beta}(V_{t},U_{t+1})\mid\mathcal{F}_{t}]=h_{\beta}(V_{t}).
\]
Since the processes $\left(  U_{t}\right)  _{t\in\mathbb{N}}$ and $\left(
V_{t}\right)  _{t\in\mathbb{N}}$ are uniformly bounded, and $G_{\beta}$ is
jointly continuous,   there exists a finite constant $C$ such that
\[
\mathbb{E}[G_{\beta}(V_{t},U_{t+1})^{2}\mid\mathcal{F}_{t}]\leq C\qquad
\text{for all }t\in\mathbb{N}.
\] 
Therefore, for all $t\in\mathbb{N}$,
\begin{align*}
\mathbb{E}\left[  Z_{t+1}^{2}\mid\mathcal{F}_{t}\right]   &  =\mathbb{E}%
\left[  \left(  Z_{t}+\alpha_{t}G_{\beta}(V_{t},U_{t+1})\right)  ^{2}%
\mid\mathcal{F}_{t}\right] \\
&  =Z_{t}^{2}+2\alpha_{t}Z_{t}h_{\beta}(V_{t})+\alpha_{t}^{2}\mathbb{E}%
[G_{\beta}(V_{t},U_{t+1})^{2}\mid\mathcal{F}_{t}]\\
&  \leq Z_{t}^{2}-2\mu\alpha_{t}Z_{t}^{2}+C\alpha_{t}^{2}.
\end{align*}
Since
 $
\sum_{t=1}^{\infty}\alpha_{t}^{2}<\infty,
 $ 
the   almost-supermartingale convergence theorem of Robbins and Siegmund (1971; see the remark
  below) implies that $Z_{t}^{2}$ converges almost surely and that
\[
\sum_{t=1}^{\infty}\alpha_{t}Z_{t}^{2}<\infty\qquad\text{a.s.}%
\]

Next we show that the almost-sure limit of $Z_{t}^{2}$ must be zero. Indeed,
if $Z_{t}^{2}\rightarrow L$ with $L\left(  \omega\right)  >0$ on a set $W$ of
positive measure, then, for each $\omega\in W$, there exists $t_{\omega}%
\in\mathbb{N}$ such that $Z_{t}^{2}\left(  \omega\right)  \geq L\left(
\omega\right)  /2$ for all $t\geq t_{\omega}$, and hence
\[
\sum_{t=1}^{\infty}\alpha_{t}Z_{t}^{2}\left(  \omega\right)  \geq
\sum_{t=t_{\omega}}^{\infty}\alpha_{t}Z_{t}^{2}\left(  \omega\right)  \geq
\sum_{t=t_{\omega}}^{\infty}\alpha_{t}\frac{L\left(  \omega\right)  }%
{2}=\infty,
\]
contradicting
\[
\sum_{t=1}^{\infty}\alpha_{t}Z_{t}^{2}<\infty.
\]
Thus
\[
Z_{t}^{2}\rightarrow0\qquad\text{a.s.}%
\]
and therefore $Z_{t}\rightarrow0$ a.s.~and so%
\[
V_{t}\rightarrow v^{\star}=\mathbb{E}_{\beta}[U]\qquad\text{a.s.}%
\]
as wanted.
\end{proof}

\bigskip \noindent\textbf{Remark.}
 \label{rem:RS}The
almost-sure convergence argument above is a direct application of Theorem~1 of
Robbins and Siegmund (1971). Recall its statement: if $(\mathcal{F}_{n}%
)_{n\in\mathbb{N}}$ is a filtration and $z_{n},\beta_{n},\xi_{n},\zeta_{n}$
are non-negative $\mathcal{F}_{n}$-measurable random variables satisfying
\begin{equation}
\mathbb{E}[z_{n+1}\mid\mathcal{F}_{n}]\leq z_{n}(1+\beta_{n})+\xi_{n}%
-\zeta_{n},\label{eq:SR}%
\end{equation}
then on the event $\{\sum_{n}\beta_{n}<\infty,\ \sum_{n}\xi_{n}<\infty\}$ we
have, almost surely, that $\lim_{n}z_{n}$ exists and is finite, and $\sum
_{n}\zeta_{n}<\infty$.  
In the proof above, consider%
\[
z_{t}:=Z_{t}^{2},\qquad\beta_{t}:=0,\qquad\xi_{t}:=C\alpha_{t}^{2},\qquad
\zeta_{t}:=2\mu\alpha_{t}Z_{t}^{2}.
\]
These random variables are non-negative, $\mathcal{F}_{t}$-measurable, and the
inequality%
\[
\mathbb{E}\left[  Z_{t+1}^{2}\mid\mathcal{F}_{t}\right]  \leq Z_{t}%
^{2}+C\alpha_{t}^{2}-2\mu\alpha_{t}Z_{t}^{2}
\]
that we derive, is precisely (\ref{eq:SR}) with $n=t$. Moreover, since
$\beta_{t}\equiv0$, the condition $\sum_{t}\beta_{t}<\infty$ holds on $\Omega
$, and $\sum_{t}\xi_{t}=C\sum_{t}\alpha_{t}^{2}<\infty$ holds again on
$\Omega$ by the standing assumption on the learning plasticity coefficients
$(\alpha_{t})_{t\in\mathbb{N}}$. The Robbins--Siegmund theorem therefore
guarantees simultaneously that: 
\begin{itemize}
\item $Z_{t}^{2}$ converges almost surely to a finite limit, and

\item $\sum_{t=1}^{\infty}\zeta_{t}=2\mu\sum_{t=1}^{\infty}\alpha_{t}%
Z_{t}^{2}<\infty$ almost surely, i.e. $\sum_{t=1}^{\infty}\alpha_{t}%
Z_{t}^{2}<\infty$ a.s.
\end{itemize}
  as claimed in the proof.
 
\section{Proofs  of results in Section \ref{sec:expectiled-more}
}

\begin{proof}[Proof of Theorem \ref{th:expectiled-utility-bis}.]
We give the proof for the maxmin case in detail and then   the maxmax case follows by symmetric arguments.

\underline{(ii) $\Rightarrow$ (i), maxmin case.}
Let $\beta\geq 0$ and suppose that $\succsim$ is represented by
$X\mapsto \E_\beta[u(X)]$.  
Theorems 
\ref{th:GulNeweyPowellPlus}
 and 
 \ref{th:expectiled-utility} show that $\succsim$ is a probabilistically sophisticated  maxmin expected utility preference. 
It remains to verify betweenness.  Suppose $X\sim Y$, and put
\[
m=\E_\beta[u(X)]=\E_\beta[u(Y)].
\]
Let $\Lambda $ be independent of $(X,Y)$ and write $\lambda=P(\Lambda )$.  Since
$u(X\Lambda Y)=u(X)1_\Lambda +u(Y)1_{\Lambda ^c}$, independence gives
\begin{align*}
\E\left[(u(X\Lambda Y)-m)^+\right]  
& =\lambda \E\left[(u(X)-m)^+\right]+(1-\lambda)\E\left[(u(Y)-m)^+\right] \\
& =(1+\beta)\left(\lambda \E\left[(m-u(X))^+\right]
 +(1-\lambda)\E\left[(m-u(Y))^+\right]\right) \\
& =(1+\beta)\E\left[(m-u(X\Lambda Y))^+\right].
\end{align*}
Thus $m=\E_\beta[u(X\Lambda Y)]$, and so $X\Lambda Y\sim X$.  This proves
betweenness.

\underline{(i) $\Rightarrow$ (ii), maxmin case.}
By maxmin expected utility, there are a continuous nonconstant
$u:\mathcal X\to\mathbb R$ and a nonempty set $\mathcal Q$ of probability
measures such that the preference is represented by
\[
I(u(X)),\qquad I(U):=\min_{Q\in\mathcal Q}\E^Q[U],
\]
where $I$ is defined on all simple real-valued random variables.  The functional
$I$ is normalized, monotone, constant-additive, positively homogeneous, and
superadditive.  Hence $\rho(U):=-I(U)$ is a coherent risk measure on
$B_0(\Omega,\mathcal F)$ in the sense of Ziegel (2016).

We first show that $I$, and hence $\rho$, is law invariant with respect to the
reference probability $P$.  Since $\mathcal X$ is connected and $u$ is continuous
and nonconstant, $J:=u(\mathcal X)$ is a nondegenerate interval.  Let
$U,V\in B_0(\Omega,\mathcal F)$ have the same $P$-distribution.  Choose
$a>0$ and $b\in\mathbb R$ such that $aU+b$ and $aV+b$ take values in $J$.
Then there exist simple acts $X,Y$ such that
$u(X)=aU+b$ and $u(Y)=aV+b$.  The equality in distribution of $aU+b$ and
$aV+b$ implies
\[
P(u(X)\geq u(x))=P(u(Y)\geq u(x))\qquad\text{for all }x\in\mathcal X.
\]
By probabilistic sophistication, $X\sim Y$.  Therefore
$I(aU+b)=I(aV+b)$, and by constant additivity and positive homogeneity,
$I(U)=I(V)$.  Hence $I$ is law invariant.

Next we show that the distributional functional induced by $I$ satisfies betweenness for lotteries  specified by \eqref{eq:betweenness}.  Let $U,V\in B_0(\Omega,\mathcal F)$ satisfy $I(U)=I(V)$, and let
$\lambda\in[0,1]$.  Since $U$ and $V$ are simple and $P$ is nonatomic, there
exists an event $\Lambda $ independent of $(U,V)$ such that $P(\Lambda )=\lambda$; this
is obtained by splitting each atom of the finite algebra generated by $U$ and
$V$ in the proportion $\lambda$.  Choose again $a>0$ and $b\in\mathbb R$ such
that $aU+b$ and $aV+b$ take values in $J$, and choose acts $X,Y$ with
$u(X)=aU+b$ and $u(Y)=aV+b$.  Since
$I(aU+b)=aI(U)+b=aI(V)+b=I(aV+b)$, we have $X\sim Y$.  Betweenness gives $X\Lambda Y\sim X$, hence
\[
I\left(a(U1_\Lambda +V1_{\Lambda ^c})+b\right)=I(aU+b).
\]
Using constant additivity and positive homogeneity once more,
\[
I(U1_\Lambda +V1_{\Lambda ^c})=I(U).
\]
Since $\Lambda $ is independent of $(U,V)$, the distribution of $U1_\Lambda +V1_{\Lambda ^c}$ is
$\lambda(P\circ U^{-1})+(1-\lambda)(P\circ V^{-1})$.  Thus the law-invariant
functional generated by $I$ satisfies betweenness for lotteries.

We now apply the result in Ziegel (2016, Corollary 4.6). 
Ziegel's result holds for functionals on $L^\infty$, but it can be extended to functionals defined on $B_0(\Omega,\mathcal F)$ via supnorm continuity. Note that monotonicity and constant additivity imply supnorm continuity (see, e.g.~Lemma 4.3 of F\"{o}llmer and Schied, 2016). By Ziegel (2016, Corollary 4.6),
a law-invariant coherent risk
measure satisfying betweenness for lotteries is a negative expectile.  Applied to $\rho=-I$, this gives a number $\beta\geq0$ such that
\[
\rho(U)=-\E_\beta[U]\qquad\text{for all }U\in B_0(\Omega,\mathcal F).
\]
Consequently $I(U)=\E_\beta[U]$ for all simple $U$, and therefore
\[
X\succsim Y
\iff I(u(X))\geq I(u(Y))
\iff \E_\beta[u(X)]\geq \E_\beta[u(Y)].
\]
This proves the maxmin assertion.


 {The maxmax case} is analogous.  Finally, the uniqueness statements follow from those in Theorem~\ref{th:expectiled-utility}.
\end{proof}

\begin{proof}[Proof of Theorem \ref{th:hilbert-DRO}.]
Equivalence \eqref{eq:hilbert} shows that $Q\in B_{\mathcal H}(P,\log (1+\beta))$ if and only if
$\varphi$ belongs to the set $\mathcal M_\beta$ in
Lemma~\ref{lm-propertyEx}. The asserted minimization formula is therefore
exactly the dual representation in that lemma. Moreover, some algebra and Theorem 3 of Delbaen (2013), show that the extreme points of $\mathcal M_\beta$ are precisely the Radon-Nikod\'{y}m derivatives of the Jeffrey
updates of $P$ with bias $\beta$ for $D\in \mathcal{F}$ such that  $0<P(D)<1$. For $P(D)=0,1$,  the Jeffrey
update of $P$ given $D$ is $P$ itself.
\end{proof}

\begin{proof}[Proof of Theorem \ref{th-generalex}.]
We only provide the proof of the disappointment hedging case as the elation speculating case is similar.
(ii) $\Rightarrow$ (i). The proof follows a similar argument to that of Theorem \ref{th:expectiled-utility}.
(i) $\Rightarrow$ (ii).
Let $A=u(\mathcal{X})$, which is a nondegenerate interval. Choose $a>0$ and $b$ so that $0 \in \operatorname{int}(a A+b)$, and define
$$
\bar{u}=a u+b, \quad \bar{I}(Z)=a I\left(\frac{Z-b}{a}\right)+b .
$$
Then we have
$
\bar{I}(\bar{u}(X))=a I(u(X))+b,
$
so the same preference is represented. The functional $\bar{I}$ is monotone, normalized on constants, and supnorm continuous. 
Every simple variable in $a A+b$ is realized by a simple act with utility $\bar{u}$. Probabilistic sophistication gives strong stochastic monotonicity of $\bar{I}$ on $B_0(\Omega,\mathcal F,aA+b)$. For $U,V,W\in B_0(\Omega,\mathcal F,aA+b)$, suppose $\bar{I}(U)=\bar{I}(V)$ and $D_{\bar{I}}(W)=D_{\bar{I}}(U)$. Let $X,Y,Z\in \mathbb X$ be such that $U=a u(X)+b$, $V=a u(Y)+b$ and $W=a u(Z)+b$. It holds that
\begin{align*}
I(u(X))=\frac{\bar{I}(U)-b}{a}=\frac{\bar{I}(V)-b}{a}=I(u(Y))
\end{align*}
and 
\begin{align*}
D_X&=\{\omega:X(\omega)\prec X\}=\{\omega: au(X(\omega))+b<\bar{I}(au(X)+b)\}=D_{\bar{I}}(U)\\
&=D_{\bar{I}}(W)
=\{\omega: au(Z(\omega))+b<\bar{I}(au(Z)+b)\}=\{\omega:Z(\omega)\prec Z\}=D_Z.
\end{align*}
Then we have
\begin{align*}
\bar{I}\left(\frac{U+W}{2}\right)&=aI\left(\frac{(U+W)/2-b}{a}\right)+b\\
&=a I\left(\frac{u(X)+u(Z)}{2}\right)+b\\
&=aI\left(u\left(\frac12 X \oplus\frac12 Z\right)\right)+b\\
&\le aI\left(u\left(\frac12 Y \oplus\frac12 Z\right)\right)+b=\bar{I}\left(\frac{V+W}{2}\right),
\end{align*}
where we have used the affineness under $\oplus$ in the third equality, and the inequality follows from disappointment hedging of the preference.
Hence, we prove that $\bar{I}$ satisfies decreasing
in concordant mixture on $B_0(\Omega, \mathcal{F}, a A+b)$.
Apply Theorem 14, we obtain $\bar{I}=\mathbb{E}_\beta$ on $B_0(\Omega, \mathcal{F}, a A+b)$ for some $\beta \geq 0$. For $U \in B_0(\Omega, \mathcal{F}, A)$, positive homogeneity and constant additivity of expectiles give
$$
a I(U)+b=\bar{I}(a U+b)=\mathbb{E}_\beta[a U+b]=a \mathbb{E}_\beta[U]+b .
$$
Hence $I(U)=\mathbb{E}_\beta[U]$ on $B_0(\Omega, \mathcal{F}, A)$. This completes the proof.
\end{proof}

\end{appendix}


\begin{thebibliography}{99999999999999999999999999999999}                                                                 %


\bibitem[Abdellaoui(2000)]{abdellaoui2000}{\color{black}Abdellaoui, M. (2000).
Parameter-free elicitation of utility and probability weighting functions.
\emph{Management Science}, 46(11), 1497--1512.\par}

\bibitem[Abdellaoui and Bleichrodt(2007)]{abdellaouiBleichrodt2007}{\color{black}Abdellaoui,
M. and Bleichrodt, H. (2007). Eliciting Gul's theory of disappointment aversion
by the tradeoff method. \emph{Journal of Economic Psychology}, 28(6), 631--645.\par}

\bibitem[Abdellaoui et al.~(2007)]{abdellaoui2007}Abdellaoui, M., Bleichrodt,
H., and Paraschiv, C. (2007). Loss aversion under prospect theory: A
parameter-free measurement. \emph{Management Science}, 53, 1659--1674.


\bibitem[Alon and Schmeidler (2014)]{alon2014}Alon, S. and Schmeidler, D.
(2014). Purely subjective maxmin expected utility. \emph{Journal of Economic
Theory}, 152, 382--412.

\bibitem[Augustin and T\'edongap(2021)]{augustinTedongap2021}{\color{black}Augustin,
P. and T\'edongap, R. (2021). Disappointment aversion, term structure, and
predictability puzzles in bond markets. \emph{Management Science}, 67(10),
6266--6293.\par}

\bibitem[Baillon (2012)]{baillon2012}Baillon, A., Driesen, B., and Wakker, P.P. (2012). Relative concave utility for risk and ambiguity. \emph{Games and
Economic Behavior}, 75, 481--489.



\bibitem[Beggs(2022)]{beggs2022}\textcolor{black}{Beggs, A. (2022). Reference points and learning. \emph{Journal of Mathematical Economics}, 100, 102621.}

\bibitem[Bellini et al.~(2021)]{BCCT21}Bellini, F., Cesarone, F., Colombo, C.
and Tardella, F. (2021). Risk parity with expectiles. \emph{European Journal
of Operational Research}, {291}(3), 1149--1163.

\bibitem[Bellini et al.~(2014)]{bellini2014}Bellini, F., Klar, B., M\"{u}ller,
A., and Rosazza Gianin, E. (2014). Generalized quantiles as risk measures.
\emph{Insurance: Mathematics and Economics}, 54, 41--48.

\bibitem[Bellini et al.~(2024)]{bellini2024arxiv}Bellini, F., Mao, T., Wang,
R., and Wu, Q. (2024). Disappointment concordance and duet expectiles.
arXiv:2404.17751.

\bibitem[Ben-Tal et al.~(2013)]{benTal2013}Ben-Tal, A., den Hertog, D.,
De Waegenaere, A., Melenberg, B., and Rennen, G. (2013). Robust solutions of
optimization problems affected by uncertain probabilities. \emph{Management
Science}, 59(2), 341--357.


\bibitem[Bleichrodt et al.(2007)]{bleichrodt2007utility}{\color{black}Bleichrodt,
H., Abellan-Perpi\~nan, J.M., Pinto-Prades, J.L., and Mendez-Martinez, I. (2007).
Resolving inconsistencies in utility measurement under risk: Tests of
generalizations of expected utility. \emph{Management Science}, 53(3), 469--482.\par}



\bibitem[Klibanoff et al.~(2000)]{casadesus2000}Casadesus-Masanell, R.,
Klibanoff, P., and Ozdenoren, E. (2000). Maxmin expected utility over Savage
acts with a set of priors. \emph{Journal of Economic Theory}, 92(1), 35--65.

\bibitem[Tebaldi and Wang (2022)]{castagnoli2022}Castagnoli, E., Cattelan, G.,
Maccheroni, F., Tebaldi, C., and Wang, R. (2022). Star-shaped risk measures.
\emph{Operations Research}, 70, 2637--2654.

\bibitem[Cerreia-Vioglio et al.~(2024)]{ccl2024}Cerreia-Vioglio, S.,
Corrao, R., and Lanzani, G. (2024). Dynamic opinion aggregation: long-run
stability and disagreement. \emph{Review of Economic Studies}, 91(3),
1406--1447.

\bibitem[Cerreia-Vioglio et al.~(2015)]{cdo2015}Cerreia-Vioglio, S.,
Dillenberger, D., and Ortoleva, P. (2015). Cautious expected utility and the
certainty effect. \emph{Econometrica}, 83, 693--728.

\bibitem[Cerreia-Vioglio et al.~(2020)]{cdo2020}Cerreia-Vioglio, S.,
Dillenberger, D., and Ortoleva, P. (2020). An explicit representation for
disappointment aversion and other betweenness preferences. \emph{Theoretical
Economics}, 15, 1509--1546.

\bibitem[Cerreia-Vioglio et al.~(2011)]{cerreia2011rational}
Cerreia-Vioglio, S., Ghirardato, P., Maccheroni, F., Marinacci, M., and
Siniscalchi, M. (2011). Rational preferences under ambiguity.
\emph{Economic Theory}, 48, 341--375.


\bibitem[Chandrasekhar et al.~(2022)]{chandrasekhar2022}Chandrasekhar, M.,
Frick, M., Iijima, R., and Le Yaouanq, Y. (2022). Dual-self representations of
ambiguity preferences. \emph{Econometrica}, 90, 1029--1061.

\bibitem[Chateauneuf et al.~(2026)]{chateauneuf2025}Chateauneuf, A.,
Maccheroni, F., and Zank, H. (2026). A separation of utility and beliefs
through betting consistency. \emph{Management Science}, 72(3), 1987--2005.

\bibitem[Chew (1983)]{chew1983}Chew, S.H. (1983). A generalization of the
quasilinear mean with applications to the measurement of income inequality and
decision theory resolving the Allais paradox. \emph{Econometrica}, 51(4),
1065--1092.

\bibitem[Chew (1989)]{chew1989}Chew, S.H. (1989). Axiomatic utility theories
with the betweenness property. \emph{Annals of Operations Research}, 19,
273--298.

\bibitem[Chew and Epstein (1989)]{chewEpstein1989}Chew, S.H. and Epstein,
L.G. (1989). A unifying approach to axiomatic non-expected utility theories.
\emph{Journal of Economic Theory}, 49, 207--240.

\bibitem[Cohen and Fausti (2024)]{cohenFausti2024}Cohen, S.N. and Fausti, E.
(2024). Hyperbolic contractivity and the Hilbert metric on probability
measures. arXiv:2309.02413.



\bibitem[Dabney et al.(2020)]{DabneyEtAl2020}Dabney, W.,
Kurth-Nelson, Z., Uchida, N., Starkweather, C.K., Hassabis, D., Munos, R.,
and Botvinick, M. (2020). A distributional code for value in dopamine-based
reinforcement learning. \emph{Nature}, 577, 671--675.

\bibitem[Dean and Ortoleva (2017)]{DO17}Dean, M. and Ortoleva, P. (2017).
Allais, Ellsberg, and preferences for hedging. \emph{Theoretical Economics},
12(1), 377--424.

\bibitem[de Finetti (1931)]{DF31}de Finetti, B. (1931). Sul concetto di media.
\emph{Giornale dell'Istituto Italiano degli Attuari}, 2(3), 369--396.


\bibitem[Dekel (1986)]{dekel1986}Dekel, E. (1986). An axiomatic
characterization of preferences under uncertainty: weakening the independence
axiom. \emph{Journal of Economic Theory}, 40(2), 304--318.

\bibitem[Dekel (1989)]{dekel1989}Dekel, E. (1989). Asset demands without the independence axiom. \emph{Econometrica}, 57(1), 163--169.

\bibitem[Delbaen (2013)]{delbaen2013}Delbaen, F. (2013). A remark on the
structure of expectiles. arXiv:1307.5881v1 [math.PR]. 

\bibitem[Epstein and Zin (1989)]{epsteinZin1989}\textcolor{black}{Epstein, L.G. and Zin, S.E. (1989). Substitution, risk aversion, and the temporal behavior of consumption and asset returns: A theoretical framework. \emph{Econometrica}, 57(4), 937--969.}

\bibitem[Fishburn (1982)]{fishburn1982}Fishburn, P.C. (1982).
\emph{The Foundations of Expected Utility}. D. Reidel Publishing Company,
Dordrecht.

\bibitem[F\"{o}llmer and Schied (2016)]{FS16}F\"{o}llmer, H. and Schied, A.
(2016). \emph{Stochastic Finance. An Introduction in Discrete Time}. Fourth
Edition. Walter de Gruyter, Berlin.

\bibitem[Frank et al.(2004)]{FrankSeebergerOReilly2004}Frank,
M.J., Seeberger, L.C., and O'Reilly, R.C. (2004). By carrot or by stick:
cognitive reinforcement learning in Parkinsonism. \emph{Science}, 306(5703),
1940--1943.

\bibitem[Gerfen and Surmeier(2011)]{GerfenSurmeier2011}Gerfen,
C.R. and Surmeier, D.J. (2011). Modulation of striatal projection systems by
dopamine. \emph{Annual Review of Neuroscience}, 34, 441--466.

\bibitem[Ghirardato and Marinacci (2001)]{ghirardato2001}Ghirardato, P. and
Marinacci, M. (2001). Risk, ambiguity, and the separation of utility and
beliefs. \emph{Mathematics of Operations Research}, 26(4), 864--890.


\bibitem[Ghirardato and Marinacci (2002)]{ghirardato2002}Ghirardato, P. and
Marinacci, M. (2002). Ambiguity made precise: A comparative foundation. \emph{Journal of Economic Theory}, 102, 251--289.



\bibitem[Ghirardato et al.~(2001)]{gmms2001}Ghirardato, P., Maccheroni, F.,
Marinacci, M., and Siniscalchi, M. (2001). A subjective spin on roulette
wheels (full version). \emph{Caltech Social Science Working Paper}, 1127,
August 2001. SSRN:278235.


\bibitem[Ghirardato et al.~(2003)]{gmms2003}Ghirardato, P., Maccheroni, F.,
Marinacci, M., and Siniscalchi, M. (2003). A subjective spin on roulette
wheels. \emph{Econometrica}, 71(6), 1897--1908.

\bibitem[Ghirardato et al.~(2005)]{ghirardato2005}Ghirardato, P., Maccheroni,
F., and Marinacci, M. (2005). Certainty independence and the separation of
utility and beliefs. \emph{Journal of Economic Theory}, 120, 129--136.

\bibitem[Ghirardato and Pennesi (2020)]{ghirardato2020}Ghirardato, P. and
Pennesi, D. (2020). A general theory of subjective mixtures. \emph{Journal of
Economic Theory}, 188, 105056.

\bibitem[Gilboa and Schmeidler (1989)]{gilboa1989}Gilboa, I. and Schmeidler,
D. (1989). Maxmin expected utility with non-unique prior. \emph{Journal of
Mathematical Economics}, 18, 141--153.

\bibitem[Gneiting (2011)]{gneiting2011}Gneiting, T. (2011). Making and
evaluating point forecasts. \emph{Journal of the American Statistical
Association}, 106, 746--762.

\bibitem[Gollier and Muermann(2010)]{gollierMuermann2010}{\color{black}Gollier,
C. and Muermann, A. (2010). Optimal choice and beliefs with ex ante savoring
and ex post disappointment. \emph{Management Science}, 56(8), 1272--1284.\par}

\bibitem[Gul (1991)]{gul1991}Gul, F. (1991). A theory of disappointment
aversion. \emph{Econometrica}, 59, 667--686.

\bibitem[Gul and Pesendorfer (2026)]{GP26}
Gul, F. and Pesendorfer, W. (2026). Lorenz expected utility and first-order risk aversion. Working paper, Princeton University.


\bibitem[Jeffrey (1965)]{J65}
Jeffrey, R.C. (1965). \emph{The Logic of Decision}. McGraw-Hill, New York. 

\bibitem[K\"{o}bberling and Wakker (2003)]{kobberling2003}K\"{o}bberling, V.
and Wakker, P.P. (2003). Preference foundations for nonexpected utility: A
generalized and simplified technique. \emph{Mathematics of Operations
Research}, 28, 395--423.



\bibitem[Luce and Raiffa (1957)]{luce1957}Luce, R.D. and Raiffa, H. (1957).
\emph{Games and Decisions: Introduction and Critical Survey}. Wiley, New York.



\bibitem[Maccheroni et al.~(2025)]{MMWW25}Maccheroni, F., Marinacci, M., Wang,
R. and Wu, Q. (2025). Risk aversion and insurance propensity. \emph{American
Economic Review}, {115}(5), 1597--1649.

\bibitem[Machina and Schmeidler (1992)]{machina1992}Machina, M.J. and
Schmeidler, D. (1992). A more robust definition of subjective probability.
\emph{Econometrica}, {60}(4), 745--780.

\bibitem[Milnor (1954)]{milnor1954}Milnor, J.W. (1954). Games against nature.
In:  Thrall, R.M., Coombs, C.H., and Davis, R.L.  (Eds.), \emph{Decision
Processes}, 49--60. Wiley, New York.

\bibitem[Mohajerin Esfahani and Kuhn (2018)]{esfahaniKuhn2018}Mohajerin
Esfahani, P. and Kuhn, D. (2018). Data-driven distributionally robust
optimization using the Wasserstein metric: Performance guarantees and
tractable reformulations. \emph{Mathematical Programming}, 171, 115--166.

\bibitem[Newey and Powell (1987)]{newey1987}Newey, W.K. and Powell, J.L.
(1987). Asymmetric least squares estimation and testing. \emph{Econometrica},
55, 819--847.

\bibitem[Robbins and Siegmund (1971)]{RobbinsSiegmund1971}Robbins, H. and
Siegmund, D. (1971). A convergence theorem for non negative almost
supermartingales and some applications. In: \emph{Optimizing Methods in
Statistics} (Rustagi, J.S., ed.), Academic Press, 233--257.

\bibitem[Rossi and Schleef(2026)]{rossiSchleef2026}{\color{black}Rossi, M. and
Schleef, F. (2026). Quality disclosures and disappointment: Evidence from the
Academy nominations. \emph{Management Science}, 72(11), 9274--9297.\par}

\bibitem[Savage (1954)]{S54}Savage, L. J. (1954). \emph{The Foundations of
Statistics}. Wiley.



\bibitem[Schmeidler (1989)]{S89}Schmeidler, D. (1989). Subjective probability
and expected utility without additivity. \emph{Econometrica}, 57(3), 571--587.

\bibitem[Schultz et al.(1997)]{SchultzDayanMontague1997}Schultz,
W., Dayan, P., and Montague, P.R. (1997). A neural substrate of prediction and
reward. \emph{Science}, 275(5306), 1593--1599.


\bibitem[Tom et al.~(2007)]{T07}Tom, S.M., Fox, C.R., Trepel, C. and Poldrack,
R.A. (2007). The neural basis of loss aversion in decision-making under risk.
\emph{Science}, 315(5811), 515--518.

\bibitem[Van Parys et al.(2021)]{vanParys2021}{\color{black}Van Parys, B.P.G.,
Mohajerin Esfahani, P., and Kuhn, D. (2021). From data to decisions:
Distributionally robust optimization is optimal. \emph{Management Science},
67(6), 3387--3402.\par}

\bibitem[Wakker and Deneffe(1996)]{wakkerDeneffe1996}{\color{black}Wakker, P. and
Deneffe, D. (1996). Eliciting von Neumann-Morgenstern utilities when
probabilities are distorted or unknown. \emph{Management Science}, 42(8),
1131--1150.\par}

\bibitem[Wasserman (1992)]{wasserman1992density}Wasserman, L. (1992).
Invariance properties of density ratio priors. \emph{Annals of Statistics},
20(4), 2177--2182.

\bibitem[Wasserman and Kadane (1992a)]{wassermanKadane1992bounds}Wasserman,
L. and Kadane, J.B. (1992a). Computing bounds on expectation. \emph{Journal
of the American Statistical Association}, 87(418), 516--522.

\bibitem[Wasserman and Kadane (1992b)]{wassermanKadane1992}Wasserman, L. and
Kadane, J.B. (1992b). Symmetric upper probabilities. \emph{Annals of
Statistics}, 20(4), 1720--1736.

\bibitem[Yaari (1969)]{Y69}Yaari, M.E. (1969). Some remarks on measures of
risk aversion and on their uses. \emph{Journal of Economic Theory}, 1, 315--329.

\bibitem[Ziegel (2016)]{ziegel2016}Ziegel, J.F. (2016). Coherence and
elicitability. \emph{Mathematical Finance}, 26, 901--918.
\end{thebibliography}
\end{document}